\documentclass{article} 

\usepackage{amsfonts}
\usepackage{graphicx,epstopdf}
\usepackage{booktabs}
\usepackage{subcaption}
\usepackage{rotating}
\usepackage{tabularx}
\usepackage{adjustbox}
\usepackage{multicol}
\usepackage{multirow}
\usepackage{lscape}
\usepackage{cite}
\newcommand{\package}[1]{\textbf{#1}} 

\usepackage{caption}
\usepackage{capt-of}
\usepackage{algorithm2e}
\SetKwIF{If}{ElseIf}{Else}{if}{}{else if}{else}{end if}%
\SetKwFor{While}{while}{}{end while}%
\SetKwFor{For}{for}{}{end for}%
\RestyleAlgo{ruled}
\ifpdf
  \DeclareGraphicsExtensions{.eps,.pdf,.png,.jpg}
\else
  \DeclareGraphicsExtensions{.eps}
\fi
\usepackage{hyperref}
\usepackage[dvipsnames]{xcolor}
\hypersetup{colorlinks=true, allcolors=NavyBlue}

\newtheorem{theorem}{Theorem}[section]

\newtheorem{lemma}{Lemma}[section]

\newcommand{\qed}{\hfill\rule{2.1mm}{2.1mm}}

\newcommand{\A}{\bs{A}}

\renewcommand{\L}{\mathcal{L}}

\renewcommand{\O}{\mathcal{O}}

\newcommand{\R}{\mathbf{R}}

\usepackage{amsmath}
\DeclareMathOperator*{\argmin}{\arg\min}
\DeclareMathOperator{\range}{range}

\newcommand{\diag}{\mathrm{diag}}
\newcommand{\Diag}{\mathrm{Diag}}
\newcommand{\dom}{\mathrm{dom}}

\DeclareMathOperator*{\prox}{prox}
\newcommand{\rank}{\mathrm{rank}}
\newcommand{\sign}{\mathrm{sign}}

\newcommand{\st}{\mathrm{s.t.}\;}

\renewcommand{\b}{\bs{b}}
\newcommand{\bb}{{\bs{\beta}}}
\newcommand{\bt}{{\bs{\theta}}}

\renewcommand{\d}{\bs d}

\newcommand{\e}{\bs {1}}
\newcommand{\I}{\bs{I}}
\newcommand{\p}{\bs {p}}
\newcommand{\bs}{\boldsymbol}
\renewcommand{\u}{\bs{u}}
\newcommand{\x}{\bs {x}}
\newcommand{\X}{\bs {X}}
\newcommand{\y}{\bs {y}}
\newcommand{\Y}{\bs {Y}}
\newcommand{\z}{\bs{z}}

\newcommand{\bO}{\bs{\Omega}}

\newcommand{\branchdef}[1] {\ensuremath{ \left\{\begin{array}{rl} #1 \end{array} \right. }} 
\newcommand{\rbra}[1]{\ensuremath{\left( #1 \right)}} 
\newcommand{\bra}[1]{\ensuremath{\left\{ #1 \right\}}} 

\newcommand{\ceil}[1]{\ensuremath{\left\lceil #1 \right\rceil}}

\newcommand{\ds}[1]{\displaystyle{#1}}

\newcommand{\half}{\frac{1}{2}}

\newcommand{\ra}{\rightarrow}

\usepackage[top=1in, bottom=1in, left=1in, right=1in]{geometry} 
\usepackage{authblk}

\begin{document}





\title{Proximal Methods for Sparse Optimal Scoring and Discriminant Analysis}
\author[1]{Summer Atkins} 
\author[2]{Gudmundur Einarsson} 
\author[3]{Brendan Ames}
\author[4]{Line Clemmensen} 
\affil[1]{Department of Mathematics, Louisiana State University, Baton Rouge, LA 70803-4918, srnatkins@lsu.edu}
\affil[2]{deCODE Genetics, gudmundur.einarsson2@decode.is}
\affil[3]{Department of Mathematics, University of Alabama,	Box 870350, Tuscaloosa, AL 35487-0350, bpames@ua.edu}
\affil[4]{Department of Applied Mathematics and Computer Science, Technical University of Denmark, Building 324, 2800 Kongens Lyngby, Denmark, lkhc@dtu.dk}

\maketitle

\begin{abstract}
 Linear discriminant analysis (LDA) is a classical method
		for dimensionality reduction, where
		discriminant vectors are sought to project data to a lower dimensional
		space for optimal separability of classes. Several recent papers have
		outlined strategies, based on exploiting sparsity of the discriminant vectors, for performing LDA in the high-dimensional setting where the number of features exceeds the number of observations in the data.
		However, many of these proposed methods lack scalable methods for solution
		of the underlying optimization problems.
		We consider an optimization scheme for solving
		the sparse optimal scoring formulation of LDA based
		on block coordinate descent. Each iteration of this algorithm requires an update of a scoring vector, which admits an analytic formula, and an update of the corresponding discriminant vector, which requires solution of a convex subproblem; we will propose several variants of this algorithm where the proximal gradient method
		or the alternating direction method of multipliers is used to solve this subproblem.
		We show that the per-iteration cost of these methods scales linearly in
		the
		dimension of the data provided restricted regularization
		terms are employed, and cubically in the dimension
		of the data in the worst case.
		Furthermore, we establish that  when this block coordinate descent
		framework
		generates convergent subsequences of iterates, then
		these subsequences converge to the stationary points of the
		sparse optimal scoring problem.
		We demonstrate
		the effectiveness of our new methods with empirical results
		for classification of Gaussian data and data sets drawn from
		benchmarking repositories, including time-series and multispectral X-ray data, and provide
		 \texttt{Matlab} and \texttt{R} implementations of our optimization schemes.

\end{abstract}
%

\section{Introduction}
Sparse discriminant techniques have become 
popular in the last decade due to their ability to provide increased interpretation as well as predictive performance for high-dimensional problems where few observations are present.
These approaches typically build upon successes from sparse linear regression, in particular the
LASSO and its variants (see \cite[Section 3.4.2]{elements} and \cite{hastie2015statistical}), by augmenting existing schemes for linear discriminant analysis (LDA)
with sparsity-inducing  regularization terms, such as the $\ell_1$-norm and elastic net.

Thus far, little focus has been put on the optimization strategies of these sparse discriminant methods, nor their computational cost.
We propose three novel optimization strategies to obtain discriminant directions in the high-dimensional setting where the number
of observations $n$ is much smaller than the ambient dimension $p$
or when features are highly correlated, and analyze the convergence properties of these methods.
The methods are proposed for multi-class sparse discriminant analysis using the
sparse optimal scoring formulation
with elastic net penalty
proposed in \cite{clemmensen2011sparse};  adding both the $\ell_1$- and $\ell_2$-norm penalties gives sparse solutions which, in particular, are competitive when high correlations exist in feature space due to the grouping behaviour of the $\ell_2$-norm.
The first two strategies are 
proximal gradient methods
based on modification of the (fast) iterative shrinkage algorithm \cite{beck2009fast} for linear
inverse problems.
The third method uses a variant of the alternating direction method of multipliers similar
to that proposed in \cite{ames2014alternating}.
We will formally define each of these terms and discuss them in greater detail in Section~\ref{sec:prox}.
We will see that these heuristics allow efficient classification of
high-dimensional data, which was previously
impractical using the current state of the art for sparse discriminant
analysis.
For example,  if a diagonal
or low-rank Tikhonov regularization term is used
and the number of observations is very small relative to $p$, then the per-iteration cost of each of our algorithms is $\O(p)$;
that is, the per-iteration cost of our approach scales linearly with the
number of features of our data.
Finally, we provide implementations of our algorithms in the form of the \texttt{R} package \package{accSDA}(see~\cite{accSDA})
and \texttt{Matlab} software \footnote{Available at \href{http://bpames.people.ua.edu/software}{http://bpames.people.ua.edu/software}}.

\subsection{Existing approaches for sparse LDA}


We begin with a brief overview of existing sparse discriminant analysis techniques.
Methods such as \cite{witten2011penalized,tibshirani2003class,fan2008independence} assume independence between the features in the given data. This can lead to poor performance in terms of feature selection as well as predictions, in particular when high correlations exist. Thresholding methods such as \cite{shao2011sparse}, although proven to be asymptotically optimal, ignore the existing multi-linear correlations when thresholding low correlation estimates. 
Thresholding, furthermore, does not guarantee an invertible correlation matrix, and often pseudo-inverses must be utilized.

For two-class problems,  the analysis of \cite{mai2013connection} established
an equivalence between the three methods described in \cite{wu2008sparse,clemmensen2011sparse,mai2012direct}. These three approaches are formulated as constrained versions of Fisher's discriminant problem, the optimal scoring problem, and a least squares formulation of linear discriminant analysis, respectively. For scaled regularization parameters, \cite{mai2013connection} showed that they all behave asymptotically as Bayes rules. Another two-class sparse linear discriminant method is the linear programming discriminant method proposed in \cite{cai2011direct}, which finds
an $\ell_1$-norm penalized estimate of the product between covariance matrix and difference in means.

The sparse optimal scoring (SOS) problem  was originally formulated in
\cite{clemmensen2011sparse} as a multi-class problem seeking
at most $K-1$ sparse discriminating directions, where $K$ is the number of classes present,  whereas \cite{mai2013connection} was formulated for binary problems.
This approach builds on earlier work by Hastie et.~al~\cite{hastie1994flexible, hastie1995penalized}.
Mai and Zou later proposed a multi-class sparse discriminant analysis  (MSDA) based on the Bayes rule formulation of linear discriminant analysis
in \cite{mai2015multiclass}. It imposes only the $\ell_1$-norm penalty, whereas the SOS imposes an elastic net penalty ($\ell_1$- plus $\ell_2$-norm). Adding the $\ell_2$-norm can give better predictive performance, in particular when very high correlations exist in data. MSDA, furthermore, finds all discriminative directions at once, whereas SOS finds them sequentially via deflation. A sequential solution can be an advantage if the number of classes is high, and a solution involving only a few directions (the most discriminating ones) is needed. On the other hand, if $K$ is small, finding all directions at once may be advantageous, in order to not propagate errors in a sequential manner.
Sparse optimal scoring models based on the group-LASSO are considered in~\cite{merchante2012efficient,roth2008group}, while sparse optimal scoring and sparse discriminant analysis has become a popular tool in cognitive neuroscience~\cite{grosenick2008interpretable}, among many other domains.

Finally, the zero-variance sparse discriminant analysis approach of \cite{ames2014alternating}
reformulates the sparse discriminant analysis problem as an $\ell_1$-penalized nonconvex
optimization problem in order to sequentially identify discriminative directions
in the null-space of the pooled within-class scatter matrix. Most relevant for our discussion here is the use of proximal methods to approximately solve the nonconvex optimization
problems in \cite{ames2014alternating}; we will adopt a similar approach for solving the SOS problem.

\subsection{Contributions}

The proposed optimization algorithms are not inherently novel, as they are specializations of existing algorithmic frameworks widely used in sparse regression.
However, the application of these proximal algorithms to solve the sparse optimal scoring problem is novel.
Moreover, this specialization highlights the strong relationship between the optimal scoring and regression: optimal scoring generalizes regression, in the sense that optimal scoring simultaneously fits both a linear model and a quantitative encoding of categorical class labels.

Further, the proposed sparse optimal scoring heuristics offer a substantial improvement upon the current state of the art for optimal scoring in terms of computational efficiency when the number of predictor variables is large and dense discriminant vectors are desired.
Specifically, we show that the computation required for each iteration of the proposed algorithms scales linearly with the number of predictor variables, which yields a potential significant improvement upon the cubic scaling of the least angle regression-based algorithm initially adopted in~\cite{clemmensen2011sparse}; our comparison of computational complexity can be found in~Section~\ref{sec: comp} and Appendix~\ref{sec:it-comp}.

We performed a detailed empirical analysis of our proposed heuristics for sparse optimal scoring, including comparisons of classification accuracy and computational complexity for simulated data and real-world data sets drawn from the UC Riverside Time Series Archive~\cite{keogh2006ucr,dau2019ucr}, investigation of convergence phenomena, and scaling tests. These empirical results agree with our theoretical analyses of computational complexity (Sect.~\ref{sec: comp}, App.~\ref{sec:it-comp}) and convergence (Sect.~\ref{convergence}).


\subsection{Notation}
\label{sec:notation}

Before we proceed, we first summarize the notation that is used throughout the text.
We denote the space of $p$-dimensional vectors by $\R^p$ and the space of $m$ by $n$ real matrices by $\R^{m\times n}$.
Bold capital letters, e.g., $\X$, will be used to denote matrices, lower-case bold letters, e.g., $\x$, $\bs{\beta}$, denote vectors, and unbolded letters will denote scalars (unless otherwise noted).
\newcommand{\0}{\bs{0}}
We denote the $p$-dimensional all-zeros and all-ones vectors by $\0_p$ and $\e_p$; we omit the subscript $p$ when the dimension is clear by context. We denote the transpose of a matrix $\X$ by $\X^T$ and the inverse of (nonsingular) $\X$ by $\X^{-1}$.
On the other hand, we use lower-case superscripts to indicate the indices of elements of sequences of vectors, e.g., $\{\x^i\}_{i=0}^\infty = \x^0, \x^1, \dots$ and use subscripts to denote the indices of elements of sequences of scalars, e.g., $\{\alpha_i\}_{i=0}^\infty = \alpha_0, \alpha_1, \dots.$
Subscripts will also indicate indices of entries of vectors and matrices; for example, the value in the first row and second column of matrix $\X$ will be denoted by $x_{12}$.
We will reserve the character $L$ to denote the Lipschitz constant of a given Lipschitz continuous operator $g:\R^p \mapsto \R^m$.
\renewcommand{\L}{\mathcal{L}}
Conversely, we will use the notation $\L_\mu$ to denote the augmented Lagrangian with respect to parameter $\mu$ of a given equally constrained optimization problem; we will also use $\L$ to denote the (unaugmented) Lagrangian function.
Finally, we denote the \emph{subdifferential}, i.e., the set of subgradients, of a convex function $f: \R^n \mapsto \R$ at vector $\x$ by
\[
  \partial f(\x) = \bra{ \bs{\phi} \in \R^n: f(\y) \ge f(\x) + \bs{\phi}^T(\y - \x) \text{ for all } \y \in \dom f }.
\]


\section{An Alternating Direction Method for Sparse Discriminant Analysis}
\label{sec:prox}

In this section, we describe a block coordinate descent approach for approximately solving the sparse optimal scoring problem for linear discriminant analysis.
Proposed 
in \cite{hastie1994flexible}, the optimal scoring problem recasts linear discriminant
analysis as a generalization of linear regression where both the response variable,
corresponding to an optimal labeling or scoring of the classes, and linear model parameters,
which yield the discriminant vector, are sought. Specifically, suppose that we have
the $n\times p$ data matrix $\X$, where the rows of $\X$ correspond to observations in $\R^p$
sampled from one of $K$ classes;
we assume that the data has been centered so that the sample mean is the zero vector $\bs 0 \in \R^p$.

Optimal scoring generates a sequence of discriminant
vectors and conjugate scoring vectors as follows.
Suppose
that we have identified the first $j-1$ discriminant vectors $\bb_1, \dots, \bb_{j-1} \in \R^p$ and scoring vectors $\bt_1,\dots, \bt_{j-1} \in \R^K$.
To calculate the $j$th discriminant vector $\bb_j$ and scoring vector $\bt_j$,
we solve the optimal scoring criterion problem
\begin{equation} \label{eq: os prob}
\begin{array}{rl}
\displaystyle{\argmin_{\bt \in \R^K,\,\bb \in \R^p}}
& \|\Y \bs\theta -\X \bs\beta\|^2  \\
\st & \frac{1}{n} \bs\theta^T \Y^T\Y \bs\theta = 1, \; \;\bs\theta^T \Y^T \Y  \bs\theta_\ell = 0 \; \; \forall \ell < j,
\end{array}
\end{equation}
where $\Y$ denotes the $n\times K$ indicator matrix for class membership,
defined by $y_{\ell m} = 1$ if the $\ell$th observation belongs to the $m$th class, and
$y_{\ell m} = 0$ otherwise,
and $\|\cdot\|: \R^n \ra \R$ denotes the vector $\ell_2$-norm on $\R^n$
defined by $\|\y \| = \sqrt{ y_1^2 + y_2^2 + \cdots + y_n^2}$
for all $\y \in \R^n$.
We direct the reader to \cite{hastie1994flexible} for further details regarding the derivation of \eqref{eq: os prob}.

A variant of the optimal scoring
problem called \emph{sparse discriminant analysis} or \emph{sparse optimal scoring}, which employs regularization via the elastic
net penalty function  is proposed in \cite{clemmensen2011sparse}.
As before, suppose that we have identified the first $j-1$ discriminant vectors $\bb_1, \dots,
\bb_{j-1}$ and scoring vectors
$\bt_1,\dots, \bt_{j-1}$.
We calculate the $j$th sparse discriminant vector $\bb_j$ and scoring vector $\bt_j$
as the optimal solutions of the optimal scoring criterion problem
\begin{equation} \label{eq: prob}
\begin{array}{rl}
\displaystyle{\argmin_{\bt \in \R^K,\,\bb \in \R^p}}
& \|\Y \bs\theta -\X \bs\beta\|^2 + \gamma \bs\beta^T \bs\Omega \bs\beta + \lambda \|\bs\beta\|_1 \\
\st & \frac{1}{n} \bs\theta^T \Y^T\Y \bs\theta = 1,\;\;  \bs\theta^T \Y^T \Y \bs\theta_\ell = 0 \;\; \forall \ell < j,
\end{array}
\end{equation}
where $\|\cdot\|_1: \R^p \ra \R$ denotes the vector
$\ell_1$-norm on  $\R^p$ defined by
$\|\x\|_1 = |x_1| + |x_2| + \cdots + | x_p |$
for all $\x \in \R^p$, $\Y \in \R^{n\times K}$ is again the indicator matrix for class membership,
$\lambda$ and $\gamma$ are nonnegative tuning parameters, and $\bs{\Omega}$ is a $p\times p$ positive definite matrix.
That is,~\eqref{eq: prob} is the result of adding regularization to the optimal scoring problem using a linear combination of the Tikhonov penalty term $\bs\beta^T \bs\Omega \bs\beta$ and the $\ell_1$-norm penalty $\|\bs\beta\|_1$;
we will provide further discussion regarding the choice of $\bs\Omega$ in Section~\ref{sec: comp}.
The optimization problem~\eqref{eq: prob} is nonconvex, due to the presence
of nonconvex spherical constraints. As such, we do not expect to find a globally optimal
solution of \eqref{eq: prob} using iterative methods.

\subsection{Block Coordinate Descent for Sparse Optimal Scoring}
Clemmensen et al.~propose
a
block coordinate descent method to solve~\eqref{eq: prob}
in \cite{clemmensen2011sparse}.
This approach has been widely adopted, with nearly 600 citations according to Google Scholar\footnote{Citation count  available from \href{https://scholar.google.com}{scholar.google.com}, accessed November 11, 2021.} and over 119000 downloads of the \texttt{R} implementation \texttt{sparseLDA} from the Comprehensive R Archive Network (CRAN)\footnote{Download count available from ~\href{https://cranlogs.r-pkg.org/badges/grand-total/sparseLDA}{cranlogs.r-pkg.org/badges/grand-total/sparseLDA}, accessed November 11, 2021.}.
This approach can be described as follows.
Suppose that we have an estimate $(\bt^i, \bb^i)$ of $(\bt_j, \bb_j)$ after $i$ iterations.
To update $\bt^i$, we fix $\bb = \bb^i$ and solve the optimization problem
\begin{equation} \label{eq: t prob}
\begin{array}{rl}
\bt^{i+1} = \ds {\argmin_{\bs\theta \in \R^K}} & \|\Y \bs\theta -\X {\bs\beta^i}\|^2 \\
\st & \frac{1}{n} \bs\theta^T \Y^T\Y \bs\theta = 1, \;\; \bs\theta^T \Y^T \Y \bs\theta_\ell = 0 \; \;\forall \ell < j.
\end{array}
\end{equation}
The subproblem \eqref{eq: t prob} is nonconvex in $\bt$, however,
it is known that \eqref{eq: t prob} admits an analytic solution and can be solved
exactly in polynomial time.
Indeed, we have the following lemma providing an analytic update formula for $\bt$. Note that this update requires $\O(K^3 + pn)$ floating point operations to perform the necessary matrix products.
For completeness, we provide a proof of Lemma~\ref{lem: theta update} in  Appendix~\ref{sec: theta}. See~\cite[Section 2.2]{clemmensen2011sparse} for more details. 

\begin{lemma} \label{lem: theta update}
	The problem \eqref{eq: t prob} has optimal solution
	\begin{equation}
	\bt^{i+1} = s  \rbra{\bs I - \frac{1}{n}\bs Q_j \bs Q_j^T \Y^T\Y} (\Y^T\Y)^{-1} \Y^T \X \bs\beta^i,
	\label{eq: theta formula}
	\end{equation}
	where
	$\bs Q_j$ is the $K\times j$ matrix with columns consisting of the $j-1$ scoring vectors
	$\bt_1, \dots, \bt_{j-1} $ and the all-ones vector $\e \in \R^K$, and
	$s$ is a proportionality constant ensuring that $(\bt^{i+1})^T \bs  \Y^T \Y \bt^{i+1} = n$.
  In particular, $\bt^{i+1}$ is given by
  \begin{equation}\label{eq:t-formula-2}
    \bs{w} = \rbra{\bs I - \frac{1}{n}\bs Q_j \bs Q_j^T \Y^T\Y} (\Y^T \Y)^{-1} \Y^T \X \bs\beta^i,
    \hspace{0.25in}
    \bt^{i+1} = \frac{\sqrt{n} \bs w }{\|\Y\bs w\|}.
  \end{equation}
\end{lemma}

After we have updated $\bt^{i+1}$, we obtain $\bb^{i+1}$ by solving the unconstrained optimization
problem
\begin{equation} \label{eq: b prob}
\bb^{i+1} = \ds{ \argmin_{\bs\beta \in \R^p}} \|\Y \bt^{i+1} -\X \bs\beta\|^2
+ \gamma \bs\beta^T \bs\Omega \bs\beta + \lambda \|\bs\beta\|_1.
\end{equation}
That is, we update $\bb^{i+1}$ by solving the generalized elastic net problem \eqref{eq: b prob}.
We stop this block update scheme if the relative change in consecutive iterates is smaller than a desired tolerance. Specifically, we stop the algorithm if
\[
  \max \bra{ \frac{\|\bt^{i+1} - \bt^i\|}{\|\bt^{i+1}\|}, \frac{\|\bb^{i+1} - \bb^i\|}{\|\bb^{i+1}\|} }
  \le \epsilon
\]
for stopping tolerance $\epsilon > 0$.
We delay discussion of strategies until after a discussion of the convergene properties of this alternating minimization scheme. Specifically, we discuss an algorithm based on least-angle regression (LARS) for solving~\eqref{eq: b prob} suggested in~\cite{clemmensen2011sparse} in Section~\ref{sec:LARS} and our proposed improvements in Section~\ref{sec:pms}.

\subsection{Convergence of the Block Coordinate Descent Algorithm}
\label{convergence}
\newcommand{\M}{\bs{M}}
\renewcommand{\v}{\bs{v}}

Before we proceed to the statement of our algorithms for updating $\bb$,
we investigate the convergence properties of the block coordinate descent method given by Algorithm~\ref{alg: BCD}.
Our two main results, Theorem~\ref{thm: f convergence} and Theorem~\ref{thm: convergence}, are specializations of standard results for alternating minimization algorithms; we provide proofs of these results as appendices. We should note that these two theorems establish convergence properties of Alg.~\ref{alg: BCD} that are independent of the algorithm used to solve Subproblem~\eqref{eq: b prob}. In particular, these convergence theorems hold if we use any of the iterative methods suggested in the following section for solving~\eqref{eq: b prob}, as well as the LARS Algorithm for solving~\eqref{eq: b prob} suggested in~\cite{clemmensen2011sparse}.

We first note that the Lagrangian $\L:\R^K \times \R^p \times \R \times \R^{j-1} \ra \R$
of \eqref{eq: prob} is given by
\begin{align*} \label{eq: Lag}
\L (\bt, \bb, \psi, \v)
        = & \|\Y  \bt - \X \bb\|^2 + \gamma \bb^T \bO \bb + \lambda \|\bb\|_1 \\ &+  \psi (\bt^T \Y^T \Y  \bt - n)
+ \v^T \bs U \bt,
\end{align*}
where $\bs U^T = ( \Y^T \Y \bt_1, \Y^T \Y \bt_2, \dots, \Y^T \Y \bt_{j-1})$.
Note that the Lagrangian is \emph{not} a convex function in general. However, $\L$ is the sum of the (possibly) nonconvex quadratic $\|\Y  \bt - \X \bb\|^2 + \psi (\bt^T \Y^T \Y  \bt - n) + \gamma \bb^T \bO\bb + \v^T \bs U \bt$ and the convex nonsmooth function $\lambda \|\bb\|_1$; therefore, $\L$ is subdifferentiable, with subdifferential at $(\bb, \bt)$ given by the sum of the gradient of the smooth term at $(\bb, \bt)$
and the subdifferential of the convex nonsmooth term at $(\bb, \bt)$.

We now provide our first convergence result, specifically, that Algorithm~\ref{alg: BCD}
generates a convergent sequence of function values.

\begin{theorem}
	\label{thm: f convergence}
	Suppose that the sequence of iterates $\{(\bt^i, \bb^i)\}_{i=0}^\infty$
	is generated by Algorithm~\ref{alg: BCD}.
	Then the sequence of objective function values
	$\{F(\bt^i, \bb^i)\}_{i=0}^\infty$ defined by
	$
	F(\bt, \bb) := \|\Y\bt - \X\bb\|^2 + \gamma \bb^T \bO \bb + \lambda \|\bb\|_1
	$
	is convergent.
\end{theorem}
We include a proof of Theorem~\ref{thm: f convergence} in Appendix~\ref{sec: convergence1}.

We also have the following theorem, which establishes that every convergent subsequence	of $\{(\bt^i, \bb^i)\}_{i=1}^\infty$	converges to a stationary point of \eqref{eq: prob}.

\begin{theorem} \label{thm: convergence}
	Let $\{(\bt^i, \bb^i)\}_{i=1}^\infty$ be the sequence of points generated by Algorithm~\ref{alg: BCD}.
	If $\{(\bt^{i_\ell}, \bb^{i_\ell})\}_{\ell=1}^\infty$ is a convergent subsequence of~$\{(\bt^i, \bb^i)\}_{i=1}^\infty$ with limit $ ( \bt^*, \bb^*) $
	then  $ ( \bt^*, \bb^*) $
	is a \emph{stationary point} of \eqref{eq: prob}:
	$({\bt^*}, {\bb^*})$ is feasible for \eqref{eq: prob} and
	there exists $\psi^* \in \R $ and $\v^* \in \R^{j-1}$ such that
	$
	\bs 0 \in \partial \L(\bt^*,\bb^*, \psi^*, \v^*),
	$
	where $\partial \L(\bt, \bb, \psi, \v)$ denotes the subdifferential of the Lagrangian function $\L$
	with respect to the primal variables $(\bt, \bb)$.
\end{theorem}
A proof of Theorem~\ref{thm: convergence} can be found in Appendix~\ref{sec: convergence2}.

We conclude this section by noting that we expect the block coordinate descent method to converge after exactly one full iteration in the absence of rounding error in the two-class case, i.e., $K = 2$. In this case, the optimal solution of~\eqref{eq: t prob} is given by the projection of any vector $\bs{\theta}$ not in the span of the all-ones vector $\e$ onto the set
\[ \bra{\bs \theta \in \R^k: \bs\theta^T \Y^T\Y \bs\theta = n,~\bs{\theta}^T \Y^T \Y \e = 0 }; \]
this is equivalent to the optimal solution given by Lemma~\ref{lem: theta update}. This solution is uniquely defined (up to sign) and is obtained after the initial $\bs{\theta}$ update if the initial solution $\bs{\theta}$ is not a scalar multiple of $\e$.
This suggests that Algorithm~\ref{alg: BCD} will converge after exactly one iteration if we solve~\eqref{eq: b prob} exactly. In practice, we may require multiple iterations of Algorithm~\ref{alg: BCD} depending on the relative stopping tolerances of Algorithm~\ref{alg: BCD} and the method used to solve~\eqref{eq: b prob}.

\newcommand{\Q}{\bs{Q}}
\begin{algorithm}[t]
  \KwData{Given stopping tolerance $\epsilon$, and maximum number of iterations $N$.}
  \KwResult{Discriminant vectors $(\bt^*_1, \bb^*_1)$, $(\bt^*_2, \bb^*_2)$, \dots, $(\bt^*_{K-1}, \bb^*_{K-1})$ calculated as approximate solutions of \eqref{eq: prob}.}
  \For{$j = 1,2, \dots, K-1$}{
    {Initialize $\bt_j^0$ as the projection of $K$-dimensional vector $\z$ with entries sampled uniformly at random from the interval $[0,1]$ onto the feasible region using the identity
    \[
      \bt_j^0 = \rbra{\I - \frac{1}{n}\bs Q_j \bs Q_j^T \Y^T\Y}(\Y^T \Y)^{-1} \z, \hspace{0.25in}
      \bt_j^0 = \frac{\sqrt{n}\bt_j^0}{\|\Y \bt_j^0\|};
    \]}

    {Calculate the $j$th scoring and discriminant vector pair $(\bt^*_j, \bb^*_j)$ as the limit point of the sequence $\{(\bt^i_j, \bb^i_j)\}_{i=0}^\infty$ as follows:}

    \For{$i = 0, 1,2 \dots N$}{
  		{Update $\bb^i_j$ as the solution of \eqref{eq: b prob} with $\bt=\bt^i_j$ using the solution returned by one of \eqref{alg: pg}, \eqref{alg: apg}, \eqref{alg: ADMM}, and \eqref{alg: apg bt}}\;

      {Update $\bt^{i+1}_j$ by
      \[
          \bs{w} = \rbra{\I - \frac{1}{n}\bs Q_j \bs Q_j^T \Y^T\Y}(\Y^T \Y)^{-1} \Y^T \X {\bb^i_j}, 
          \hspace{0.25in}
          \bt^{i+1} = \frac{\sqrt{n}\bs{w}}{\|\Y \bs{w}\|};\;
      \]}

      {Declare convergence if the residual between consecutive iterates is smaller than desired tolerance:}

      {
      \If{$\max\bra{ \frac{\|\bt^{i+1} - \bt^i\|}{\|\bt^{i+1}\|}, \frac{\|\bb^{i+1} - \bb^i\|}{\|\bb^{i+1}\|}} < \epsilon$}{
        The algorithm has converged\;
        \textbf{break}\;
        }
      }
    }
  }
 \caption{Block Coordinate Descent for SDA \eqref{eq: prob}}
 \label{alg: BCD}
\end{algorithm}

\subsection{The Least Angle Regression Algorithm}
\label{sec:LARS}

It is suggested in  \cite{clemmensen2011sparse} that \eqref{eq: b prob} can be solved
using the \emph{least angle regression (LARS-EN)} 
algorithm  proposed 
in~\cite{zou2005regularization} for solving elastic net regularized linear inverse problems, which in turn generalizes the LARS algorithm for $\ell_1$ regularized linear inverse problems proposed in~\cite{efron2004least}.
This method  sequentially updates elastic net estimates of the solution of~\eqref{eq: b prob}. The main computational step of the $i$th iteration is the inversion of a coefficient matrix of the form $\X_{A_i}^T \X_{A_i} + \gamma \bs{\Omega}$, where $A_i$ is the active variable set, i.e, the indices of nonzero entries of the $i$th iterate $\bs\beta^{i}$. To minimize computational costs, a low-rank dow- ndating procedure is used to update the Cholesky factorization of $\X_{A_{i-1}}^T \X_{A_{i-1}} + \gamma \bs{\Omega}$ and only nonzero-cofficients and active variable sets are stored each iteration.
The specialization of the LARS algorithm for solution of~\eqref{eq: b prob} is included in the \texttt{Matlab} and \texttt{R} package \text{sparseLDA}~\cite{SLDA}.

It is known that the number of iterations of the LARS algorithm acts as a tuning parameter for sparsity of solution (induced by the $\ell_1$ penalty in~\eqref{eq: b prob}). That is, the sequence of iterates generated by the LARS algorithm is equivalent to set of solutions of~\eqref{eq: b prob} for a particular sequence of choices of regularization parameter $\lambda$.
This yields two potential improvements over other candidate methods for solving~\eqref{eq: b prob}. First, we do not need to repeatedly solve~\eqref{eq: b prob} to tune the parameter $\lambda$, e.g., as part of a cross validation or boosting scheme. Instead, we can obtain a full regularization path with a single call to the LARS algorithm.
Second, in many applications, it is more natural to seek a solution with a specific maximum cardinality, rather than some desired value of $\ell_1$ penalty; the LARS algorithm allows the use of this more natural interpretation of the regularization process as a stopping criterion. We direct the reader to \cite[Section~3.5]{zou2005regularization} for further details.

This approach carries a computational cost on the order of
$\O(mnp + m^3)$, where $m$ is the desired number of nonzero coefficients at termination,
which is prohibitively expensive if both $p$ and $m$ are large.
For example, if $m = cp$ for some constant $c \in (0,1)$, then the
per-iteration cost scales cubically with $p$.
Moreover, we should note that we solve a different instance of~\eqref{eq: b prob} during each iteration of Algorithm~\ref{alg: BCD}, since the value of $\bt$ changes each iteration.
Thus, we are required to compute the full regularization path during each iteration of Algorithm~\ref{alg: BCD}.
We should also note that the solutions given by LARS-based algorithms are not directly comparable to those given by proximal gradient methods since they do not correspond to the same optimization problem. Generally, the deflationary process for calculating discriminant-scoring vector pairs leads to different optimization problems for calculating $(\bt_k, \bb_k)$ for $k > 1$, due to differences in results of solving the prior subproblems between heuristics. This phenomena is especially pronounced when using LARS-based algorithms, compared to our proposed proximal methods, since the LARS-based algorithms implicitly use dynamically updated regularization parameters $\lambda_i$, which inherently depend on the choice of discriminant and scoring vectors chosen in earlier steps of~Alg.~\ref{alg: BCD}.

Finally, we note that coordinate descent methods have been widely adopted for calculation of elastic net regularized generalized linear models; see~\cite{friedman2010regularization} for further details.
However, we are unaware of any application of coordinate descent methods for solution of the elastic net regularized optimal scoring problem~\eqref{eq: prob}.

\section{Proximal Methods for Updating $\bs{\beta}$}
\label{sec:pms}

Our primary contribution is a collection of algorithms for solving the $\bb$-update subproblem
\eqref{eq: b prob}. Specifically, we specialize three classical algorithms, each based on the evaluation of
proximal operators, to obtain novel numerical methods for solution of the sparse optimal scoring problem.
We will see that these algorithms require significantly fewer computational resources than
least angle regression if we exploit structure in the regularization parameter $\bO$.

\renewcommand{\d}{\bs{d}}
\subsection{The Proximal Gradient Algorithm}
Given a convex function $f:\R^p \ra \R$, the \emph{proximal operator} $\prox_f: \R^p \ra \R^p$
of $f$ is defined by
\[
\prox_f(\y) = \argmin_{\x\in\R^p} \bra{ f(\x) + \frac{1}{2} \|\x - \y \|^2},
\]
which yields a point that balances the competing objectives of being near $\y$
while simultaneously minimizing $f$.
The use of proximal operators is a classical technique in optimization, particularly as surrogates
for gradient descent steps for minimization of nonsmooth functions.
For example, consider the optimization problem
\begin{equation} \label{eq: prox prob}
\min_{\x \in \R^p} f(\x) + g(\x),
\end{equation}
where $f: \R^p \ra \R$ is differentiable and $g: \R^p \ra \R$ is potentially nonsmooth.
That is, \eqref{eq: prox prob} minimizes an objective that can be decomposed
as the sum of a differentiable function $f$ and nonsmooth function $g$.
To solve \eqref{eq: prox prob}, the \emph{proximal gradient method} performs iterations
consisting of a step in the direction of the negative gradient $-\nabla f$ of the smooth part $f$
followed by evaluation of the proximal operator of $g$:
given iterate $\x^i$, we obtain the updated iterate $\x^{i+1}$ by
\begin{equation}\label{eq: pgm}
\bs x^{i+1} = \prox_{\alpha_i g}( \bs x^i - \alpha_i \nabla f(\bs x^i) ),
\end{equation}
where $\alpha_i$ is a step length parameter.
If both $f$ and $g$ are differentiable and the step size $\alpha_i$ is small, then this approach reduces to the
classical gradient descent iteration: $\x^{i+1} \approx \x^i - \alpha_i \nabla f(\x^i) - \alpha_i \nabla g(\x^i) $.
We direct the reader to \cite{beck2009fast,parikh2014proximal}
for more details regarding the proximal gradient method.

Expanding the residual norm term $\|\Y\bt - \X\bb\|^2$ in the objective of \eqref{eq: b prob}
and dropping the constant term shows that \eqref{eq: b prob} is
equivalent to minimizing
\begin{equation} \label{eq: prox fxn}
F(\bs\beta) = \frac{1}{2} \bb^T \A \bb + \d^T \bb + \lambda \|\bb\|_1,
\end{equation}
where $\A = 2 (\X^T \X + \gamma \bs\Omega)$ and $\d = -2 \X^T\Y \bt^{i+1}$.
We can decompose $F$
as $F (\bb) = f (\bb) + g(\bb)$,
where $f(\bs\beta) = 	\half \bs\beta^T \bs A \bs\beta + \d^T\bs\beta$ and $g(\bs\beta) = \lambda \|\bs\beta\|_1$.
Note that $F$ is strongly convex if the penalty matrix $\bs\Omega$ is positive definite;
in this case \eqref{eq: prox fxn} has a unique minimizer.
Note further that $f$ is differentiable with
$
\nabla f(\bb) = \A \bs\beta + \d.
$
Moreover, the proximal operator of the $\ell_1$-norm term $g(\bb) = \lambda \|\bb\|_1$ is given by
\[
\prox_{\lambda \|\cdot\|_1}(\y) = \sign(\y) \max\{ |\y | - \lambda \e, \bs 0\} =: S_\lambda(\y);
\]
see \cite[Section~6.5.2]{parikh2014proximal}.
The proximal operator $S_\lambda = \prox_{\lambda \|\cdot\|_1}$ is often called the \emph{soft thresholding operator}
(with respect to the threshold $\lambda$) and
$\sign:\R^p \ra \R^p$ and $\max: \R^p \times \R^p \ra \R^p$ are the element-wise sign and maximum mappings
defined by
\[
[\sign(\y)]_i = \sign(y_i) = \branchdef{ +1, &\mbox{if } y_i > 0 \\ 0, & \mbox{if } y_i=0 \\ -1, &\mbox{if } y_i < 0 }
\]
and
$
[\max(\x,\y)]_i = \max(x_i, y_i).
$
Using this decomposition, we can apply the proximal gradient method to generate
a sequence of iterates $\{\bb^i\}$ by
\begin{equation} \label{eq: SDAP update}
\bs\beta^{i+1} =  \sign(\p^i) \max \{ |\p^i| - \lambda \alpha_i \e, \bs 0 \},
\end{equation}
where
\begin{equation} \label{eq: SDAP p}
\p^i = \bs\beta^i - \alpha_i \nabla f(\bs\beta^i)
= \bs\beta^i - \alpha_i (\bs A\bs\beta^i + \d);
\end{equation}
here, $\e$ and $\bs 0$ denote the all-ones and all-zeros vectors in $\R^p$, respectively.
This proximal gradient algorithm with constant step lengths is summarized in Algorithm~\ref{alg: pg}; Algorithm~\ref{alg: pgbt} can be modified to obtain a variant of Algorithm~\ref{alg: pg} that employs a backtracking line search.
It is important to note that this update scheme is virtually identical to
the \emph{iterative soft thresholding algorithm} (ISTA) proposed in~\cite{beck2009fast} for solving $\ell_1$ regularized linear inverse problems.
Specifically, our problem and update formula differs only from that typically associated with ISTA
in the presence of the Tikhonov 
regression term $\bb^T \bs \Omega \bb$ in our model. 
As an immediate consequence, 
we can specialize the known convergence properties of ISTA (\cite[Theorem 3.1]{beck2009fast}) to show that
the sequence of function values $\{F(\bb^i)\}$  generated by Algorithm~\ref{alg: pg} 
converges sublinearly to the
optimal function value of  \eqref{eq: prox fxn}  at a rate no worse than $\O(1/i)$ for a certain  choice of step lengths $\{\alpha_i\}$.

\begin{algorithm}[t]
	\KwData{Given initial iterate $\bb^0$, sequence of step sizes $\{\alpha_i\}_{i=0}^\infty$, stopping tolerance $\epsilon$, and maximum number of iterations $N$.}
	\KwResult{Solution $\bb^*$ of \eqref{eq: b prob}.}

	\For{$i = 0, 1,2 \dots N$}{
	  Update gradient term by \eqref{eq: SDAP p}:
		\[
			\p^i = \bs\beta^i - \alpha_i (\bs A\bs\beta^i + \d);
		\]

	  Update iterate using proximal gradient step \eqref{eq: SDAP update}:
		\[
			\bs\beta^{i+1} =  \sign(\p^i) \max \{ |\p^i| - \lambda \alpha_i \e, \bs 0 \};
		\]

    Terminate if current iterate is approximately stationary:

    {\If{$\| \A \bb^{i+1} + \d + \lambda~\sign(\bb^{i+1}) \|_\infty < p \epsilon$:}
    {
      The algorithm has converged\;
      \textbf{break}\;
      }
    }
	}
  \caption{Proximal gradient method for solving~\eqref{eq: b prob}}
  \label{alg: pg}
\end{algorithm}

\begin{algorithm}[t]
	\KwData{Given initial iterate $\bb^0$, sequence of step sizes $\{\alpha_i\}_{i=0}^\infty$, stopping tolerance $\epsilon$, and maximum number of iterations $N$.}
	\KwResult{Solution $\bb^*$ of \eqref{eq: b prob}.}

	\For{$i = 0, 1,2 \dots N$}{

	  Update iterate using proximal gradient step \eqref{eq: SDAP update} and backtracking line search:
		
		\For{$k=0,1,2 \dots$ \emph{until step size accepted}}{
			Update step length:
			\[
			\bar L = \eta^{k} L_{i}, \hspace{0.25in}
			\alpha = \frac{1}{\bar L};
			\]

			Update iterate using proximal gradient step \eqref{eq: SDAP update}:
			\begin{align*}
			\p^i & = \bs\beta^i - \alpha_i (\bs A\bs\beta^i + \d); \\
			\bs\beta^{i+1} &=  \sign(\p^{i+1}) \max \{ |\p^{i+1}| - \lambda \alpha \e, \bs 0 \};
			\end{align*}

			Determine whether to accept update or increment step length:
			\smallskip
			
			\If{$(\bb^{i+1} - \y^{i+1})^T \rbra{ \frac{\bar L}{2} \I - \A} (\bb^{i+1} - \y^{i+1}) \ge 0$} {

        Accept update: set $L_{i+1} = \bar{L}$ and $\alpha_i = \bar\alpha$\;

				\textbf{break}\;
			}
		}

    Terminate if current iterate is approximately stationary:

    {\If{$\| \A \bb^{i+1} + \d + \lambda~\sign(\bb^{i+1}) \|_\infty < p \epsilon$:}
    {
      The algorithm has converged\;
      \textbf{break}\;
      }
    }
	}
  \caption{Proximal gradient method for solving~\eqref{eq: b prob} with back tracking line search}
  \label{alg: pgbt}
\end{algorithm}

\begin{theorem}[{Specialization of \cite[Theorem 3.1]{beck2009fast}}] \label{thm : prox convergence}
	Let $\{\bb^i\}$ be generated by Algorithm~\ref{alg: pg} with  initial iterate
	$\bb^0$ and constant step size
	$\alpha_i = \alpha \in (0, 1/\|\A\|)$ or step sizes chosen using the backtracking scheme given by Algorithm~\ref{alg: pgbt},
	where $\|\A\| = \lambda_{\max}(\A)$ denotes the largest eigenvalue
	of the positive semidefinite matrix $\A$.
	Suppose that $\bb^*$ is	a minimizer of $F$. Then there exists a constant $c$ such that
	\begin{equation} \label{eq: prox conv}
	F(\bb^i) - F(\bb^*) \le \frac{ c\|\A\|\|\bb^0 - \bb^*\|^2}{i}
	\end{equation}
	for any $i \ge 1$.
\end{theorem}

Note that Theorem~\ref{thm : prox convergence} implies that Algorithm~\ref{alg: pg} and Algorithm~\ref{alg: pgbt} yield $\epsilon$-suboptimal solutions for~\eqref{eq: b prob} within $\O(1/\epsilon)$ iterations.
Here, we say that a vector $\tilde\bb$ is $\epsilon$-suboptimal for~\eqref{eq: b prob} if $F(\tilde\bb) - F(\bb^*) \le \epsilon$, for each $\epsilon > 0$.
It is known that ISTA converges \emph{linearly} when the objective function $F$ is strongly convex (see~\cite[Chapter~10]{beck2017first}). We will see that the strong convexity of the objective of~\eqref{eq: prob} depends on the structure of the regularization term $\bs{\Omega}$. When $\bs{\Omega}$ is full rank, then $F$ is strongly convex.
Therefore,
the sequence of iterates generated by Algorithm~\ref{alg: pg}
converges to the unique minimizer of \eqref{eq: b prob} if the penalty parameter
$\bs\Omega$ is chosen to be positive definite.
If we choose $\bs\Omega$ to be
positive semidefinite but not full rank, then $F$ may not be strongly convex. In this case, Theorem~\ref{thm : prox convergence} establishes that the sequence of iterates generated by Algorithm~\ref{alg: pg} converges sublinearly to the minimum value of \eqref{eq: b prob} and any limit point of this sequence is a minimizer of \eqref{eq: b prob}.
We will see that using such a matrix may have attractive computational advantages despite
this loss of uniqueness.

It is reasonably easy to see that the quadratic term of $F$ in \eqref{eq: prox fxn}  is differentiable
and has Lipschitz continuous gradient with constant $L= \|\A\|$;
this is the significance of the $\|\A\|$ term in~\eqref{eq: prox conv}.
In order to ensure convergence in our proximal gradient method, we need
to estimate $\|\A\|$ to choose a sufficiently small step size $\alpha$.
Computing this Lipschitz constant may be prohibitively expensive for large $p$;
one can typically calculate $\|\A\|$ to arbitrary precision using variants of the Power Method (see \cite[Sections~7.3.1, 8.2]{gv}) at a cost of $\O(p^2 \log p)$ floating point operations.
Instead, we could use an upper bound $\tilde L \ge L$ to compute our constant step size
$\alpha = 1/\tilde L \le 1/L$.
For example, when $\bs \Omega$ is a diagonal matrix, we estimate $\|\bs A\|$ by
\begin{align*}
\|\bs A \|  =  2\| \gamma \bO + \X^T \X \|
\le 2 \gamma \|\diag (\bO) \|_\infty + 2\|\X\|^2_F 
\approx \frac{1}{\alpha},
\end{align*}
where $\diag(\bs M) \in \R^p$ is the vector of diagonal entries of the matrix $\bs M \in \R^{p\times p}$.
Here, we used the triangle inequality and the identity $\|\X^T\X\| \le \|\X\|^2_F $, where
$\|\X\|_F$ denotes the Frobenius norm of $\X$ defined by
$\|\X\|_F^2 = 
{\sum_{i=1}^n \sum_{j=1}^p x_{ij}^2}.$
The Frobenius norm and, hence,
this estimate of $\|\bs A\|$ can be computed using only $\O(np)$ floating point operations.

In practice, we stop~Algorithm~\ref{alg: pg} after a maximum number of iterations are performed or a sufficiently suboptimal solution is identified. Recall, that $\bb^*$ minimizes $F(\bb)$ if $\0 \in \partial F(\bb)$.
On the other hand, we know that $\A \bb^* + \d + \lambda~\sign(\bb^*) \in \partial F(\bb^*)$ by the structure of the subgradients of $\|\bb\|_1$. This implies that we can terminate our proximal gradient update scheme if we find $\bb^*$ such that $\A \bb^* + \d + \lambda~\sign(\bb^*)$ is close to $\0$. Specifically, we stop the iterative scheme after the $i$th iteration if
\[
  \|\A \bb^i + \d +  \lambda~\sign(\bb^i) \|_\infty = \max_j |(\A \bb^i + \d+ \lambda~\sign(\bb^i))_j| \le p \epsilon
\]
for given stopping tolerance $\epsilon > 0$.

\subsection{The Accelerated Proximal Gradient Method}
\label{sec: APG}
The similarity of our method to iterative soft thresholding and, more generally,
our use of proximal gradient steps to mimic the  gradient method for minimization
of our nonsmooth objective suggests that we may be able to use momentum terms to
accelerate convergence of our iterates. In particular, we modify the fast iterative
soft thresholding algorithm (FISTA) described in  \cite[Section~4]{beck2009fast}  to solve our
subproblem. This approach extends a variety of accelerated gradient descent methods,
most notably those of Nesterov~
\cite{nesterov1983method,nesterov2005smooth,nesterov2013gradient},
to minimization of composite convex functions;
for further details regarding the acceleration process and motivation for why such acceleration is possible, we direct the reader to the references
\cite{allen2014linear,bubeck2015geometric,flammarion2015averaging,lessard2016analysis,o2015adaptive,su2014differential,tseng2008accelerated}.

We accelerate convergence of our iterates by taking a proximal gradient step from an
extrapolation of the last two iterates. Applied to \eqref{eq: prox prob}, the
accelerated proximal gradient method features updates of the form
\begin{align}
\y^{i+1} &= \x^i + \omega_i (\x^i - \x^{i-1})  \label{eq: APG y}\\
\x^{i+1} &= \prox_{\alpha g} (\y^{i+1} - \alpha \nabla f(\y^{i+1} ) ) \label{eq: APG x},
\end{align}
where $\omega_i \in [0,1)$ is an extrapolation parameter, a standard choice of this parameter is $i/(i+3)$.
Applying this modification to our original proximal gradient algorithm
yields Algorithm~\ref{alg: apg}.
Modifying the backtracking line search of Algorithm~\ref{alg: pgbt} to use the accelerated proximal gradient update yields Algorithm~\ref{alg: apg bt}.
It can be shown that the
sequence of iterates generated by either of these algorithms converges in value to the optimal solution of \eqref{eq: b prob}
at rate
$\O(1/i^2)$.

\begin{algorithm}[t]

\KwData{Given initial iterates $\bb^0 = \bb^1$, step length $\alpha$, sequence of extrapolation parameters $\{\omega_i\}_{i=0}^\infty$, stopping tolerance $\epsilon$, and maximum number of iterations $N$.}

\KwResult{Solution $\bb^*$ of \eqref{eq: b prob}.}

\For{$i = 1,2 \dots, N$}{
	Update momentum term by \eqref{eq: APG y}:
		\[
		\y^{i+1} = \bb^i + \omega_i (\bb^i - \bb^{t-1});
		\]
	Update gradient term by \eqref{eq: SDAP p}:
	\[
	\p^{i+1} = \bs\y^{i+1} - \alpha (\bs A\bs\y^{i+1} + \d);
	\]
	Update iterate using proximal gradient step \eqref{eq: SDAP update}:
	\[
	\bs\beta^{i+1} =  \sign(\p^{i+1}) \max \{ |\p^{i+1}| - \lambda \alpha \e, \bs 0 \};
	\]

  Terminate if current iterate is approximately stationary:

  {\If{$\| \A \bb^{i+1} + \d + \lambda~\sign(\bb^{i+1}) \|_\infty < p \epsilon$:}
  {
    The algorithm has converged\;
    \textbf{break}\;
    }
  }
}
\caption{Accelerated proximal gradient method for solving \eqref{eq: b prob} with constant step size}
\label{alg: apg}
\end{algorithm}
\begin{algorithm}[t!]

	\KwData{Initial iterates $\bb^0 = \bb^{1}$, initial Lipschitz constant $L_0 >0$, scaling parameter $\eta > 1$, sequence of
		extrapolation parameters $\{\omega_i\}_{i=0}^\infty$,
    stopping tolerance $\epsilon$, and maximum number of iterations $N$.}

	\KwResult{Solution $\bb^*$ of \eqref{eq: b prob}.}

	\For{$i = 0, 1,2 \dots N$}{
  Update $\bb^{i+1}$ using accelerated proximal gradient step and backtracking line search:

		\For{$k=0,1,2 \dots$ \emph{until step size accepted}}{
			Update step length:
			\[
			\bar L = \eta^{k} L_{i}, \hspace{0.25in}
			\alpha = \frac{1}{\bar L};
			\]

			Update iterate using accelerated proximal gradient step \eqref{eq: APG y}, \eqref{eq: APG x}:
			\begin{align*}
			\y^{i+1} &= \bb^i + \omega_i (\bb^i - \bb^{t-1}) \\
			\p^{i+1} &= \bs\y^{i+1} - \alpha (\bs A\bs\y^{i+1} + \d) \\
			\bs\beta^{i+1} &=  \sign(\p^{i+1}) \max \{ |\p^{i+1}| - \lambda \alpha \e, \bs 0 \};
			\end{align*}

			Determine whether to accept update or increment step length:
			\smallskip

      \If{$(\bb^{i+1} - \y^{i+1})^T \rbra{ \frac{\bar L}{2} \I - \A} (\bb^{i+1} - \y^{i+1}) \ge 0$} {

        Accept update: set $L_{i+1} = \bar{L}$ and $\alpha_i = \bar\alpha$\;

				\textbf{break}\;
			}
    }

    Terminate if current iterate is approximately stationary:

    {\If{$\| \A \bb^{i+1} + \d + \lambda~\sign(\bb^{i+1}) \|_\infty < p \epsilon$:}
    {
      The algorithm has converged\;
      \textbf{break}\;
      }
    }
	}
\caption{Accelerated proximal gradient method for solving~\eqref{eq: b prob} with backtracking line search}
\label{alg: apg bt}
\end{algorithm}

\begin{theorem}[{Specialization of \cite[Theorem~4.4]{beck2009fast}}] \label{thm: APG conv}
	Let $\{\bb^i\}$ be generated by Algorithm~\ref{alg: apg} and constant step size
	$\alpha_i = \alpha \in (0, 1/\|\A\|) $ or generated using backtracking line search by Algorithm~\ref{alg: apg bt}
  with  initial iterate
	$\bb^0$.
	Then there exists constant $c>0$ such that
	\begin{equation} \label{eq: APG conv}
	F(\bb^i) - F(\bb^*) \le { \frac{c\|\A\|\|\bb^0 - \bb^*\|^2}{i^2} }
	\end{equation}
	for any $i \ge 1$ and minimizer $\bb^*$ of $F$.
\end{theorem}

\subsection{The Alternating Direction Method of Multipliers}
We conclude by proposing a third algorithm for
minimization of \eqref{eq: prox fxn} based on the \emph{alternating direction method of
	multipliers} (ADMM).
The ADMM is designed to minimize separable objective functions under linear coupling constraints, i.e., problems of the form
\begin{equation} \label{eq: ADMM prob}
\min_{\x \in \R^p, y \in \R^m} \Big \{ f(\x) + g(\y) : \bs A\x + \bs B \y = \bs c  \Big\},
\end{equation}
via an approximate dual gradient ascent,
where $f: \R^p \ra \R$, $g: \R^m \ra \R$, $\bs A\in\R^{r\times p}, \bs B \in \R^{r\times m},$
and $\bs c \in \R^r$;
we direct the reader to the survey \cite{boyd2011distributed}
for more details regarding the ADMM.

The minimization of the composite function $F$ defined in \eqref{eq: prox fxn} can
be written as the unconstrained optimization problem
\begin{equation} \label{eq:min b}
\min_{\bb\in \R^p} F(\bb) = \min_{\bb\in\R^p} \half \bb^T \A \bb + \d^T \bb + \lambda \|\bb\|_1.
\end{equation}
We can rewrite \eqref{eq:min b} in an equivalent form appropriate for the ADMM by splitting
the decision variable $\bb \in \R^p$ as two new variables $\x, \y \in \R^p$ with an accompanying linear
coupling constraint
$\x = \y$. Following this change of variables, we can express \eqref{eq:min b} as
\begin{equation} \label{eq: split prob}
	\begin{array}{rl}
\displaystyle{ \min_{\x,\y \in \R^p}} &  \half \x^T \bs  A \x + \d^T \x   + \lambda \|\y\|_1 \\
\st & \x - \y = \bs 0.
\end{array}
\end{equation}
The ADMM generates a sequence of iterates using approximate dual gradient ascent
as follows.
The augmented Lagrangian of \eqref{eq: split prob} is defined by
\begin{align*}
\L_\mu(\x,\y,\z) = & \; \half \x^T\bs  A \x + \d^T \x + \lambda \|\y\|_1  + \z^T (\x-\y) + \frac{\mu}{2}\|\x-\y\|^2
\end{align*}
for all $\x, \y , \z \in \R^p$; here, $\mu > 0$ is a penalty parameter
controlling the emphasis on enforcing feasibility of the primal
iterates $\x$ and $\y$.
To approximate the gradient of the dual functional of \eqref{eq: split prob}, we
alternately minimize the augmented Lagrangian with respect to $\x$ and $\y$.
We then update the dual variable $\z$ using a dual ascent step using this approximate gradient.

Suppose that we have the iterates ($\x^i, \y^i,\z^i)$ after $i$ iterations.
To update $\x$, we take
\begin{align*}
\x^{i+1} &= \argmin_{\x\in \R^p} \L_\mu(\x, \y^i, \z^i )
= \argmin_{{\x\in \R^p}} \frac{1}{2} \x^T (\mu \bs I + \bs A) \x - \x^T (-\d + \mu \y^{i} - \z^i).
\end{align*}
Applying the first order necessary and sufficient conditions for optimality, we see
that ${\x^{i+1}}$ must satisfy
\begin{equation} \label{eq: x update}
(\mu \bs I + \bs A) \x^{i+1} = -\d + \mu \y^i - \z^i.
\end{equation}
Thus, ${\x^{i+1}}$ is obtained as the solution of a linear system. Note that the coefficient
matrix $\mu \bs I + \bs A$ is independent of the iteration number $i$; we take the Cholesky decomposition of
$\mu \bs I + \bs A = \bs B\bs B^T$ during a preprocessing step and obtain ${\x^{i+1}}$ by solving
the two triangular systems given by
\renewcommand{\M}{\bs{M}}
$$
\bs B\bs B^T \x^{i+1} = - \d + \mu \y^i - \z^i.
$$
When the generalized elastic net matrix $\bO$ is diagonal, or $\M:= \mu \bs I + 2 \gamma \bO$ is otherwise
easy to invert, we can invoke the Sherman-Morrison-Woodbury formula (see \cite[Section~2.1.4]{gv}) to
solve this linear system more efficiently; more details will be provided
in 
Section~\ref{sec: comp}.
In particular, we see that
\begin{align*}
(\mu \bs I +2  \gamma \bO + 2 \X^T \X)^{-1}
         =\bs M^{-1} - 2\bs M^{-1} \X^T (\bs I + 2\X \bs M^{-1} \X^T)^{-1} \X \bs M^{-1};
\end{align*}
computing this inverse only requires computing the inverse of $\bs M$
and the inverse of the $n\times n$ matrix $\bs I + 2\X \bs M^{-1} \X^T$.

Next $\y$ is updated by
\begin{align*}
\y^{i+1} &=  \argmin_{y \in \R^p} \L_\mu(\x^{i+1}, \y, \z^i) 
         = \argmin_{\y} \lambda \|\y\|_1 + \frac{\mu}{2}\|\y - \x^{i+1} - \z^i/\mu\|^2.
\end{align*}
That is, $\y^{i+1}$ is updated as the value of the soft thresholding operator of the $\ell_1$-norm at ${\z^{i}/\mu + \x^{i+1}}$:
\begin{align}
\y^{i+1} &= S_{\lambda/\mu}(\x^{i+1} + \z^i/\mu).\label{eq: y update}
\end{align}
Finally, the dual variable $\z$ is updated using the approximate dual ascent step
\begin{equation} \label{eq: z update}
\z^{i+1} = \z^i + \mu (\x^{i+1} - \y^{i+1}).
\end{equation}
Following each iteration, we check that the Karush-Kuhn-Tucker conditions for~\eqref{eq: ADMM prob} have been approximately satisfied as a stopping criterion.
Specifically, we check if the updated iterates $(\x^{i+1}, \y^{i+1}, \z^{i+1})$  have satisfied primal and dual feasibility within relative tolerance of $\epsilon$ by checking if the inequalities
\begin{align*}
  \|\x^{i+1} - \y^{i+1}\| &\le \epsilon~\max\{ \|\x^{i+1}\|, \|\y^{i+1}\| \} \\
  \mu \|\y^{i+1} - \y^i \| & \le \epsilon~\|\y^{i+1}\|,
\end{align*}
respectively, are satisfied.
This approach is summarized in Algorithm~\ref{alg: ADMM}.

It is well-known that the ADMM generates a sequence of iterates which converge
linearly to an optimal solution of \eqref{eq: ADMM prob} under certain strong convexity
assumptions on $f$ and $g$ and rank assumptions on $\bs A$ and $\bs B$,
all of which are satisfied by our problem \eqref{eq: split prob} when
$\bO$ is positive definite
(see, for example, \cite{deng2012global, goldfarb2013fast, nishihara2015general, he20121}).
It follows that the sequence of iterates $\{\x^i, \y^i, \z^i\}$ generated by Algorithm~\ref{alg: ADMM}
converges to a minimizer of $F(\bb)$;
that is, $\x^i - \y^i \ra \bs0$ and $F(\x^i), F(\y^i)$ converge linearly to the minimum value of $F$.
The following theorem gives a worst-case convergence rate for~Algorithm~\ref{alg: ADMM}.

\begin{theorem}[{Specialization of \cite[Theorem~4.1]{he20121}}]\label{thm:ADMMconvergence}
  Suppose $(\x^i,\y^i,\z^i)$ is generated by~Algorithm~\ref{alg: ADMM}.
  Suppose further that the objective function of~\eqref{eq: ADMM prob}, $F(\x,\y) = \frac{1}{2} \x^T \A \x + \d^T \x + \lambda \|\y\|_1$, is strongly convex.
  Then the sequence of iterates satisfies $\x^i - \y^i \ra \0$ and
  \begin{equation} \label{eq:ADMMrate}
    F(\x^i,\y^i) - F(\x^*, \y^*) \le \frac{C \|\A\| \|\x^0 - \x^*\|^2}{i}
  \end{equation}
  for some constant $C > 0$, where $(\x^*,\y^*)$ is a minimizer of~\eqref{eq: ADMM prob}.
\end{theorem}

If $F$ is not strongly convex, then we should expect Algorithm~\eqref{alg: ADMM} to generate a sequence of iterates that converges sublinearly in value to the optimal value of~\eqref{eq: b prob}.

\begin{algorithm}[t]
	\KwData{Given initial iterates $\x^0 = \y^0$, step length $\mu$, stopping tolerance $\epsilon$, and maximum number of iterations $N$.}

	\KwResult{Solution $\bb^* = \x^*=\y^*$ of \eqref{eq: b prob}.}

	\For{$i = 0, 1,2 \dots, N$}{
		Update $\x$ by \eqref{eq: x update}:
		\[
		(\mu \bs I + \bs A) \x^{i+1} = -\d + \mu \y^i - \z^i;
		\]

    Update $\y$ using soft thresholding \eqref{eq: y update}:
		\[
			\y^{i+1} = S_{\lambda/\mu}(\x^{i+1} + \z^i/\mu);
		\]

    Update $\z$ using approximate dual ascent \eqref{eq: z update}:
		\begin{align*}
		\z^{i+1} = \z^i + \mu (\x^{i+1} - \y^{i+1});
		\end{align*}

    Test stopping criterion:

    {\If{$\|\x^{i+1} - \y^{i+1}\| \le \epsilon \max\{\|\x^{i+1}\|, \|\y^{i+1}\|\}$ and $\mu \|\y^{i+1} - \y^i \| \le \epsilon \|\y^{i+1}\|$}
    {
      The algorithm has converged\;
      \textbf{break}\;
      }
    }

  }

  \caption{Alternating direction method of multipliers for solving~\eqref{eq: split prob}}
	\label{alg: ADMM}
\end{algorithm}

\subsection{Computational Requirements}
\label{sec: comp}

To motivate the use of our proposed proximal methods for the minimization of \eqref{eq: b prob},
we briefly sketch the computational costs of each of our methods.
We will see that for  certain choices of regularization parameters, the number of floating
point operations needed scales linearly with the size of the data.

We begin by with the computational costs of the proximal gradient method (Alg.~\ref{alg: pg} and~\ref{alg: pgbt}). The computational requirements for the accelerated proximal gradient (Alg.~\ref{alg: apg} and~\ref{alg: apg bt}) and alternating direction method of multipliers (Alg.~\ref{alg: ADMM}) can be calculated in a similar fashion; a full calculation of the complexity of each iteration of each method can be found in~Appendix~\ref{sec:it-comp}.
The most expensive operation of each iteration of~Alg.~\ref{alg: pg} is the calculation of the gradient: $\nabla f(\bb) = \A \bb$. This requires $\O(np)$ floating point operations (flops) if the regularization matrix $\bO$ is a diagonal matrix (and $\O(p^2)$ flops for unstructured $\bO$). On the other hand, Theorem~\ref{thm : prox convergence} implies that Alg.~\ref{alg: pg} and~\ref{alg: pgbt}
generate an $\epsilon$-suboptimal solution of~\eqref{eq: b prob} within $\O(1/\epsilon)$ iterations. Putting everything together, we see that Alg.~\ref{alg: pg} and~\ref{alg: pgbt} yield $\epsilon$-suboptimal solutions using at most $\O(np/\epsilon)$ flops if $\bO$ is a diagonal matrix.
Performing similar calculations for Alg.~\ref{alg: apg},~\ref{alg: apg bt}, and~\ref{alg: ADMM} yields the total complexity estimates for each method, found in Table~\ref{t:total-flops}.
We remind the reader that the flop counts given in Table~\ref{t:total-flops} are restricted to the case when $\bs{\Omega}$ is diagonal or easy to invert;
all methods scale cubically in $p$ if $\bs\Omega$ is unstructured and/or difficult to invert.

\begin{table}[t]
  \centering
  \begin{tabular}{lllll} \toprule
    Method  & PG & APG & ADMM & LARS \\ \midrule
    Diagonal $\bs\Omega$ & $\O(np/\epsilon)$ & $\O(np/\sqrt{\epsilon})$ & $\O(n^2 p / \epsilon)$ & $\O(mnp + m^3)$ \\
    Dense $\bs\Omega$ & $\O(p^2/\epsilon)$ & $\O(p^2/\sqrt{\epsilon})$ & $\O(p^3 + p^2/\epsilon)$ & $\O(mnp + m^3)$ \\\bottomrule
  \end{tabular}
  \caption{Upper bounds on total number of floating point operations required to calculate an $\epsilon$-suboptimal solution (PG, APG, ADMM) or solution containing $m$ nonzero entries (LARS) using of~\eqref{eq: b prob} diagonal regularization matrix $\bO$ and dense, unstructured $\bO$.}
  \label{t:total-flops}
\end{table}

The complexity estimates found in Table~\ref{t:total-flops} establish that the accelerated proximal gradient method is more efficient for solution of~\eqref{eq: b prob} than least angle regression if
\begin{equation}\label{eq:APG-best}
  \frac{np}{\sqrt{\epsilon}}<< m n p + m^3,
\end{equation}
where $m$ is the number of nonzero entries of the optimal solution $\bb^*$ when $\bO$ is a diagonal matrix. In particular, if $m$ is moderately large, e.g., $m = c(np)^{1/3}$ for sufficiently large $c$, APG is significantly more efficient for solution of~\eqref{eq: b prob} than LARS. In practice, this improvement in computational complexity is large when $p$ is large (e.g., $p > 1000$). 

In the case that $n$ is large (or both $n$ and $p$ are large), Table~\ref{t:total-flops} suggests that ADMM should be prohibitively expensive, relative to the other methods considered.
We should note that our implementation of the ADMM for solving~\eqref{eq: b prob} is optimized for the case when $n$ is much smaller than $p$. In particular, our use of the Sherman-Morrison-Woodbury formula to update $\x$ in~Algorithm~\ref{alg: ADMM} explicitly relies on the assumption that $n < p$.

Section~\ref{sims} provides a detailed empirical analysis of the computational complexity of these algorithms for solving~\eqref{eq: b prob}.



\section{Numerical Analysis}
\label{sims}
We next compare the performance of our proposed approaches with standard methods for
penalized discriminant analysis
in several numerical experiments.
In particular, we compare the
implementations of the block coordinate descent method
Algorithm~\ref{alg: BCD}, where each discriminant direction $\bb$
is updated using either the proximal gradient method with constant step size, Algorithm~\ref{alg: pg} (PG), the proximal gradient method with backtracking line search, Algorithm~\ref{alg: pgbt} (PGB),
the accelerated proximal method with constant step size, Algorithm~\ref{alg: apg} (APG),
the accelerated proximal method with backtracking, Algorithm~\ref{alg: apg bt} (APGB),
and the alternating direction method of multipliers, Algorithm~\ref{alg: ADMM},
(ADMM),
with the least angle regression based algorithm (LARS) for solving the sparse optimal scoring problem
proposed 
in 	\cite{clemmensen2011sparse}.

\subsection{Gaussian data}
\label{sec:Gauss}

We first perform simulations investigating efficacy of our heuristics
for
classification of Gaussian data.
In each experiment, we generate data consisting of $p$-dimensional vectors from one of $K$ multivariate normal distributions.
Specifically, we obtain training observations corresponding to the $i$th class, $i=1,2,\dots, K$, by sampling 25 observations from the multivariate normal distribution with
mean $\bs{\mu_i} \in \R^p$ satisfying
\begin{equation}\label{eq:mu-def1}
    [\bs{\mu_i}]_j =
    \begin{cases}
      0.7, & \text{if } 100(i-1) < j \le 100 i \\
      0, & \text{otherwise},
    \end{cases}
\end{equation}
for all $j=1,2,\dots, p$,
and covariance matrix $\bs\Sigma \in \R^{p\times p}$
chosen so that all features are correlated with
$\Sigma_{ij} = r$ for all $i \neq j$ and $\Sigma_{ii} = 1$
	for all $i$. We conduct the experiment for all
	$ K \in \{2,4\}, r\in \{0,0.1,0.5,0.9\}$.
For each experiment, we sample $250$ testing observations from each class in the same manner as the training data.
We set $p = 1500$ in each simulation.
For each $(K, r)$ pair, we generate $20$ data sets and use nearest centroid
classification following projection onto the span of the discriminant directions
calculated using Algorithm~\ref{alg: BCD} and PG, PGB, APG, APGB, ADMM, or LARS to solve~\eqref{eq: b prob}, or SZVD.

The sparse discriminant analysis heuristics require training of the regularization
parameters $\gamma$, $\bs{\Omega}$,  and $\lambda$. In all experiments,
we set $\gamma = 10^{-3}$ and $\bs\Omega$ to be the
$p\times p$ identity matrix $\bs\Omega = \bs I$.
We train the remaining parameter
$\lambda$ using $5$ fold cross validation. Specifically, we choose $\lambda$ from a set
of potential $\lambda$ of the form $ \bar\lambda/2^{c}$ for $c= 3,2,\dots, 0, -1$,
and $\bar\lambda$ chosen so that the problem has nontrivial solution for all considered
$\lambda$.
Note that \eqref{eq: prox fxn} has optimal solution given by
$\bb^* = \bs A^{-1}\d$ if we set $\lambda = 0$; this implies that choosing
\begin{equation}\label{eq:barlam}
	\bar\lambda = \frac{(\bb^*)^T \d - \frac{1}{2} (\bb^*)^T \bs A \bb^*}{\|\bb^*\|_1}
\end{equation}
ensures that there exists at least one solution $\bb ^*$ with value strictly less than zero.
We choose the value of  $\lambda$ with fewest average number
of misclassification errors over training-validation splits
amongst all $\lambda$ which yield discriminant
vectors containing at most $25\%$ nonzero entries; in the event of a tie, we select the value of $\lambda$ which yields discriminant vectors with smallest average cardinality among all with minimum validation score.
Note that we could have applied other resampling methods, e.g., boot strapping, to train the regularization parameter using the same criteria with minimal changes to the experiment.
To reduce computation time, we use less conservative stopping criteria during the cross validation procedure.
We terminate each proximal algorithm in the inner loop after $1000$ iterations
or a $10^{-4}$ suboptimal solution is obtained.
We perform exactly one iteration of the outer block coordinate
descent loop if $K = 2$ and the outer loop is stopped after a maximum number of $250$ iterations
or a $10^{-3}$ suboptimal solution has been found if $K =4$.
After the optimal value of $\lambda$ is determined using cross validation, we train using the full training set; 
we terminate the algorithm after $1000$ iterations or a $10^{-5}$ suboptimal solution is obtained.
The augmented Lagrangian parameter $\mu = 2$ was used for  the ADMM method in all experiments.
We use the value of $\bar L = 0.25$ for the initial estimate of the Lipschitz constant and $\eta = 1.25$ for the scalar multiplier in the backtracking line search.

The LARS algorithm for updating $\bb$ was terminated after the convergence criterion is met with tolerance $10^{-3}$ or an iterate with more than $0.25 p$ nonzero entries was found. As with the other methods, we perform either one outer iteration ($K=2$) or at most $250$ outer iterations, stopping when a $10^{-3}$ suboptimal solution is found $(K=4)$.
Since LARS generates the full regularization path for each value of $\bt$, we do not need to perform cross validation to tune the parameter $\lambda$.
Instead, we choose the value of $\bb$ with minimum classification error on the \emph{training} set as our optimal discriminant value at termination.

These stopping criteria and parameter choices were chosen empirically in order to yield accurate classifiers using a minimal number of iterations; in particular, using modest stopping tolerances limits the number of iterations performed, which tends to limit overfitting in practice.

We also include the  Sparse Zero-Variance Discriminant Analysis (SZVD)  method  
proposed in \cite{ames2014alternating} in our comparisons.
We train the regularization parameter $\gamma$ in SZVD in a fashion similar to that above.
We set the maximum value of the regularization parameter $\gamma$	to
\begin{equation}\label{eq:bargam}
  \bar{\gamma} = \frac{\hat{\bb}^T \bs{B} \hat\bb}{ \|\hat\bb\|_1},
\end{equation}
where
$\hat \bb$ is the optimal solution of the unpenalized SZVD problem
and $\bs{B}$ is the  sample between-class
covariance matrix of the training data.
We choose $\gamma$ from the exponentially spaced grid  $\bar{\gamma}/2^c$ for $c = 3,2,1, 0, -1$
using $5$ fold cross-validation; this approach is consistent with that in \cite{ames2014alternating}.
We select the value of $\gamma$ which minimizes misclassification error amongst
all sets of discriminant vectors with at most $35\%$ nonzero entries;
this acceptable sparsity threshold is chosen to be higher than that in the SOS experiments,
due to the tendency of SZVD to misconverge to the trivial all-zero solution for large values of
$\gamma$.
We stop SZVD after a maximum of 1000 iterations or a solution satisfying the stopping
tolerance of $10^{-5}$ is obtained.
We use the augmented Lagrangian penalty parameter $\beta = 1.25$ in SZVD in all experiments.
All simulations and experiments in the following sections were performed using \texttt{Matlab 2019b} on a standard node of the high performance computing system at the Alabama Supercomputer Center. The \texttt{Matlab} implementation of Algorithm~\ref{alg: BCD}  can be obtained from \href{https://github.com/gumeo/accSDA_matlab}{https://github.com/gumeo/accSDA\_matlab}.
We use the Matlab package~\texttt{sparseLDA}~\cite{SLDA} with modification to return the full optimization path to solve~\eqref{eq: b prob} using the LARS algorithm.

\begin{figure}
\centering

  \begin{subfigure}[b]{0.45\textwidth}
      \includegraphics[width=\textwidth]{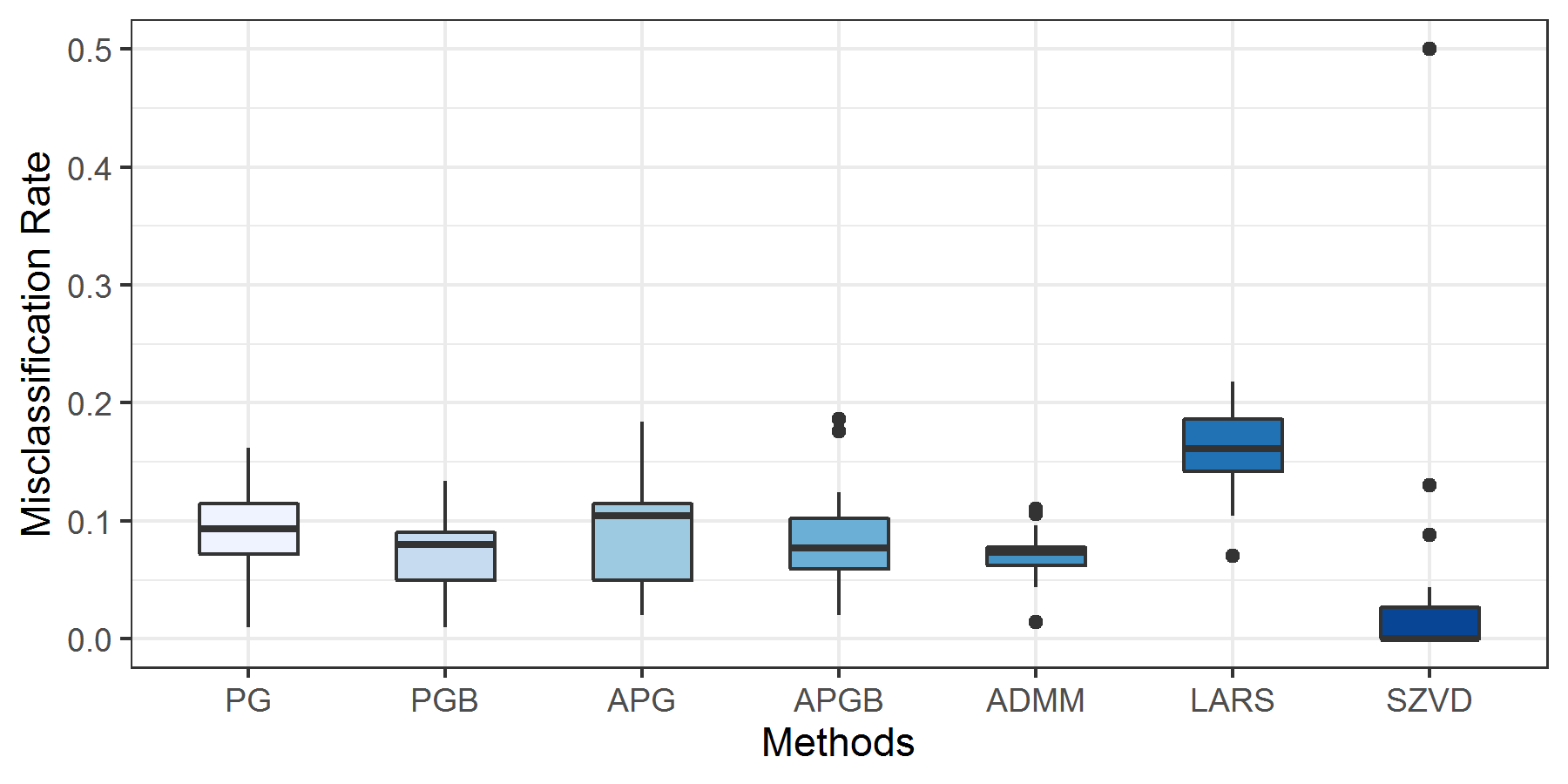}
      \caption{$K=2, r=0$}
      \label{fig:(2,0))}
  \end{subfigure}
  ~
  \begin{subfigure}[b]{0.45\textwidth}
      \includegraphics[width=\textwidth]{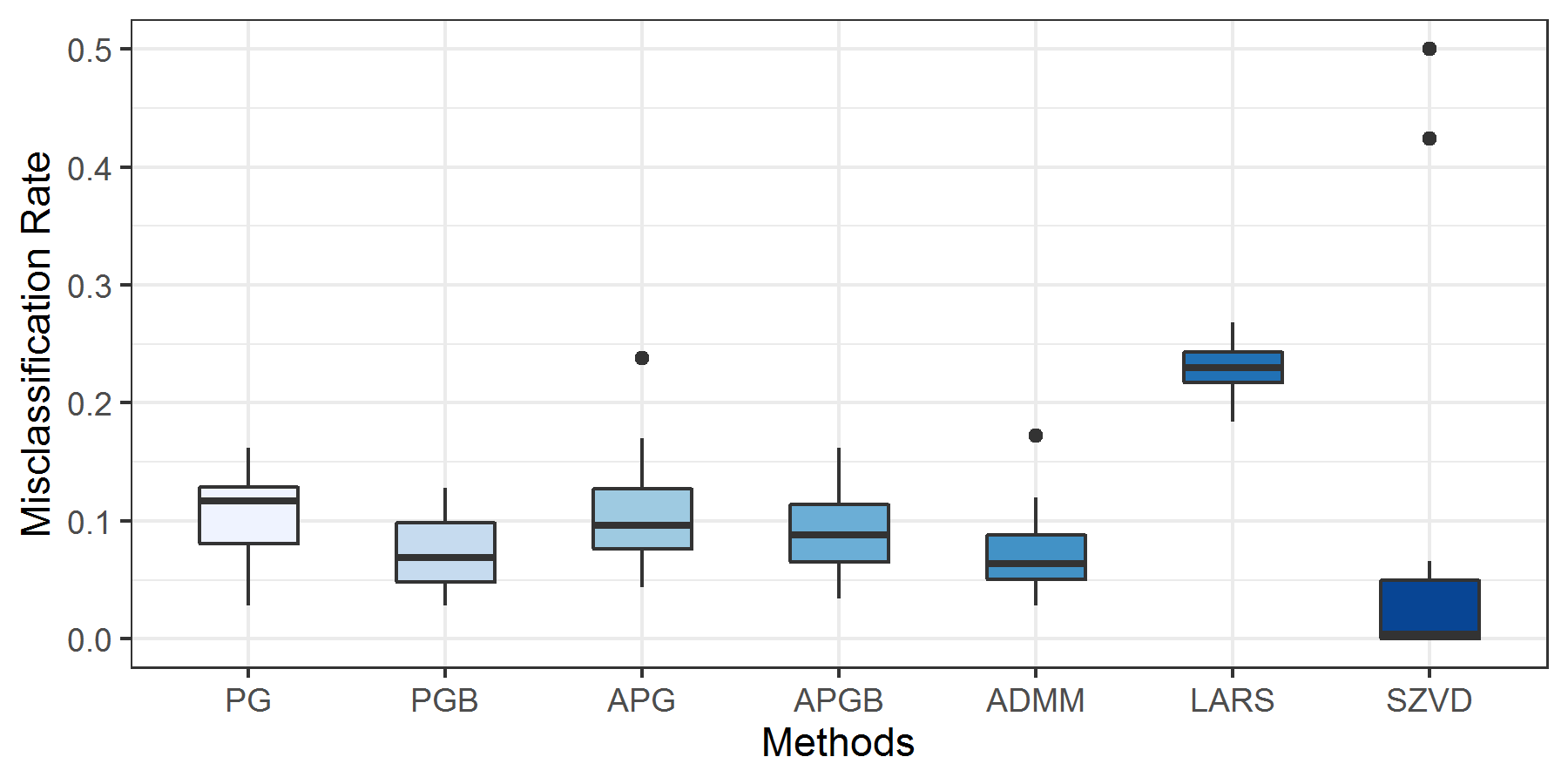}
      \caption{$K=2, r=0.1$}
      \label{fig:(2,1))}
  \end{subfigure}

  \begin{subfigure}[b]{0.45\textwidth}
      \includegraphics[width=\textwidth]{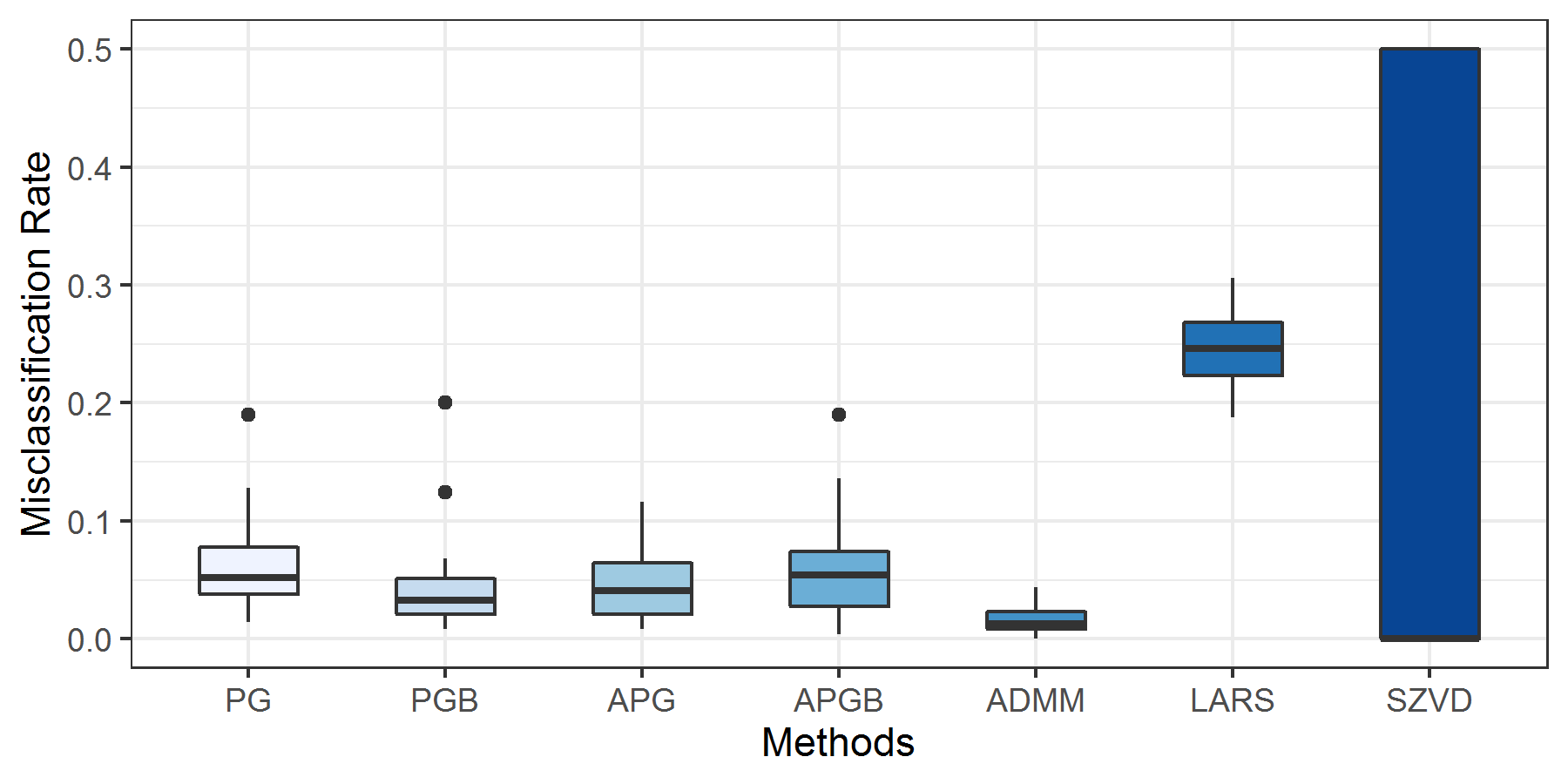}
      \caption{$K=2, r=0.5$}
      \label{fig:(2,5))}
  \end{subfigure}
  ~
  \begin{subfigure}[b]{0.45\textwidth}
      \includegraphics[width=\textwidth]{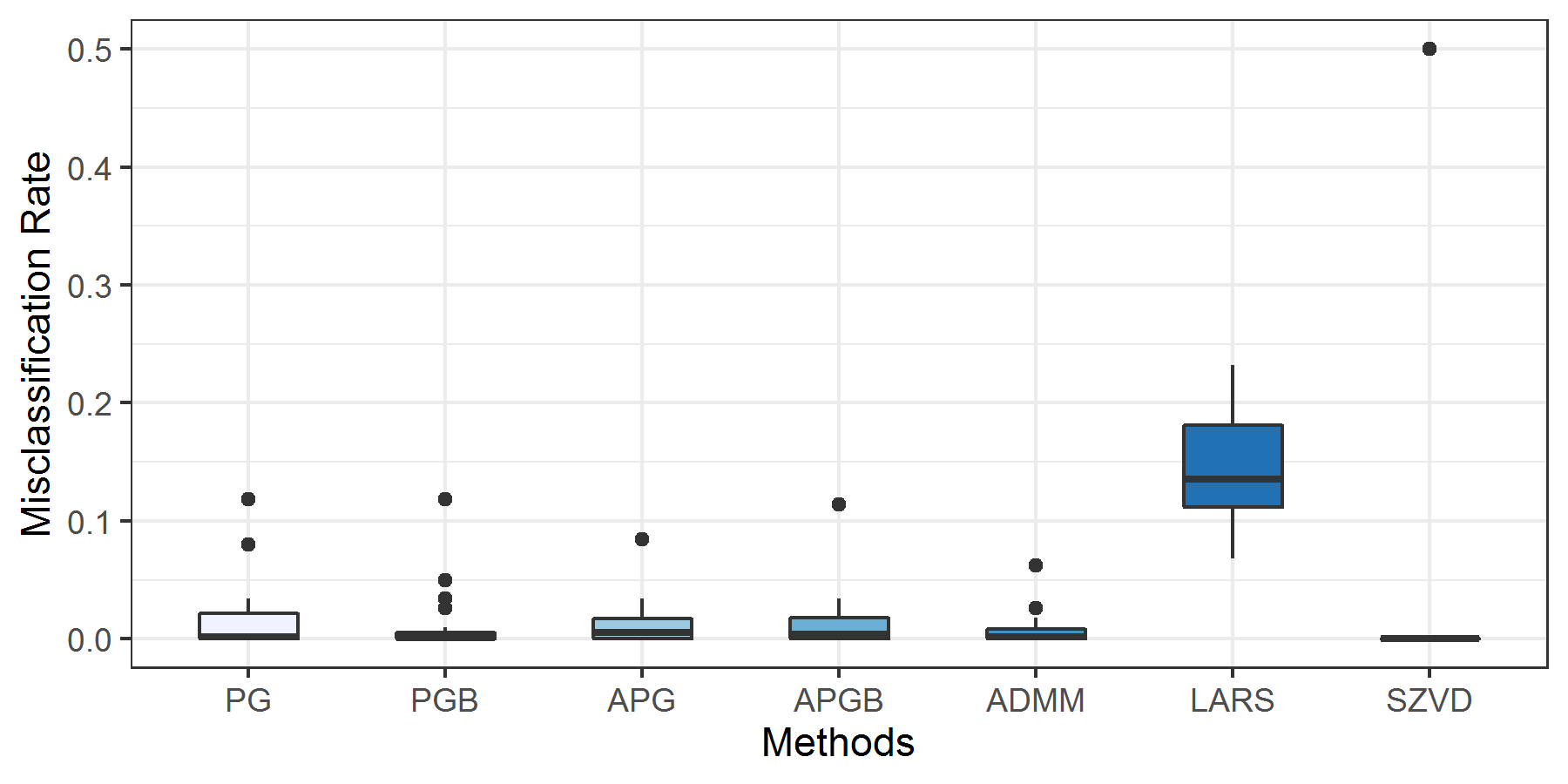}
      \caption{$K=2, r=0.9$}
      \label{fig:(2,9))}
  \end{subfigure}
  \caption{Box plot of out-of-sample misclassification rate for
  $2$-class Gaussian data with class means defined by~\eqref{eq:mu-def1} and covariance vector $\bs{\Sigma}$ for given values of $r$. For each data set, we use nearest centroid classification following projection onto discriminant vectors given by the sparse zero variance method (SZVD) or optimal scoring vectors calculated using the  proximal gradient method with constant stepsize (PG), with backtracking line search (PGB), accelerated proximal method with constant stepsize (APG) and backtracking line search (APGB), alternating direction method of multipliers (ADMM), and least angle regression (LARS). In all experiments, $n_{train} = 50$ and $n_{test} = 500$. }
  \label{fig:gauss=err}

\end{figure}

\begin{figure}
  \centering
  \begin{subfigure}[b]{0.45\textwidth}
      \includegraphics[width=\textwidth]{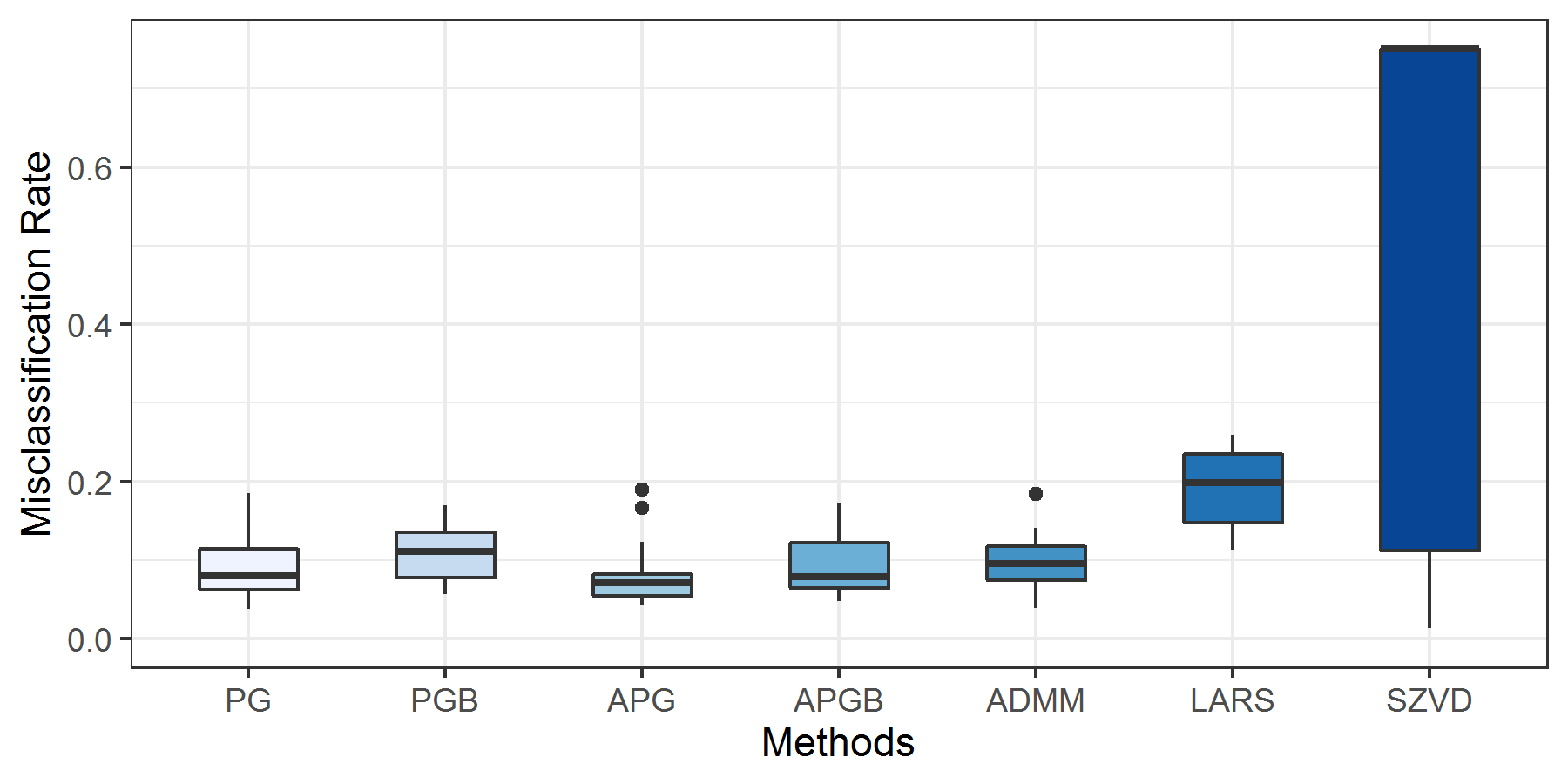}
      \caption{$K=4, r=0$}
      \label{fig:(4,0))}
  \end{subfigure}
  ~
  \begin{subfigure}[b]{0.45\textwidth}
      \includegraphics[width=\textwidth]{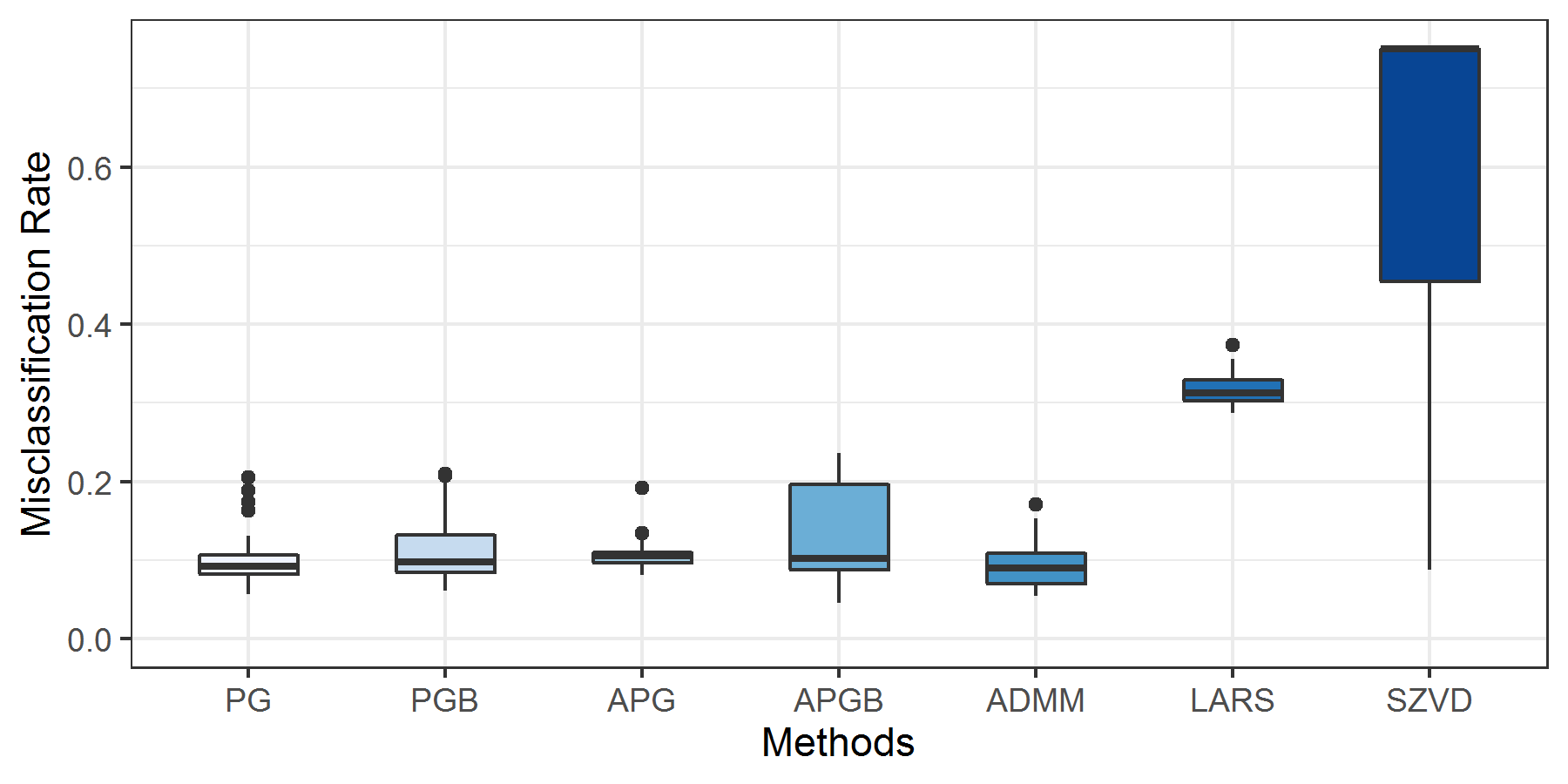}
      \caption{$K=4, r=0.1$}
      \label{fig:(4,1))}
  \end{subfigure}

  \begin{subfigure}[b]{0.45\textwidth}
      \includegraphics[width=\textwidth]{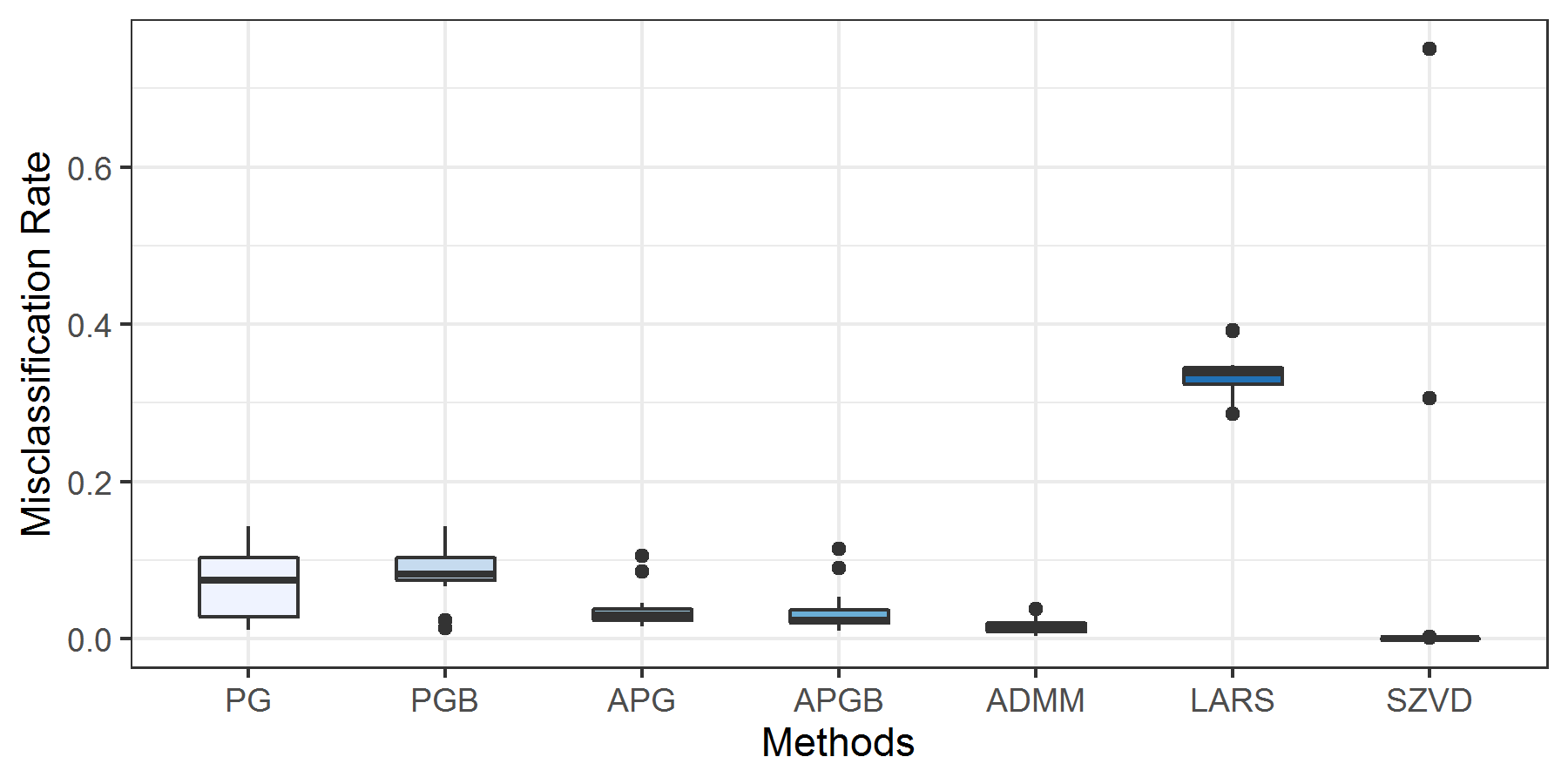}
      \caption{$K=4, r=0.5$}
      \label{fig:(4,5))}
  \end{subfigure}
  ~
  \begin{subfigure}[b]{0.45\textwidth}
      \includegraphics[width=\textwidth]{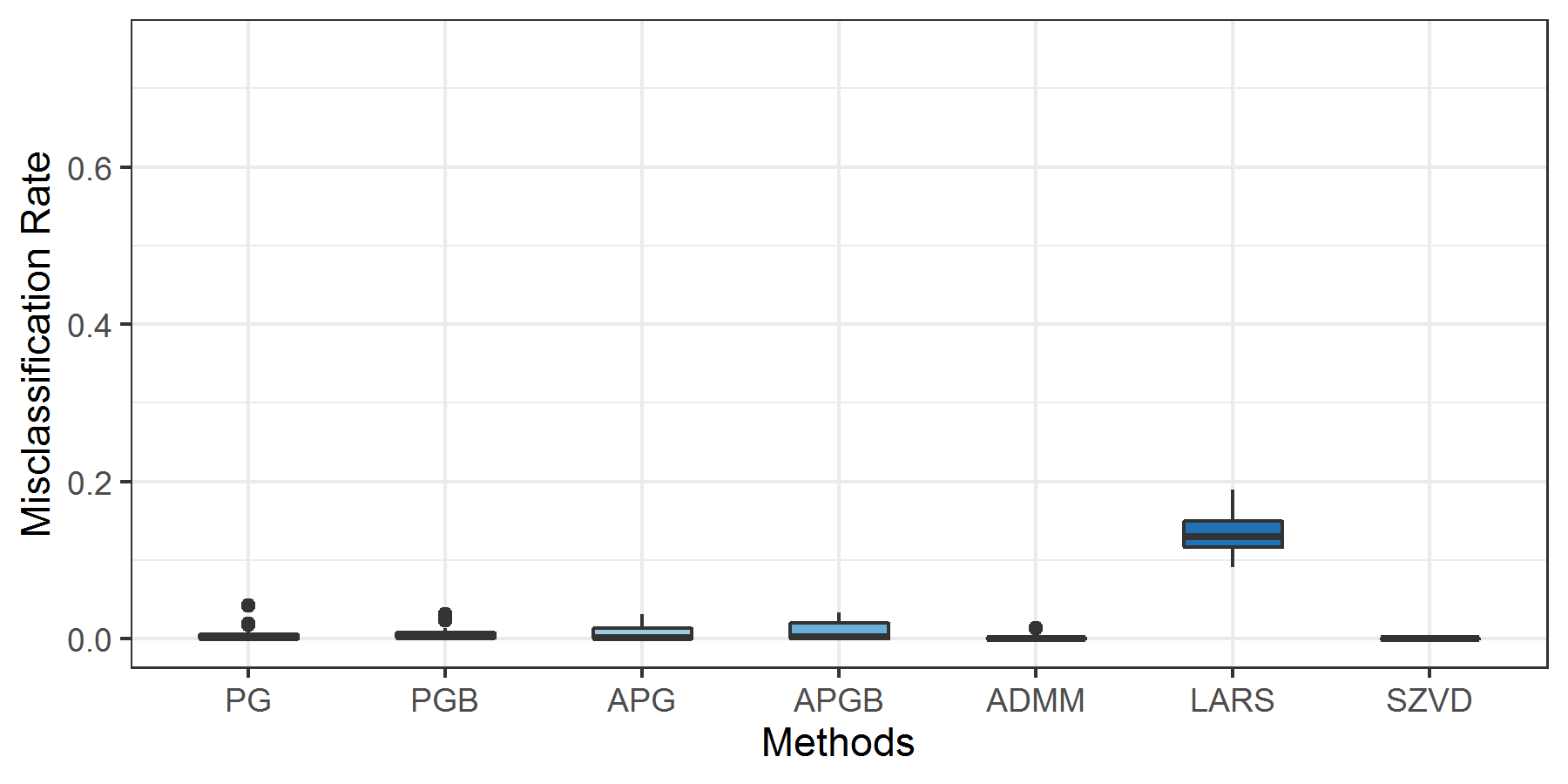}
      \caption{$K=4, r=0.9$}
      \label{fig:(4,9))}
  \end{subfigure}
  \caption{Box plot of out-of-sample misclassification rate  using discriminant vectors calculated using  APG, APGB, PG, PGB, ADMM, SZVD, and LARS for
  $4$-class Gaussian data with class means defined by~\eqref{eq:mu-def1} and covariance vector $\bs{\Sigma}$ for given values of $r$.
  In all experiments, $n_{train} = 100$ and $n_{test} = 1000$.
  }
  \label{fig:gauss=err4}
\end{figure}

\begin{figure}[!htbp]
\centering

  \begin{subfigure}[b]{0.45\textwidth}
      \includegraphics[width=\textwidth]{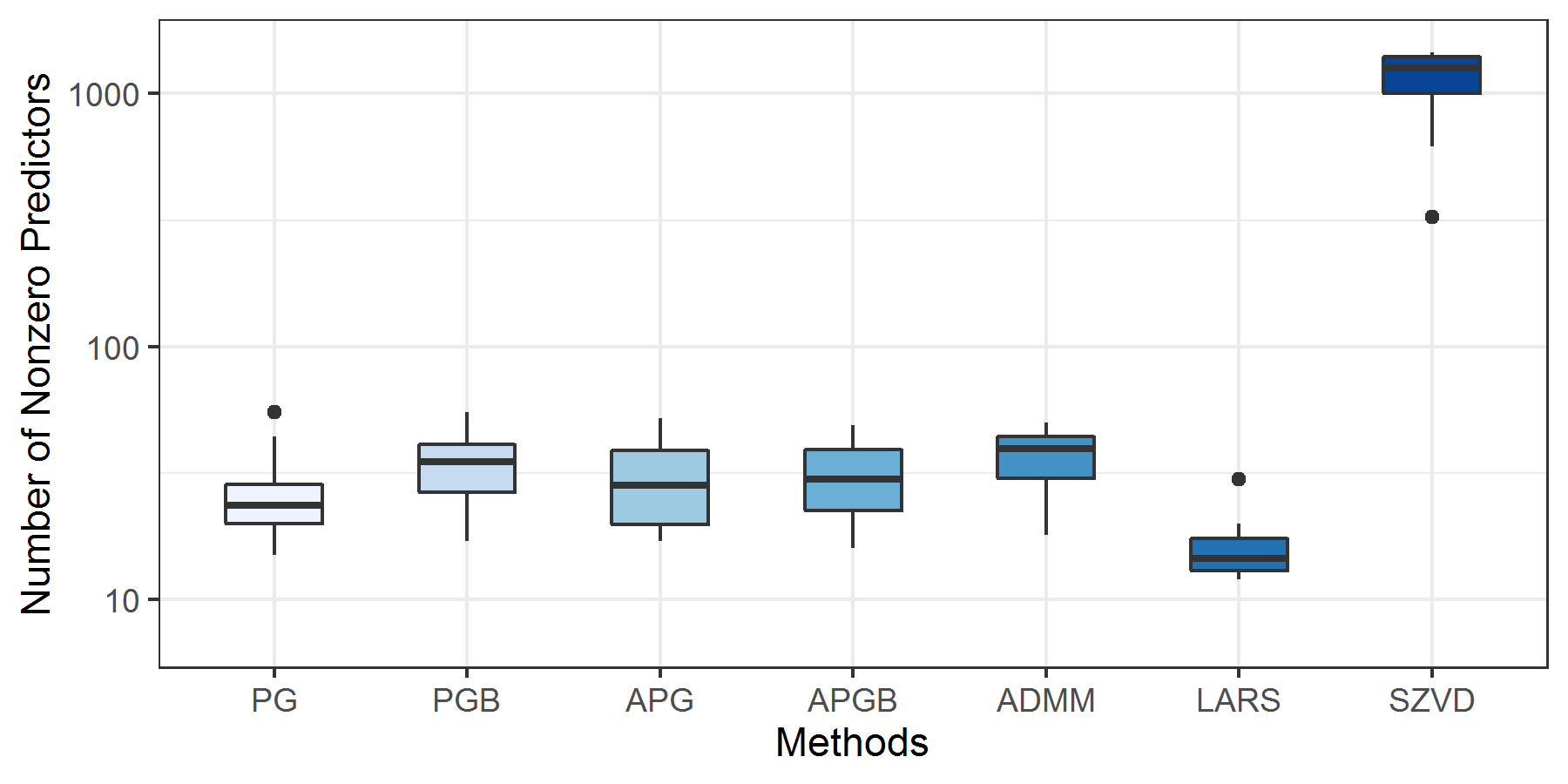}
      \caption{$K=2, r=0$}
      \label{fig:(2,0)card}
  \end{subfigure}
  ~
  \begin{subfigure}[b]{0.45\textwidth}
      \includegraphics[width=\textwidth]{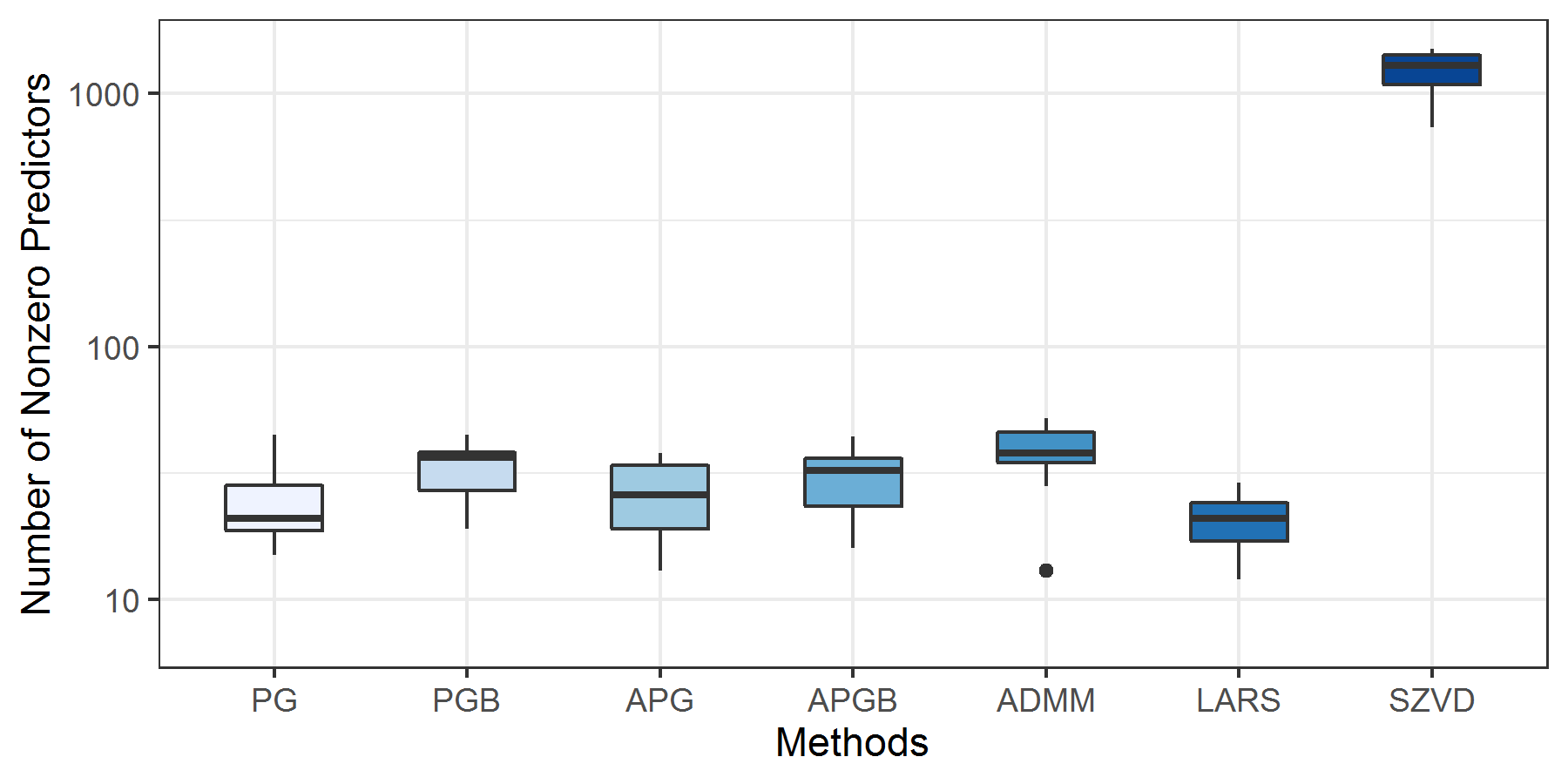}
      \caption{$K=2, r=0.1$}
      \label{fig:(2,1)card}
  \end{subfigure}

  \begin{subfigure}[b]{0.45\textwidth}
      \includegraphics[width=\textwidth]{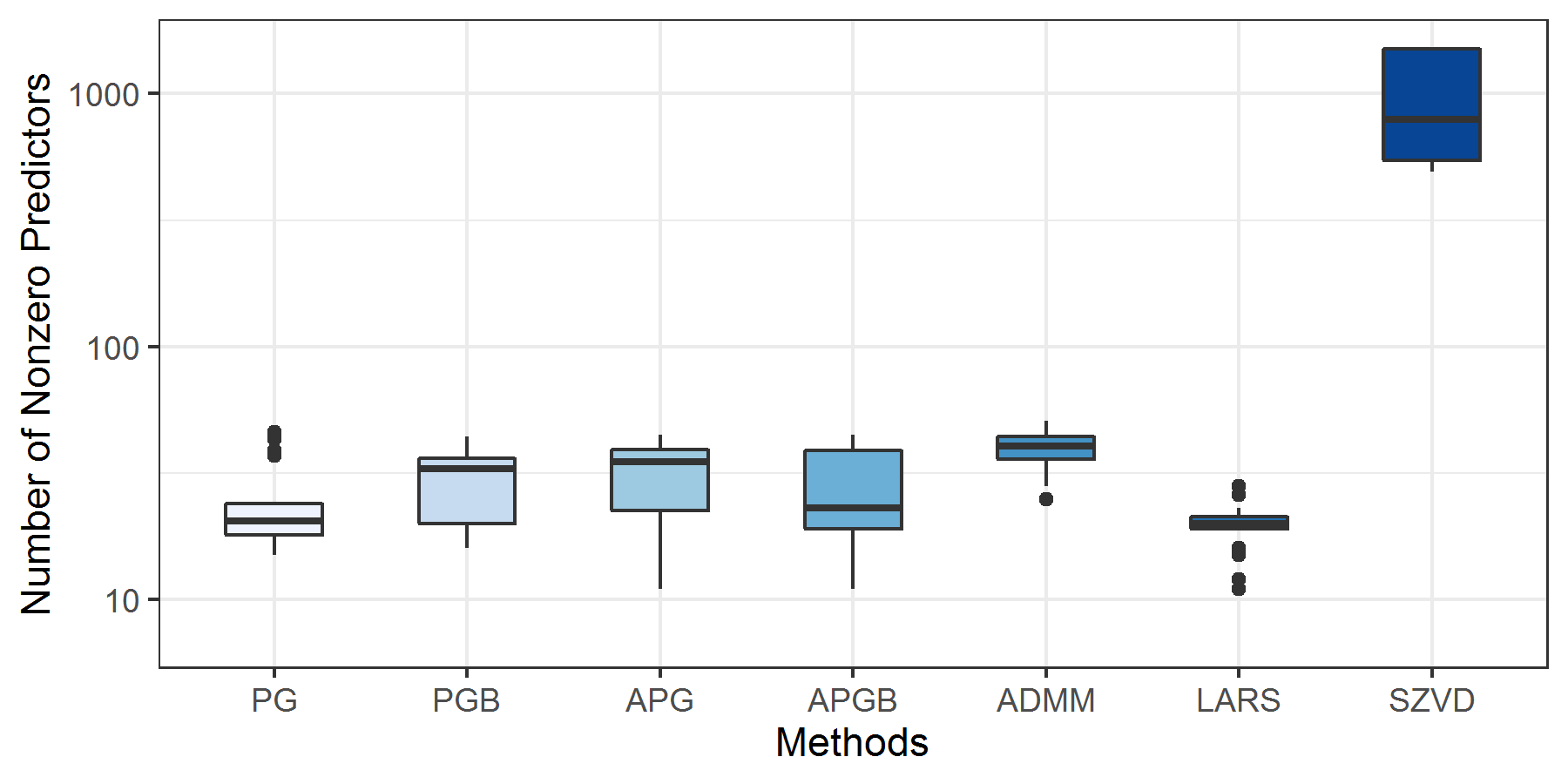}
      \caption{$K=2, r=0.5$}
      \label{fig:(2,5)card}
  \end{subfigure}
  ~
  \begin{subfigure}[b]{0.45\textwidth}
      \includegraphics[width=\textwidth]{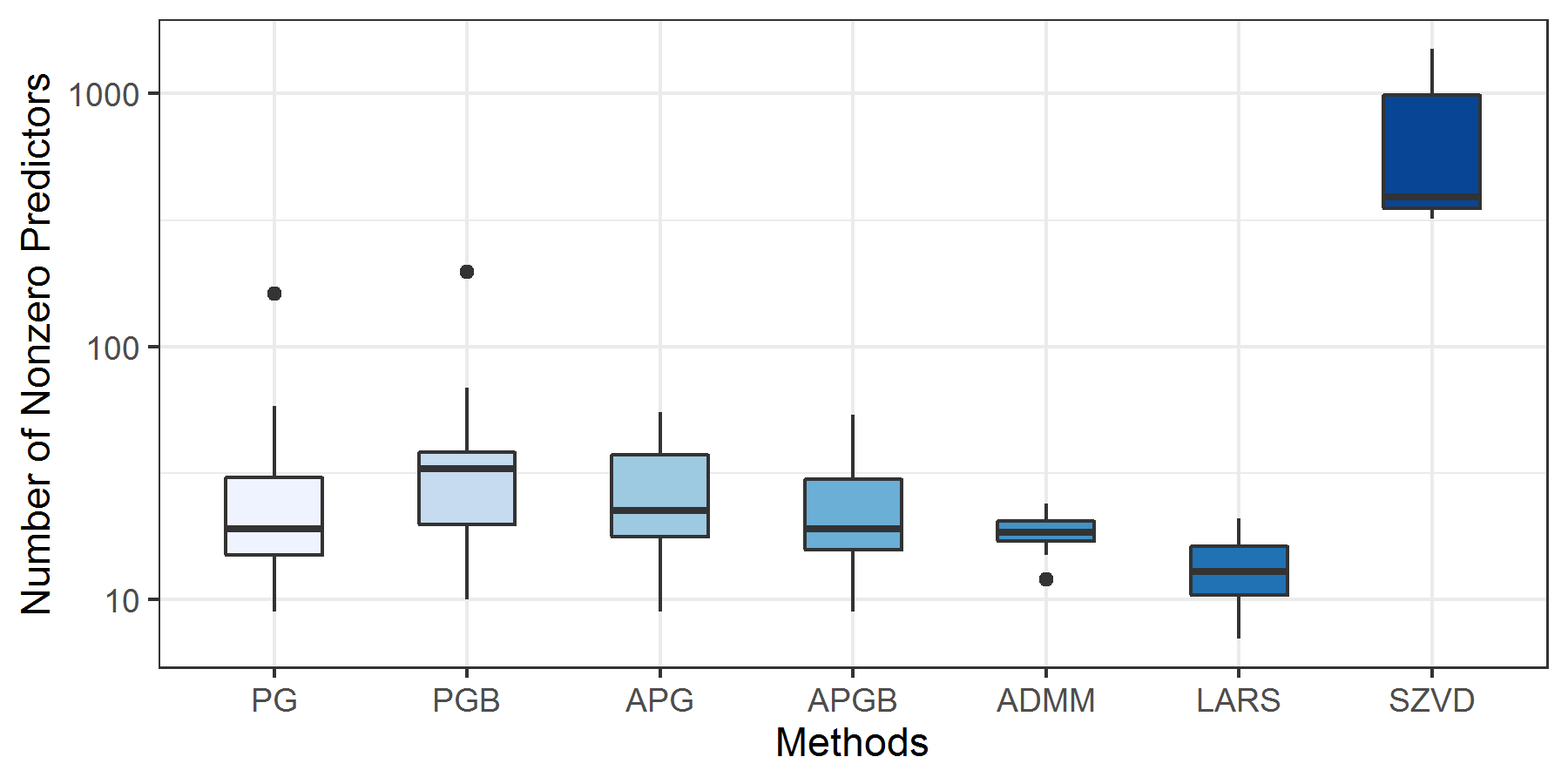}
      \caption{$K=2, r=0.9$}
      \label{fig:(2,9)card}
  \end{subfigure}
  \caption{Box plot of number of nonzero entries of discriminant vectors and optimal scoring vectors 
  calculated using APG, APGB, PG, PGB, ADMM, SZVD, and LARS for classifying $2$-class Gaussian data with mean-vectors defined by~\eqref{eq:mu-def1} and covariance vector $\bs{\Sigma}$ for given values of $r$.}
  \label{fig:gauss=card}

\end{figure}

\begin{figure}
  \centering
  \begin{subfigure}[b]{0.45\textwidth}
      \includegraphics[width=\textwidth]{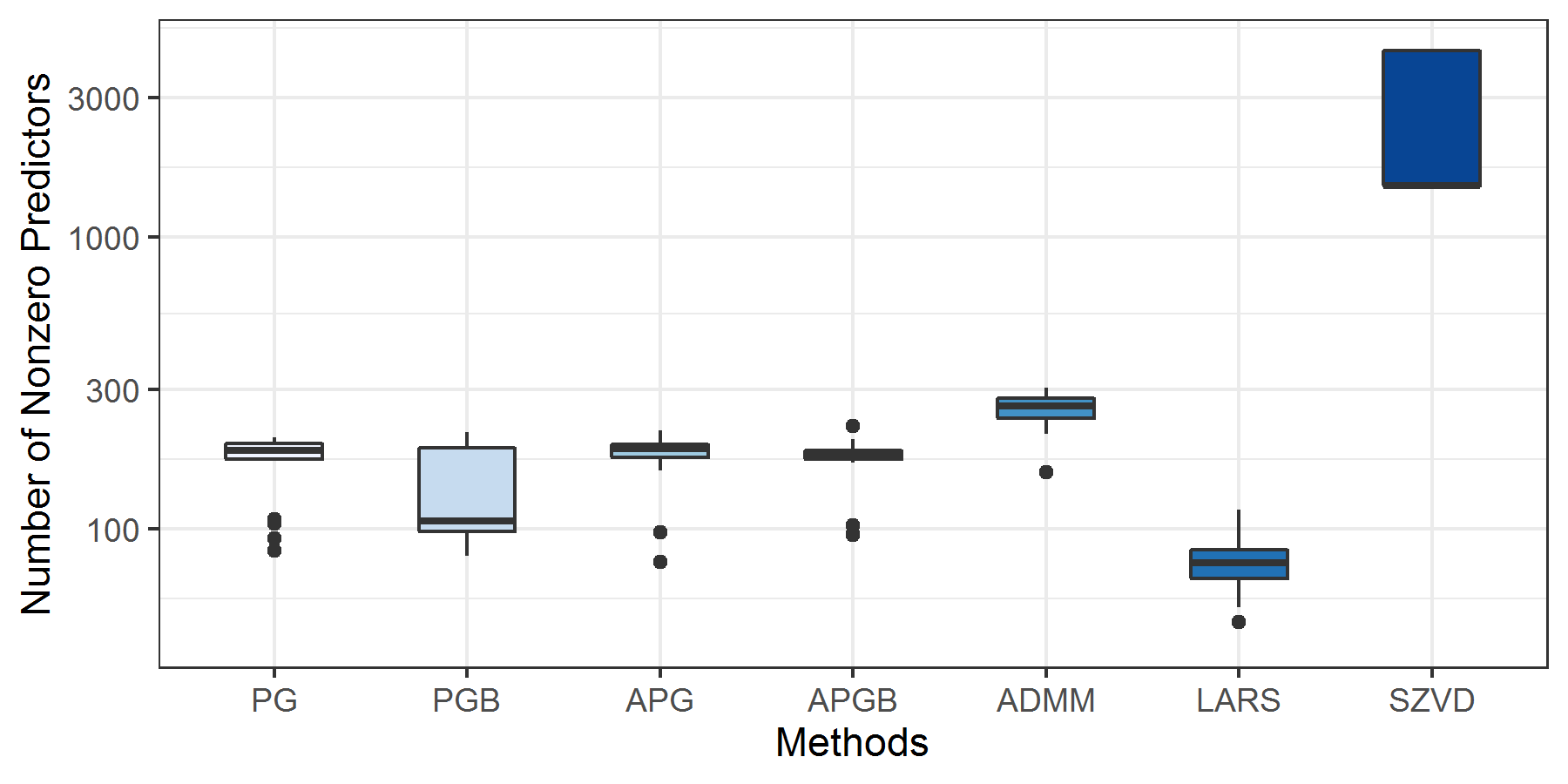}
      \caption{$K=4, r=0$}
      \label{fig:(4,0)card}
  \end{subfigure}
  ~
  \begin{subfigure}[b]{0.45\textwidth}
      \includegraphics[width=\textwidth]{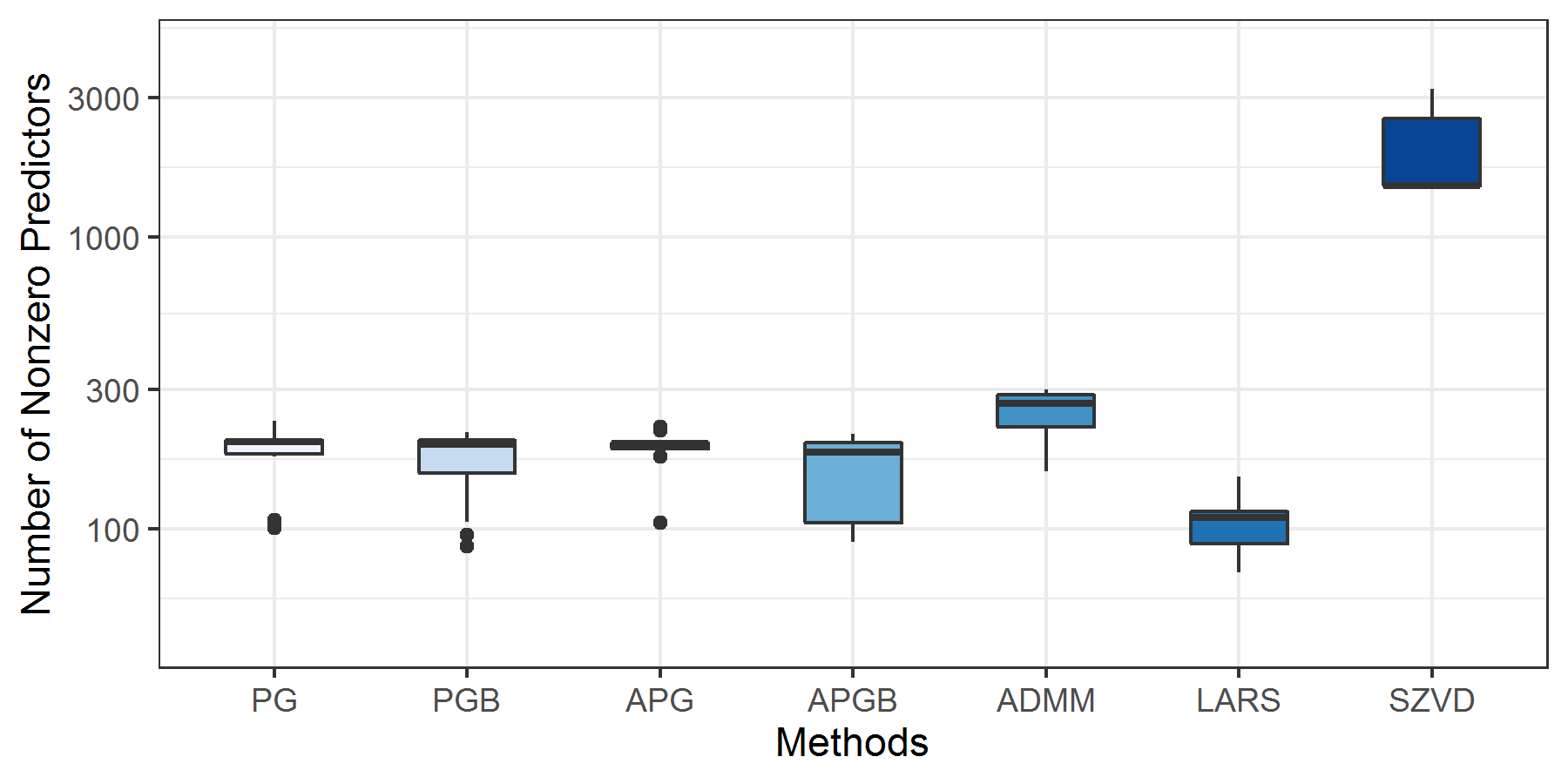}
      \caption{$K=4, r=0.1$}
      \label{fig:(4,1)card}
  \end{subfigure}

  \begin{subfigure}[b]{0.45\textwidth}
      \includegraphics[width=\textwidth]{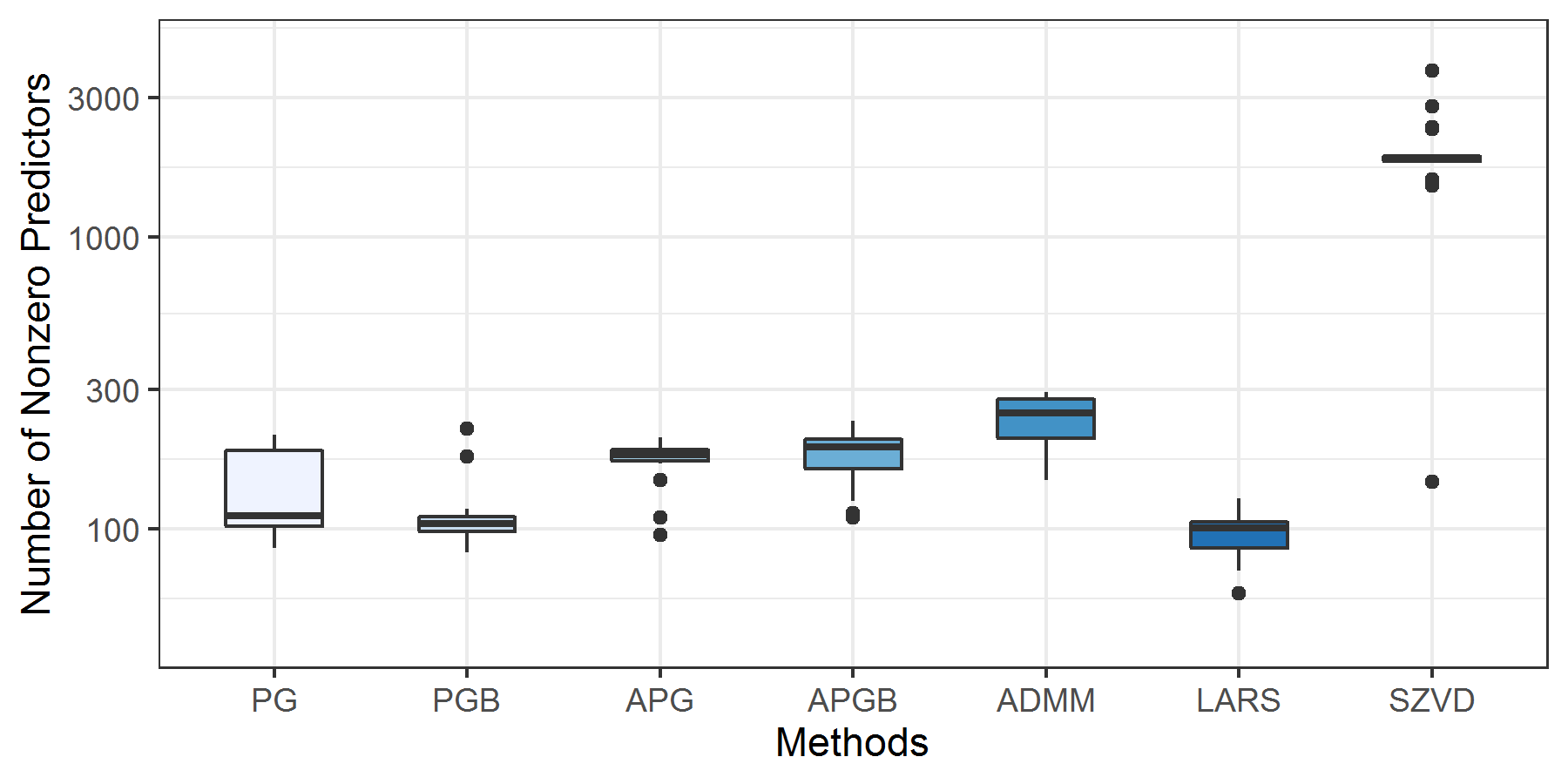}
      \caption{$K=4, r=0.5$}
      \label{fig:(4,5)card}
  \end{subfigure}
  ~
  \begin{subfigure}[b]{0.45\textwidth}
      \includegraphics[width=\textwidth]{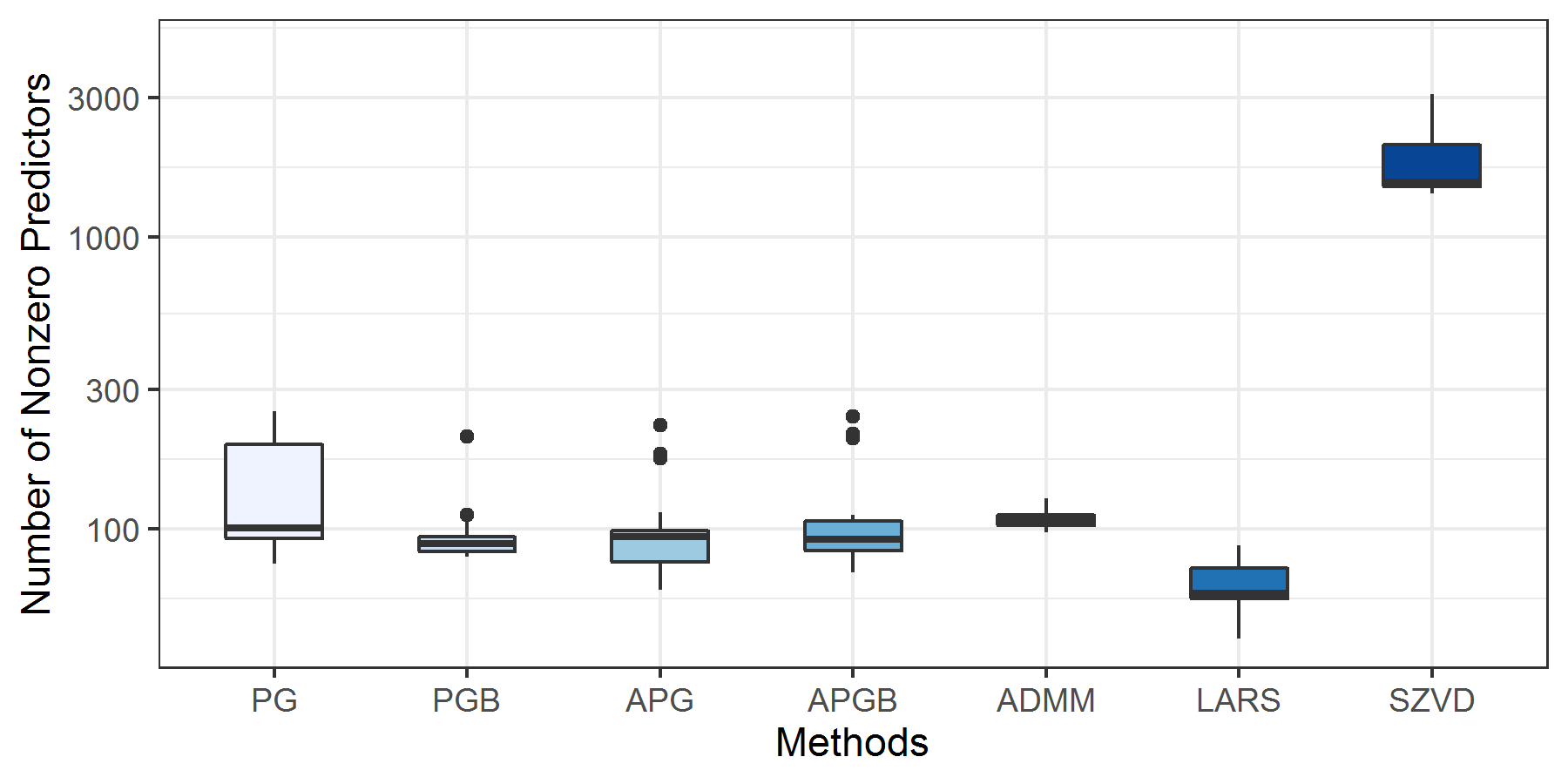}
      \caption{$K=4, r=0.9$}
      \label{fig:(4,9)card}
  \end{subfigure}
  \caption{Box plot of number of nonzero entries of discriminant vectors and optimal scoring vectors 
  calculated using APG, APGB, PG, PGB, ADMM, SZVD, and LARS for classifying $4$-class Gaussian data with mean-vectors defined by~\eqref{eq:mu-def1} and covariance vector $\bs{\Sigma}$ for given values of $r$.}
  \label{fig:gauss=card4}
\end{figure}

\begin{figure}
\centering

  \begin{subfigure}[b]{0.45\textwidth}
      \includegraphics[width=\textwidth]{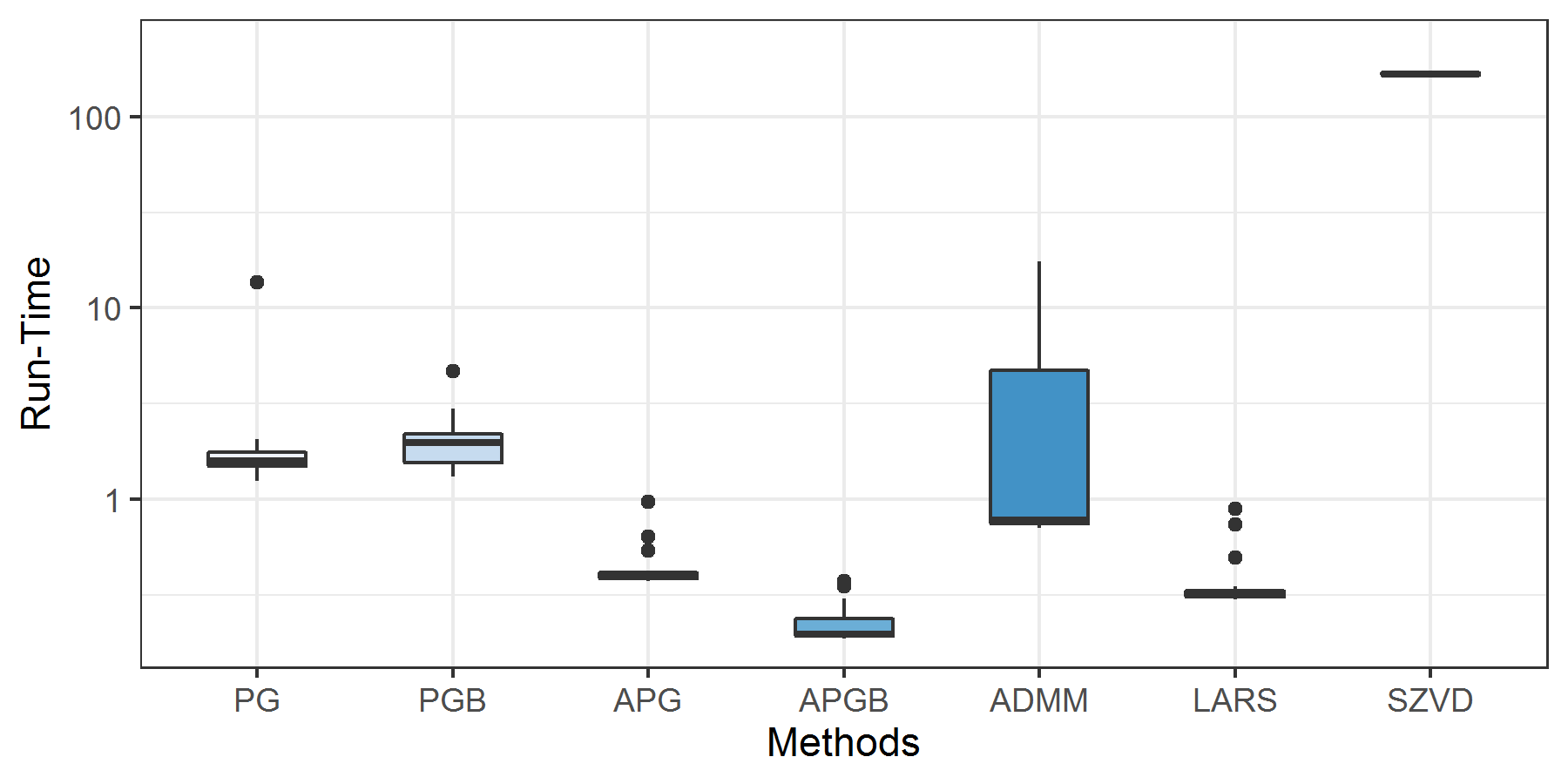}
      \caption{$K=2, r=0$}
      \label{fig:(2,0)Time}
  \end{subfigure}
  ~
  \begin{subfigure}[b]{0.45\textwidth}
      \includegraphics[width=\textwidth]{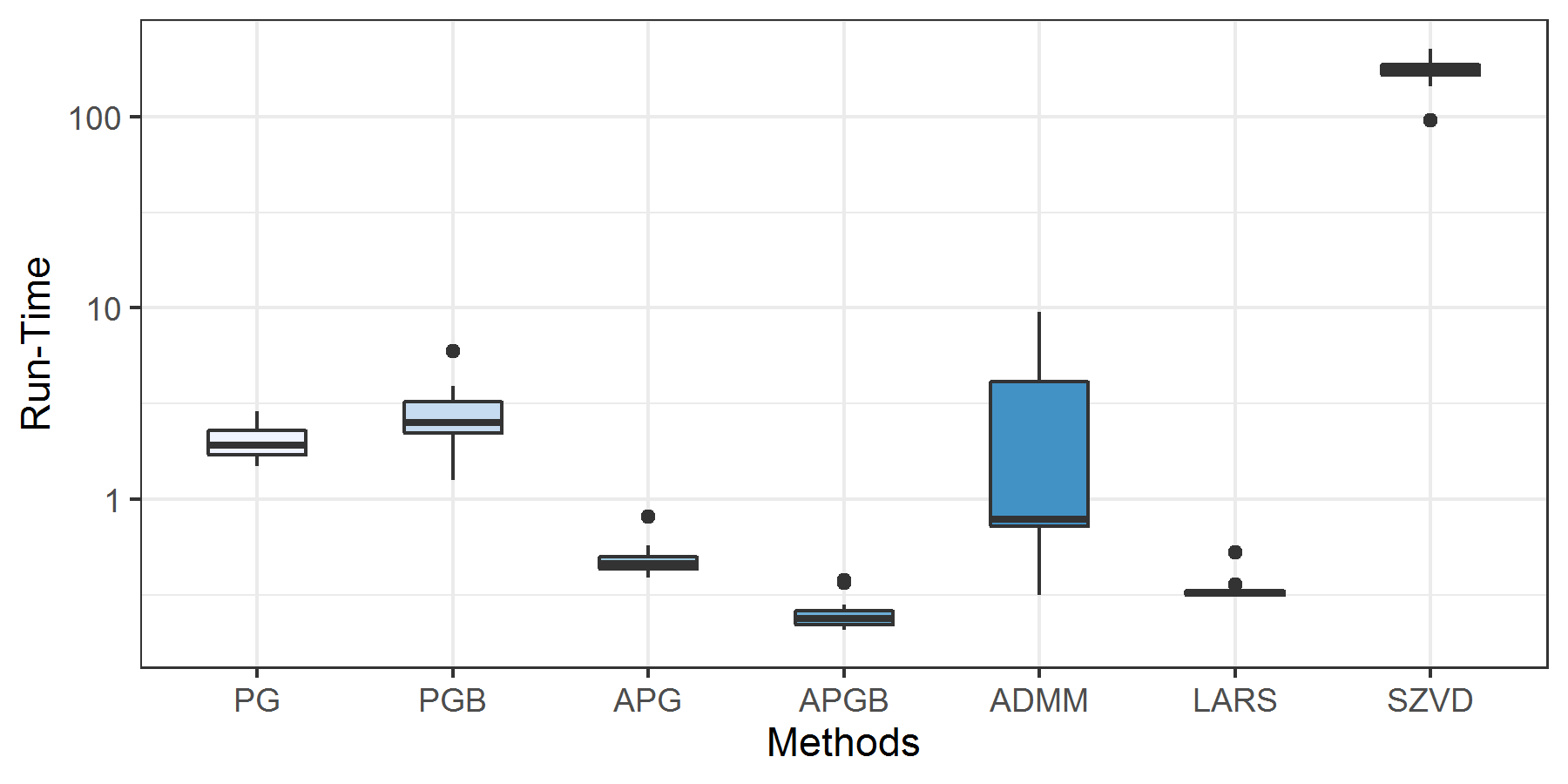}
      \caption{$K=2, r=0.1$}
      \label{fig:(2,1)Time}
  \end{subfigure}

  \begin{subfigure}[b]{0.45\textwidth}
      \includegraphics[width=\textwidth]{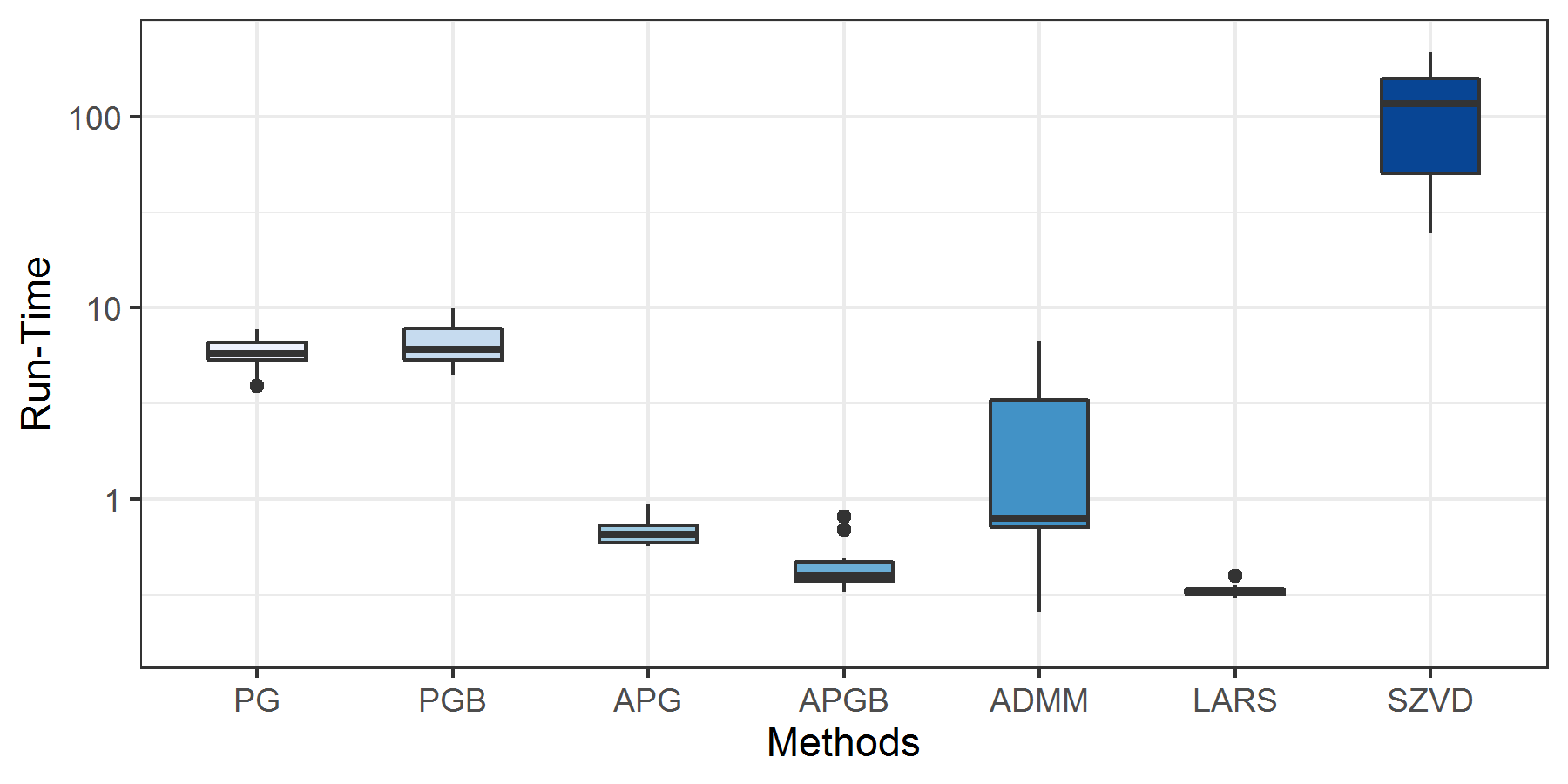}
      \caption{$K=2, r=0.5$}
      \label{fig:(2,5)Time}
  \end{subfigure}
  ~
  \begin{subfigure}[b]{0.45\textwidth}
      \includegraphics[width=\textwidth]{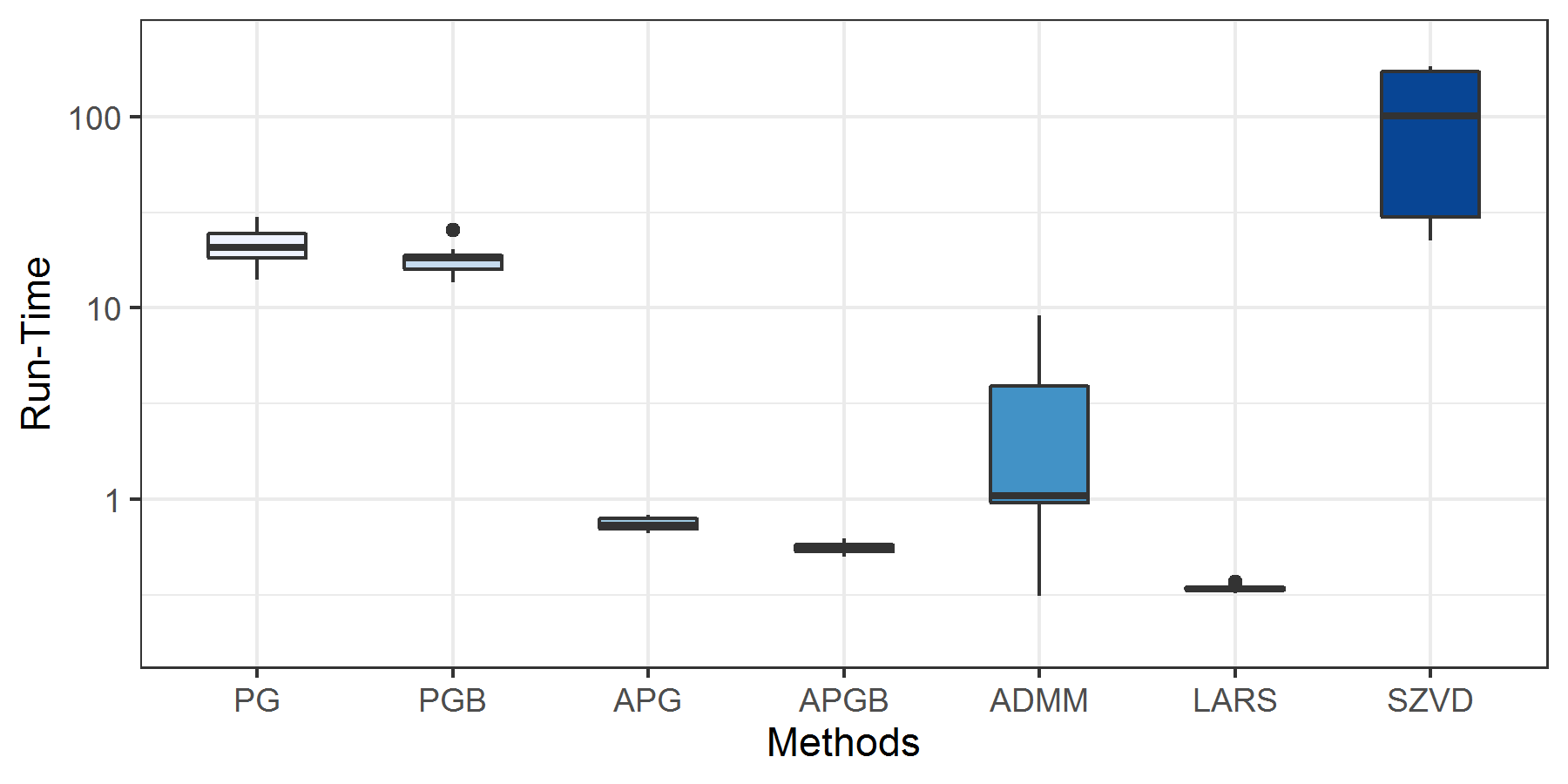}
      \caption{$K=2, r=0.9$}
      \label{fig:(2,9)Time}
  \end{subfigure}

  \caption{Box plots of run-time in seconds (s) for calculation of discriminant vectors and optimal scoring vectors (with error bars of length one standard deviation)
   using APG, APGB, PG, PGB, ADMM, SZVD, and LARS for classifying $2$-class Gaussian data with mean-vectors defined by~\eqref{eq:mu-def1} and covariance vector $\bs{\Sigma}$ for given values of $r$.}
  \label{fig:gauss=Time}
\end{figure}

\begin{figure}
  \centering
  \begin{subfigure}[b]{0.45\textwidth}
      \includegraphics[width=\textwidth]{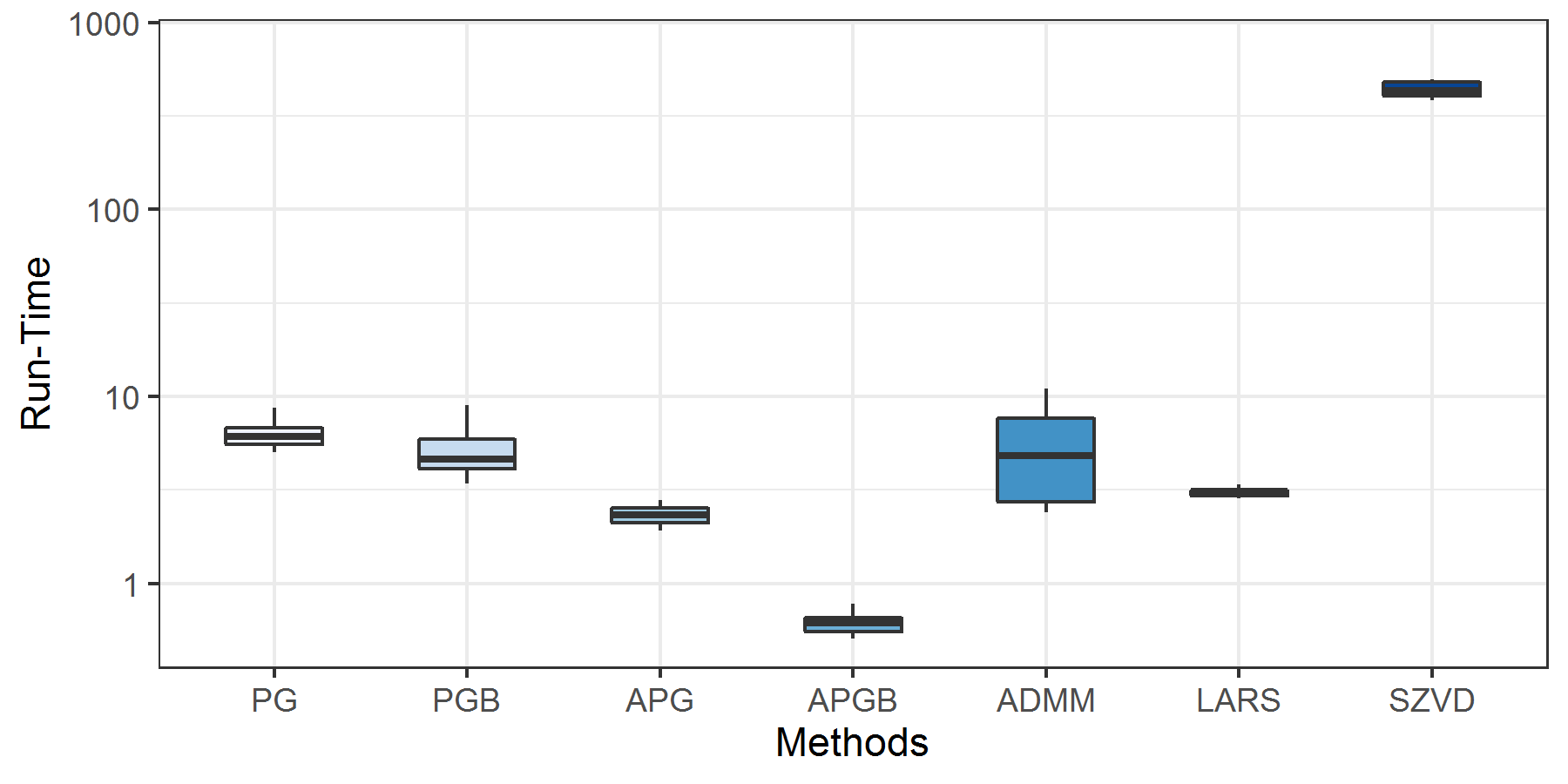}
      \caption{$K=4, r=0$}
      \label{fig:(4,0)Time}
  \end{subfigure}
  ~
  \begin{subfigure}[b]{0.45\textwidth}
      \includegraphics[width=\textwidth]{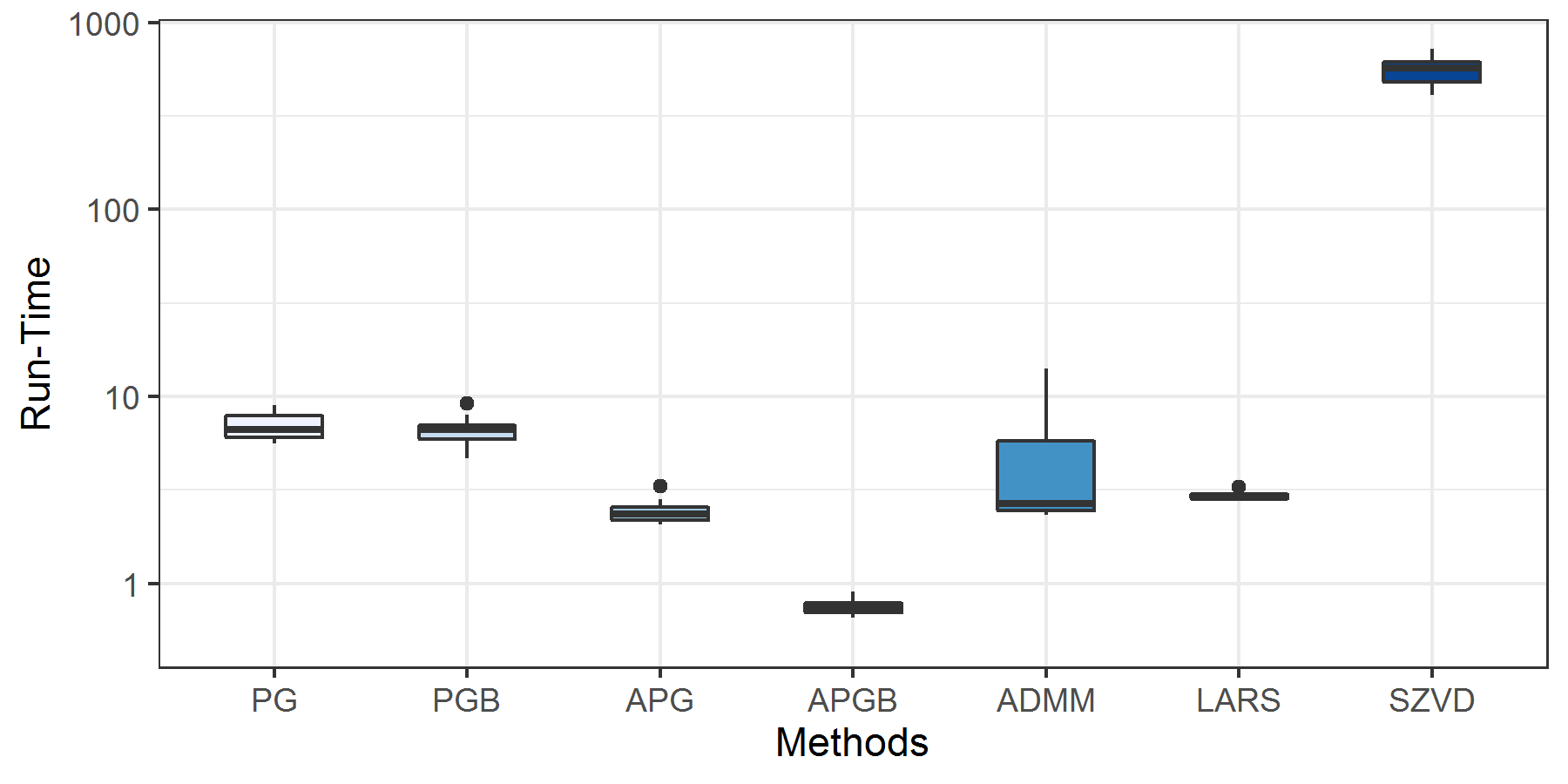}
      \caption{$K=4, r=0.1$}
      \label{fig:(4,1)Time}
  \end{subfigure}

  \begin{subfigure}[b]{0.45\textwidth}
      \includegraphics[width=\textwidth]{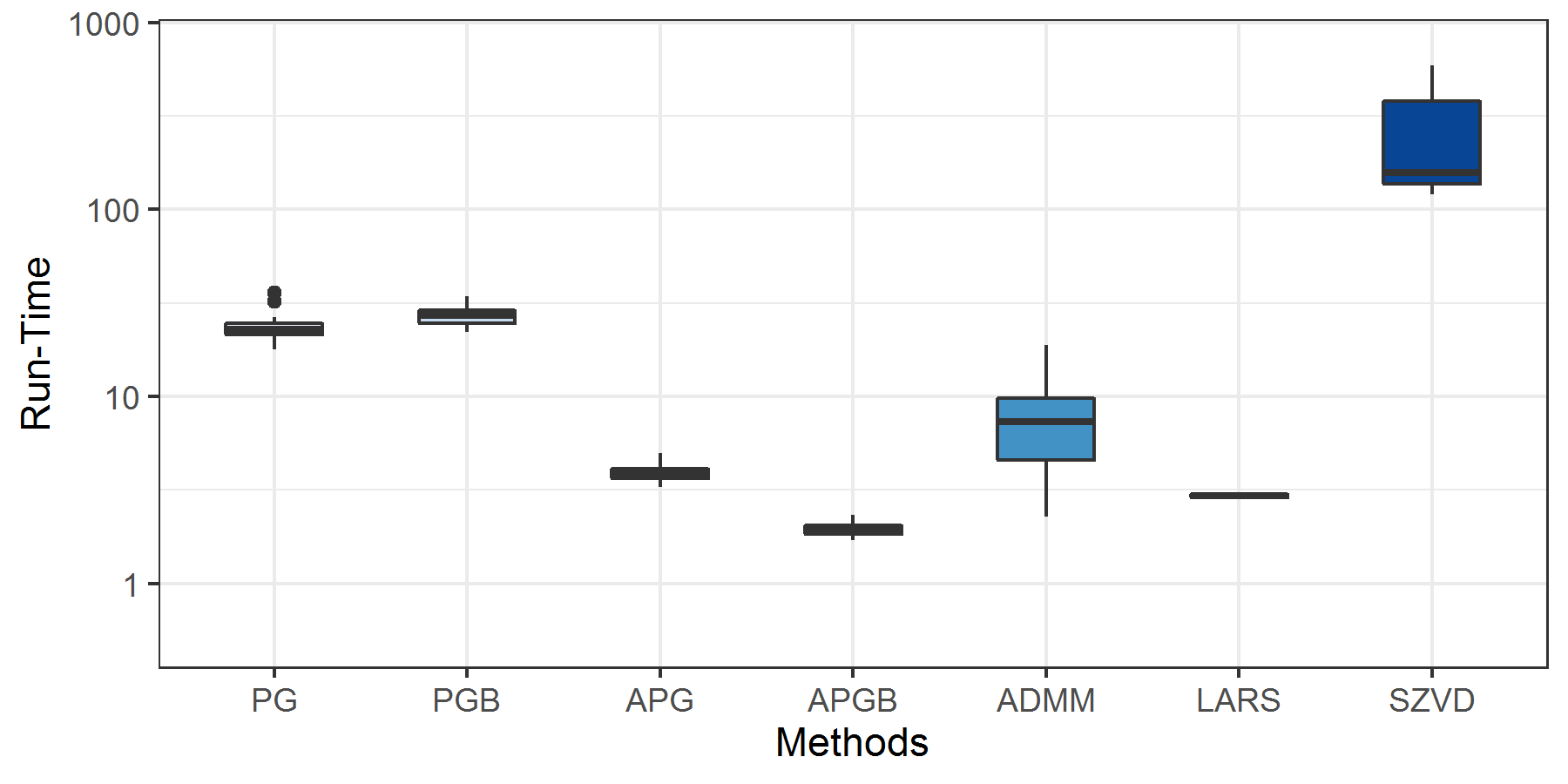}
      \caption{$K=4, r=0.5$}
      \label{fig:(4,5)Time}
  \end{subfigure}
  ~
  \begin{subfigure}[b]{0.45\textwidth}
      \includegraphics[width=\textwidth]{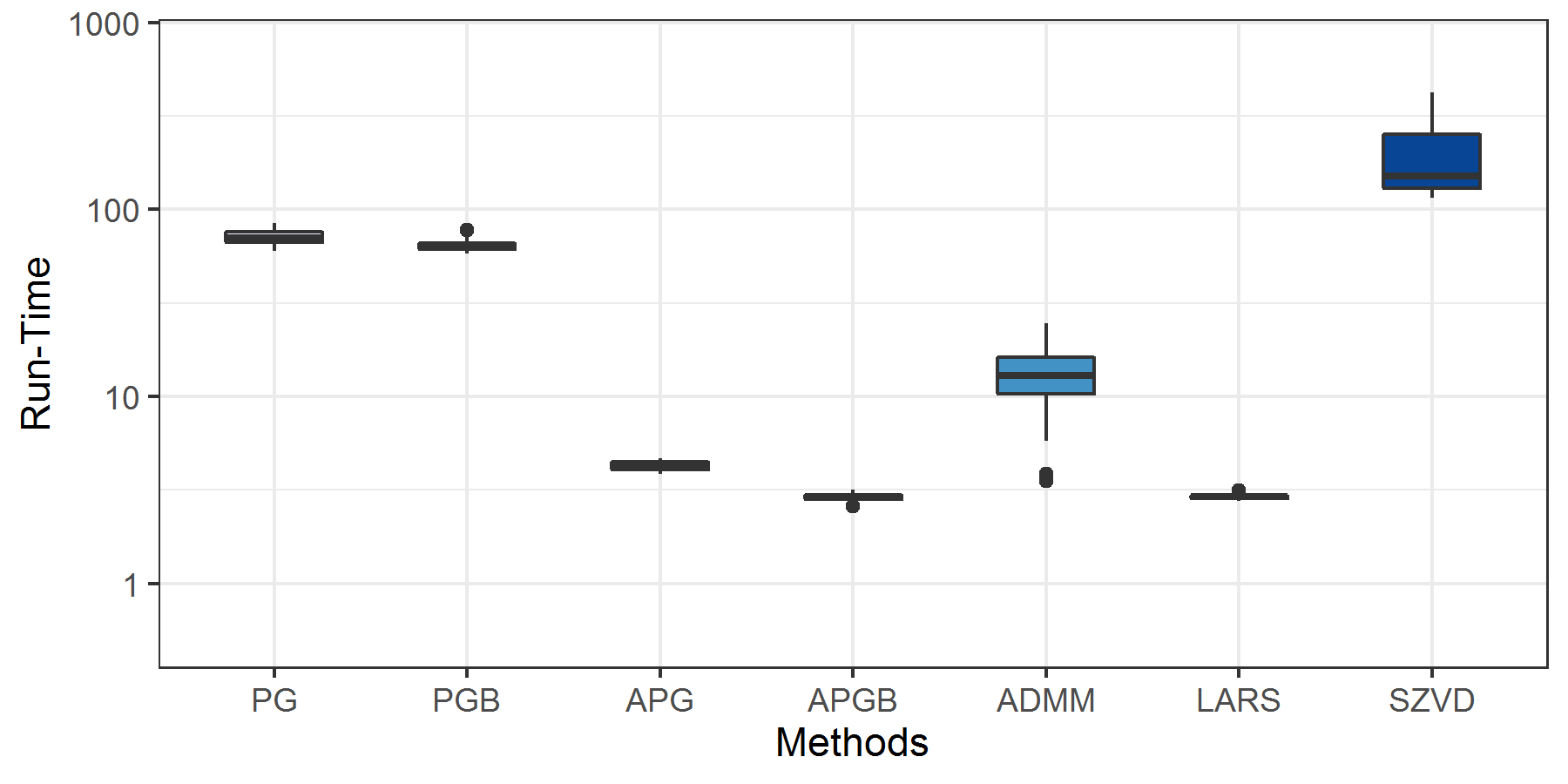}
      \caption{$K=4, r=0.9$}
      \label{fig:(4,9)Time}
  \end{subfigure}
  \caption{Box plots of run-time in seconds (s) for calculation of discriminant vectors and optimal scoring vectors (with error bars of length one standard deviation)
   using APG, APGB, PG, PGB, ADMM, SZVD, and LARS for classifying $4$-class Gaussian data with mean-vectors defined by~\eqref{eq:mu-def1} and covariance vector $\bs{\Sigma}$ for given values of $r$. 
   }
  \label{fig:gauss=Time4}
\end{figure}

%
%

Figures~\ref{fig:gauss=err},~\ref{fig:gauss=err4},~\ref{fig:gauss=card},~\ref{fig:gauss=card4},~\ref{fig:gauss=Time}, and~\ref{fig:gauss=Time4} summarize the results of these experiments;
complete tables of numerical results can be found in the electronic supplemental materials Online Resource 1 and Online Resource 2.
The run-times in Figures~\ref{fig:gauss=Time} and~\ref{fig:gauss=Time4} are reported in seconds $(s)$ and reflect the total computation time required for each method to perform cross validation to tune regularization parameters and train optimal scoring and discriminant vector pairs using the optimized regularization parameters.
The number of nonzero features reported for all methods except SZVD was the count of discriminant vector elements with value not identical to $0$. SZVD does not use the same soft thresholding step as the other methods and returns approximately sparse solutions with entries close to zero, but not exactly zero; we round entries with magnitude at most $10^{-5}$ when reporting cardinality of discriminant vectors obtained using SZVD.


From these simulations, we see that the the classical solution of the SOS problem using LARS tends to yield fewer nonzero predictor variables than the APG and ADMM proximal methods (see Fig.~\ref{fig:gauss=card} and~\ref{fig:gauss=card4}), while requiring comparable computation (see Fig.~\ref{fig:gauss=Time} and~\ref{fig:gauss=Time4}); we note that the average cardinality and run-times observed are typically within one standard deviation of each other (as indicated by the error bars in the plots).
On the other hand, the classifiers provided the LARS heuristic tend to have significantly higher misclassification rate than the proposed proximal methods (see Fig.~\ref{fig:gauss=err} and~\ref{fig:gauss=err4}). We suspect that this is due to our method for choosing the regularization parameter $\lambda$. Here, we choose the value of $\lambda$ that minimizes in-sample classification accuracy, and break any ties by choosing the value of $\lambda$ which yields the sparsest classifier. This appears to cause some overfitting of the classifier to the training data, where the in-sample high accuracy does not generalize to high out-of-sample accuracy.
In an earlier draft of this manuscript, we considered training $\lambda$ via cross-validation in a manner similar to that used for PG, PGB, APG, APGB, and ADMM; this produced a far lower classification rate, comparable to the other methods, although at a significantly higher computational cost, as were required to calculate the full regularization path for each fold and each choice of $\lambda$. In another previous draft, we chose $\lambda$ in our LARS-based analysis to minimize out-of-sample error. Again, this leads to a significant decrease in misclassification rate, although it is unfair to compare this approach to other methods which do not use this information to choose $\lambda$, and may not be used in practical applications where an out-of-sample ground truth is unavailable.

In all experiments, the sparse zero-variance discriminant analysis (SZVD) heuristic  performed poorly in terms of computation, sparsity, and accuracy.
The observed increased run-times agree with that predicted in~Table~\ref{t:total-flops}. On the other hand, the relatively high density and misclassification rate can be explained by the tendency of SZVD to misconverge to an all-zeros solution when $\lambda$ is large. We observe a moderate number of trials where SZVD yields an all-zero solution (with high error rate), and remaining trials generating discriminant vectors with many nonzero entries (corresponding to relatively small $\lambda$).

We further investigate this phenomena in the following sections.

\subsection{Convergence Experiments}
\label{sec:convtrials}

The empirical results of~Section~\ref{sec:Gauss} suggest that the use of our proposed proximal methods (PG, PGB, APG, APGB, ADMM)
for solution of subproblem~\eqref{eq: b prob} can lead to  improvement in terms of classification accuracy and overall run-time over the least angle regression algorithm in certain settings. To further illustrate this improvement, we performed a series of experiments investigating the behaviour of the objective function of~\eqref{eq: prob} during each iteration of these methods.

\begin{figure}[!htbp]
\centering

\begin{subfigure}[b]{0.95\textwidth}
    \includegraphics[width=\textwidth]{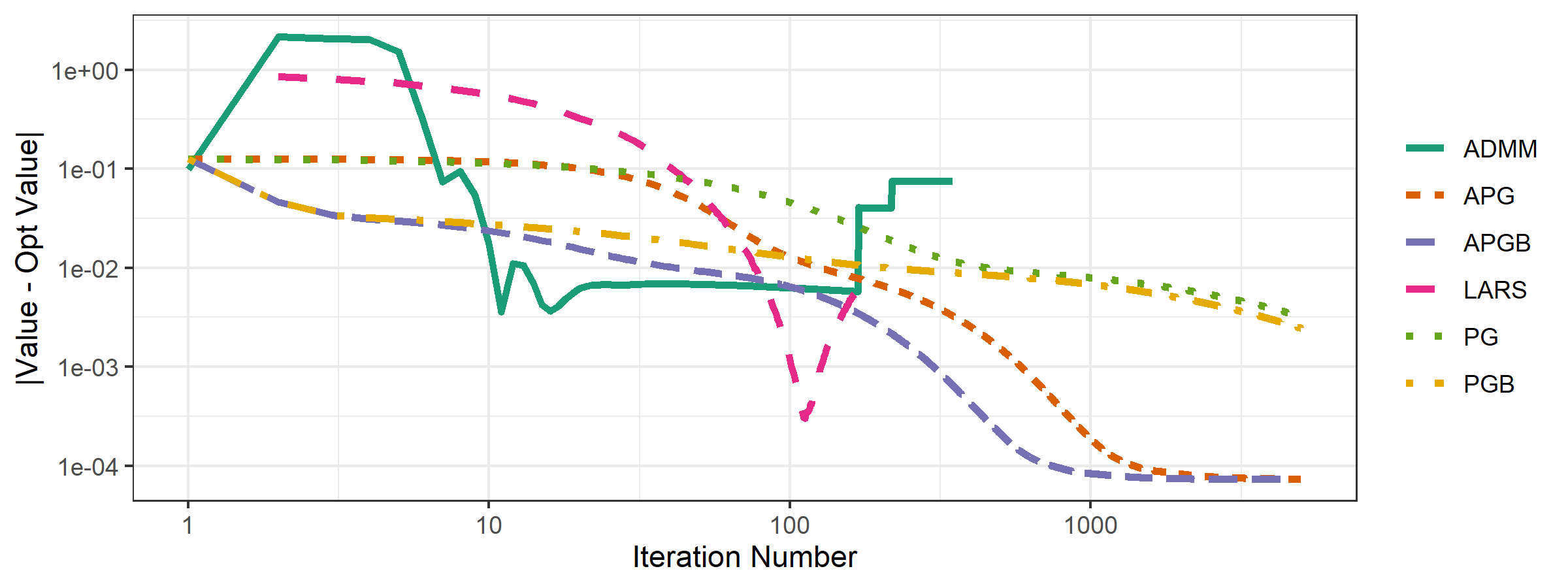}
    \caption{Difference of objective value and optimal value}
    \label{fig:gaussvalit}
\end{subfigure}

\begin{subfigure}[b]{0.95\textwidth}
    \includegraphics[width=\textwidth]{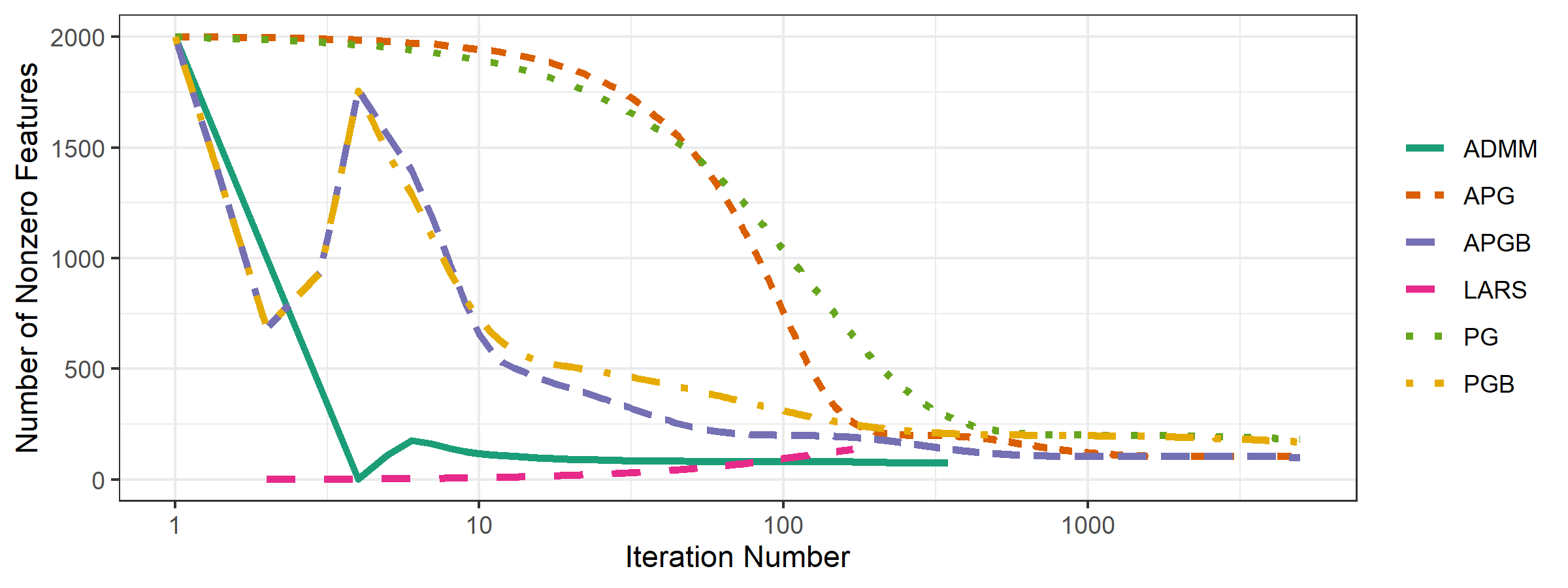}
    \caption{Number of Nonzero Features}
    \label{fig:gaussfeat}
\end{subfigure}

\begin{subfigure}[b]{0.95\textwidth}
    \includegraphics[width=\textwidth]{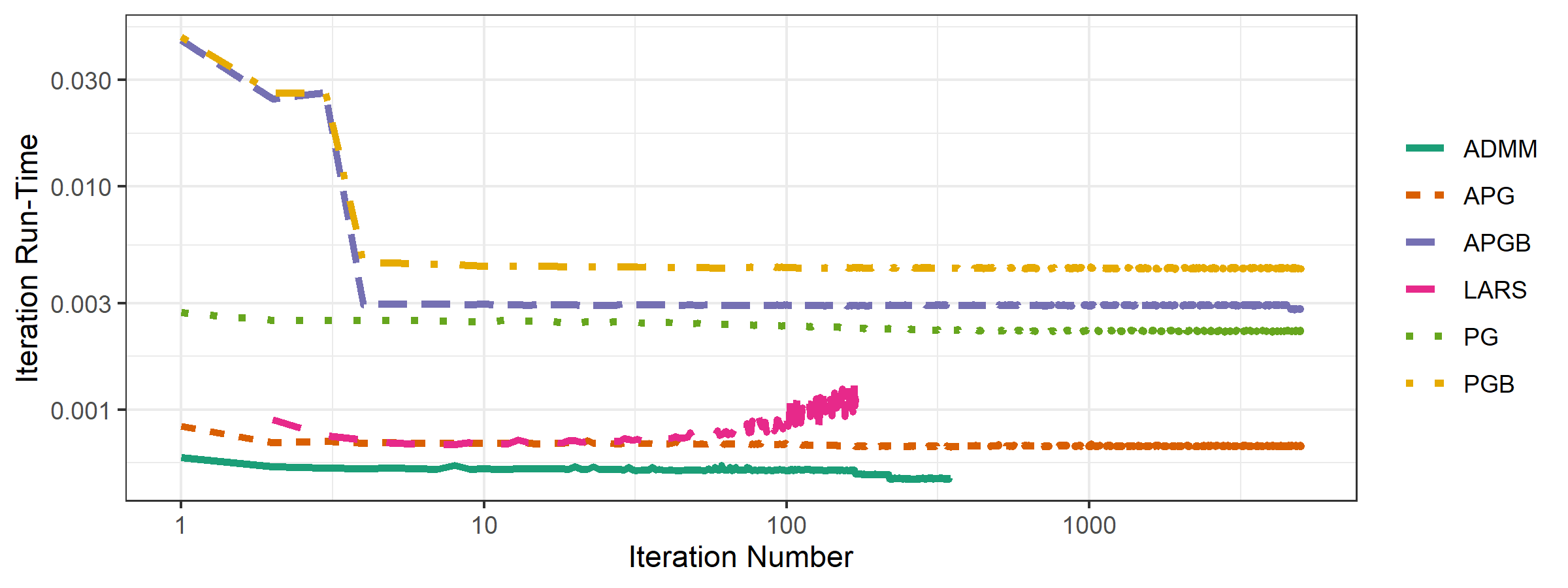}
    \caption{Iteration Run-time}
    \label{fig:gausstime}
\end{subfigure}
\caption{Absolute value of optimality gap, cardinality of iterate, and iteration run-time averaged over 50 calls to solve~\eqref{eq: b prob} for $25$ different Gaussian data sets. We note that LARS terminates in the fewest number of iterations, followed by ADMM.}
\label{fig:gauss=plot}
\end{figure}

We generated random data sets with observations sampled from one of two normal distributions, $N(\bs{\mu_1}, \bs{\Sigma})$ or $N(\bs{\mu_2}, \bs{\Sigma})$. Specifically, we sampled $n=200$ training observations from each of the $p$-dimensional multivariate Gaussian distributions for $p=2000$ with mean vectors $\bs{\mu_1}$ and $\bs{\mu_2} \in \R^p$, respectively, satisfying
\begin{equation}\label{eq:mu-def2}
    [\bs{\mu_i}]_j =
    \begin{cases}
      0.7, & \text{if } \ceil{p/3}(i-1)  < j \le \ceil{p/3} i \\
      0, & \text{otherwise},
    \end{cases}
\end{equation}
for all $j=1,2,\dots, p$, and covariance matrix $\bs{\Sigma} \in \R^{p\times p}$ constructed as Section~\ref{sec:Gauss} with $r = 0.75$.
For each data set, we train nearest centroid classifiers using discriminant vectors obtained by approximately solving~\eqref{eq: prob} with each method used above (PG, PGB, APG, APGB, ADMM, LARS) to solve ~\eqref{eq: b prob}. We stop each method after either $10000$ iterations have been performed or the stopping condition is met with tolerance $10^{-8}$.
We use regularization parameters $\gamma = 10^{-3}$, $\bs{\Omega} = \bs{I}$ and we set $\lambda = 0.25 \bar{\lambda}$, where $\bar\lambda$ is given by~\eqref{eq:barlam}; we stop LARS when a solution with cardinality $0.25 p = 500$ is found. We use augmented Lagrangian parameter $\mu = 2$ in ADMM.
We use the parameters $\bar L = 0.25$ and $\eta = 1.25$ in the back tracking line search.
Note that these stopping conditions are more strict than those used in the previous section (Sect.~\ref{sec:Gauss}); we use these stopping criteria to ensure that a large number of iterations are performed in order to obtain a clearer picture of the convergence properties of the various algorithms.
We generate $25$ problem instances and solve each problem $50$ times using each method to control for natural variation in computation time and problem instances; each algorithm will generate the same sequence of iterates and solution each time called for each problem (up to sign changes due to random initialization of $\bs{\theta}$).
We validate performance of our classifiers using balanced sets of $200$ testing observations sampled from $N(\bs{\mu_1}, \bs{\Sigma})$ or $N(\bs{\mu_2}, \bs{\Sigma})$.


We chose these data sets to isolate the relationship between the performance of our proposed algorithms for solving~\eqref{eq: b prob} and the overall performance of the proposed block coordinate descent method and nearest centroid classification.
Recall that, in the $K=2$ class case, we calculate exactly one discriminant and scoring vector pair $(\bs{\beta}, \bs{\theta})$ before Algorithm~\ref{alg: BCD} converges. By restricting our focus to a case where we solve exactly one instance of subproblem~\eqref{eq: b prob} using each method during each trial, we can directly compare the behavior of algorithms for solving~\eqref{eq: b prob}.

We recorded the objective value of~\eqref{eq: b prob} and the cardinality of the current iterate $\bs{\beta}^i$ following the $i$th iteration of each algorithm for each method and data set, as well as the run-time of the $i$th iteration. We recorded the value of the augmented Lagrangian function for the ADMM, instead of the objective function value, to indicate the trade-off between optimizing the objective and forcing feasibility ($\x = \y = \bs\beta$) of the split decision variables $\x$ and $\y$. The results of these experiments are summarized in Figure~\ref{fig:gauss=plot} and Table~\ref{tab:conv-summary}.

It is apparent from the results of these simulations that iterations of LARS are more expensive than those of the proximal methods APG and ADMM, especially when the cardinality of the iterate is relatively large.
The per-iteration cost of LARS tends to increase since LARS is an active set method and gradually adds elements to the active set each iteration; the computational complexity is an increasing function of iteration number due to this corresponding increase in cardinality each iteration.
We should also note that LARS tended to terminate more quickly (in terms of number of iterations) than the other methods, but that this is largely a consequence of the more conservative stopping criteria for the proximal methods. In particular, APG, APGB, ADMM all generate iterates with smaller optimality gap than  the suboptimality at termination of LARS iterates, within fewer iterations.
On the other hand, the per-iteration cost of each proximal  method is largely consistent across iterations. The ADMM tended to terminate in fewer iterations and yield sparser solutions than the other proximal methods; the per-iteration cost of the ADMM is also less than (or comparable to) the other proximal methods in all trials. The cardinality of iterates generated by ADMM also decreased much more quickly than those generated by the other proximal methods; this may be due to the fact that the soft thresholding operator is applied directly to the iterate $\bs{y}^i$, rather than following a gradient step applied to the previous iterate or a weighted average of the previous two iterates as in the proximal gradient and accelerated gradient methods. 
The trajectory of suboptimality of ADMM iterates is not a smooth descent like the other proximal methods since we plot the absolute value of the gap between the augmented Lagrangian value and the average optimal value. Here, the augmented Lagrangian value of the ADMM iterates tend to initially increase followed by a several iterates of sharp decrease (with value typically less than the optimal value), followed by gradual increase to the optimal value; plotting the absolute value of the optimality gap allows us plot using logarithmic scale.
Similarly, the value of LARS iterates tends to decrease, and then increase prior to termination. This is a consequence of the stopping criteria, i.e., terminating when a sufficiently dense solution is found; here, we use the iterate at termination rather than the iterate with minimum value among LARS iterates as our discriminant vector when calculating misclassification rates and cardinality in Table~\ref{tab:conv-summary}.

These trials also suggest some modest value in the use of back tracking line searches. In each set of trials, the proximal gradient methods with back tracking line search terminated in fewer iterations than with a constant step size.
However, the additional cost of performing the line search frequently caused the overall computational time of the back tracking methods to exceed that of the constant step size methods.
This additional cost observed here is more dramatic than that observed in Section~\ref{sec:Gauss}. We remind the reader that the reported computation time for the experiments of Section~\ref{sec:Gauss} includes all computation to perform cross validation to train the regularization parameter $\lambda$;
the discrepancy between the timing results in Section~\ref{sec:Gauss} and here highlights a potential sensitivity of Algorithm~\ref{alg: BCD} to choice of $\lambda$, and variation in training data (in this case with respect to training and validation splits in the cross validation scheme).

\begin{table}[t]
    \centering
    \begin{tabular}{ccccc} \toprule
    Method &  Run-times (s) & Cardinality & Iterations \\ \midrule 
         PG&43.38 &151.92&10000 \\
PGB&63.98&141.52&1000 \\
APG&28.79&100.88&10000 \\
APGB&33.87&100.92&6633.92 \\
ADMM&1.43&79.6&484.4 \\
LARS&0.21&145.88&195.72 \\ \bottomrule
    \end{tabular}
    \caption{Average run-time, cardinality of solution, and number of iterations before termination averaged over 25 Gaussian data sets, solved 50 times using each method. No methods returned classifiers yielding out-of-sample classification error in any trial. PG, PGB, APG failed to return a $10^{-8}$ suboptimal solution within $10000$ iterations, while ADMM and LARS both converge within 500 iterations on average.}
    \label{tab:conv-summary}
\end{table}

\subsection{Differences between discriminant vectors due to algorithm choice}
At this point, we should note that Subproblem~\eqref{eq: b prob} is strongly convex by the choice of regularization parameter $\bs\Omega = \I$ in all experiments considered so far.
As a consequence,~\eqref{eq: b prob} has a \emph{unique} solution. One would naively expect Algorithm~\ref{alg: BCD} to generate the same discriminant vector regardless of choice of algorithm for solving~\eqref{eq: b prob} if all other input parameters are chosen consistently. However, this is not observed in practice.
We can see from Figure~\ref{fig:gauss=plot} that the compared algorithms generate different sequences of iterates whose limit is the unique optimal solution of~\eqref{eq: b prob}. We terminate each algorithm prematurely at a suboptimal solution, which varies based on our choice of algorithm.

To illustrate this phenomena, we recorded
the iterates generated by Algorithm~\ref{alg: BCD} using APG and ADMM with $\mu=2$; we restricted our analyses to these methods to simplify our figures and similar behaviour would be observed using other algorithms for solving~\eqref{eq: b prob}. We sampled balanced training and testing sets of $n=200$ observations of dimension $p=250$ from $N(\bs{\mu_1}, \bs{\Sigma})$ and $N(\bs{\mu_2}, \bs{\Sigma})$, where $\bs{\mu_i}$ is defined according to \eqref{eq:mu-def2} and $\bs{\Sigma}$ is constructed as in the previous sections.
We recorded the iterates $\bb^i_{APG}, \bb^i_{ADMM}$ generated by each of APG and ADMM (with $\mu = 2$) following $i = 10$, $50$, $100$, $500$, $1000$, $2500$, $5000$, $7500$, $10000$ iterations for solving~\eqref{eq: b prob} with regularization parameters $\gamma = 10^{-3}$, $\bs{\Omega} = \I$, and $\lambda = 0.05\bar\lambda$.
We report the norm difference $\|\bb_{APG}^i - \bb_{ADMM}^i\|$, cardinality of $\bb_{APG}^i$ and $\bb_{ADMM}^i$, and out-of-sample classification error for nearest centroid classification following projection onto $\bb_{APG}^i$ and $\bb_{ADMM}^i$ for each value of $i$ in Table~\ref{tab:comparisons}.
We round any entry of a discriminant vector with magnitude less than $10^{-8}$ to $0$ for the purposes of calculating cardinality.
The calculated discriminant vectors are qualitatively similar but still differ significantly, particularly in cardinality.
Specifically, the discriminant vectors generated by ADMM are much sparser than those calculated by APG, particularly early in the iterative process.
This suggests that ADMM is converging to the unique optimal solution of~\eqref{eq: b prob} more quickly than APG for this particular data set and choice of augmented Lagrangian penalty parameter $\mu$; we investigate the sensitivity of ADMM to the choice of $\mu$ further in Section~\ref{sec:ADMMvsAPG}.
However, we should also note that both methods converge to essentially identical solutions within $10000$ iterations. We should also note that all of the discriminant vectors yield classifiers with zero out-of-sample classification errors.




\begin{table}[!htbp]
\begin{center}
\begin{tabular}{c c cc cc} \toprule 
\multirow{2}{*}{Iteration} & Norm & \multicolumn{2}{c}{Cardinality} \\
    &  Difference & APG & ADMM\\\midrule 
    10 & 0.1462 & 248 & 119  \\ 
  50 & 0.2468 & 240 & 78  \\ 
  100 & 0.2854 & 217 & 65  \\ 
  500 & 0.2414 & 103 & 51  \\ 
  1000 & 0.1464 & 70 & 49  \\ 
  2500 & 0.0187 & 52 & 49  \\ 
  5000 & 0.0103 & 49 & 47  \\ 
  7500 & 0.0062 & 49 & 47  \\ 
  10000 & 0.0027 & 48 & 47 \\ \bottomrule
\end{tabular}
\end{center}
\caption{Norm difference $\|\bb_{APG}^i - \bb_{ADMM}^i\|$, and cardinality of $\bb_{APG}^i$ and $\bb_{ADMM}^i$ for
$i = 10$, $50$, $100$, $500$, $1000$, $2500$, $5000$, $7500$, $10000$. No iterates produced out-of-sample classification error for nearest centroid classification following projection onto $\bb_{APG}^i$ and $\bb_{ADMM}^i$. }
\label{tab:comparisons}
\end{table}

\subsection{Scaling Experiments}
\label{sec:scaling}

We next performed a series of simulations to investigate the relationship between the performance of our algorithms and the number of features in the underlying data set.
For each
\[   p \in \{250, 300, \dots,  500,
        600, \dots, 1000,
        1250,  \dots, 2500,
        3000, 3500\},\]
we sample $50$ data sets, containing two classes drawn from the Gaussian distributions as described in Section~\ref{sec:Gauss}.
For each value of $p$, we sample $\ceil{p/10}$ training and testing observations from each class.
Items in each class are sampled from a multivariate Gaussian distribution with means $\bs{\mu_1}, \bs{\mu_2} \in \R^p$ defined by~\eqref{eq:mu-def2}.
Both class distributions have covariance matrix $\bs{\Sigma}$ having diagonal entries equal to $1$ and off-diagonal entries equal to $0.75$.

We apply nearest centroid classification following projection onto the approximate solution of~\eqref{eq: prob} given by the proximal gradient method, accelerated proximal gradient method (with and without backtracking line search) (PG, PGB, APG, APGB), alternating direction method of multipliers (ADMM), and least angle regression method (LARS).
We perform exactly one full iteration of the block coordinate descent method~(Algorithm~\ref{alg: BCD}) for each $\bs{\beta}$-subproblem solver; as before, we should expect Algorithm~\ref{alg: BCD} to terminate after one full iteration, since the optimal choice of $\bs{\theta}$ is obtained in the first iteration (in the absence of numerical error).
We choose the regularization parameters  in~\eqref{eq: prob} to be $\gamma = 10^{-3}$, $\bs{\Omega} = \bs{\I}$, and choose $\lambda$ from $\bar{\lambda}/2^c$, $c \in \{3,2,1,0,-1\}$, where $\bar{\lambda}$ is defined as in~\eqref{eq:barlam} using $5$ fold cross validation. 
We choose the value of $\lambda$ with minimum number of classification errors among those which generated discriminant vectors with at most $0.025 p$ nonzero features. We used the augmented Lagrangian penalty parameter $\mu = 2$ in ADMM.
We use the stopping tolerance $10^{-4}$ and a maximum of $1000$ iterations during the cross validation phase. 
We used backtracking parameters $\bar L = 0.25$ and $\eta = 1.25$ in each run. After $\lambda$ is chosen via cross validation, we solved~\eqref{eq: prob} using the full training set. We terminated the proximal methods when their stopping condition is met with tolerance $10^{-4}/\sqrt{p}$ or $5000$ subproblem iterations have been performed. We use a less strict stopping tolerance than earlier analyses to minimize the number of iterations performed when $p$ is large, and, thus, decreasing overall run-time of the experiment.
 
The LARS heuristic was terminated after a solution was obtained containing $0.25p$ nonzero features or the convergence criterion is satisfied with tolerance $10^{-4}/\sqrt{p}$.
The decision to choose $\lambda$ via cross validation rather than a specific value (as in the previous section) was made to account for the fact that each call of the LARS heuristic calculates a full regularization path for~\eqref{eq: prob}. We perform cross validation to choose $\lambda$ to more accurately compare the computational resources needed by PG, PGB, APG, APGB, ADMM, and LARS to calculate a full regularization path.

\begin{figure}[!htbp]
    \centering
    \begin{subfigure}[b]{0.9\textwidth}
        \includegraphics[width=\textwidth]{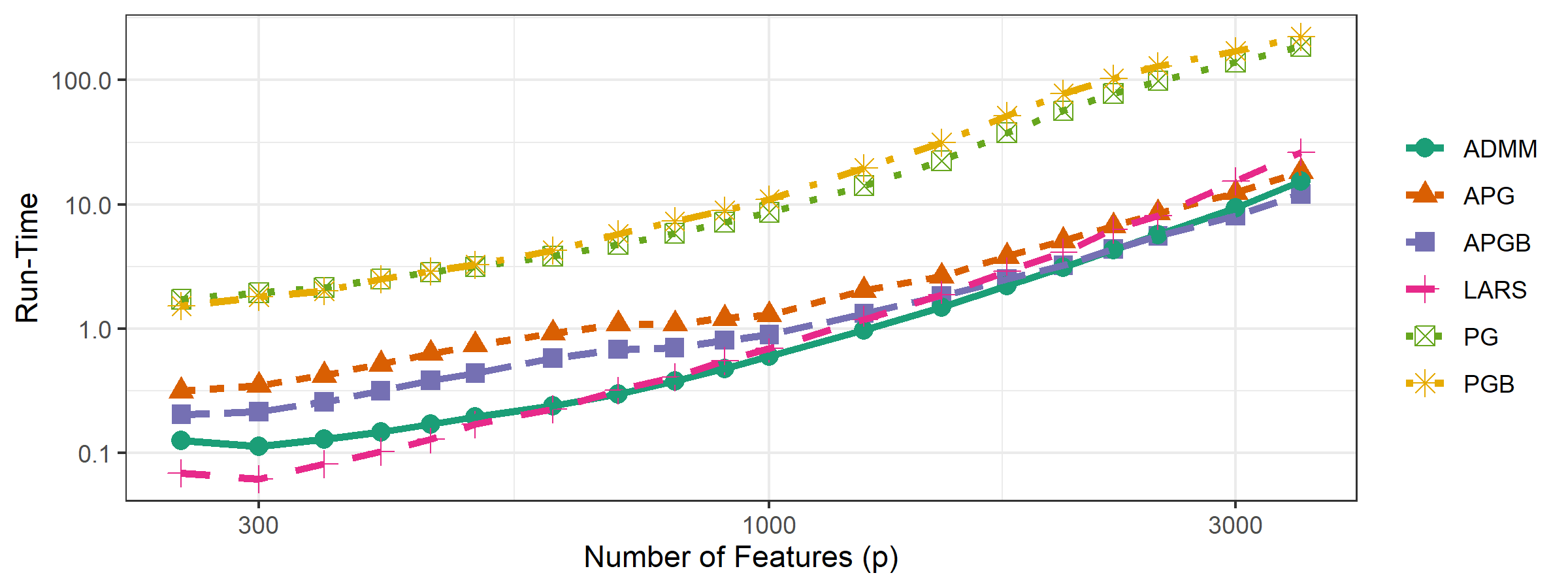}
        \caption{Time}
        \label{fig:scalTime}
    \end{subfigure}

    \begin{subfigure}[b]{0.9\textwidth}
          \includegraphics[width=\textwidth]{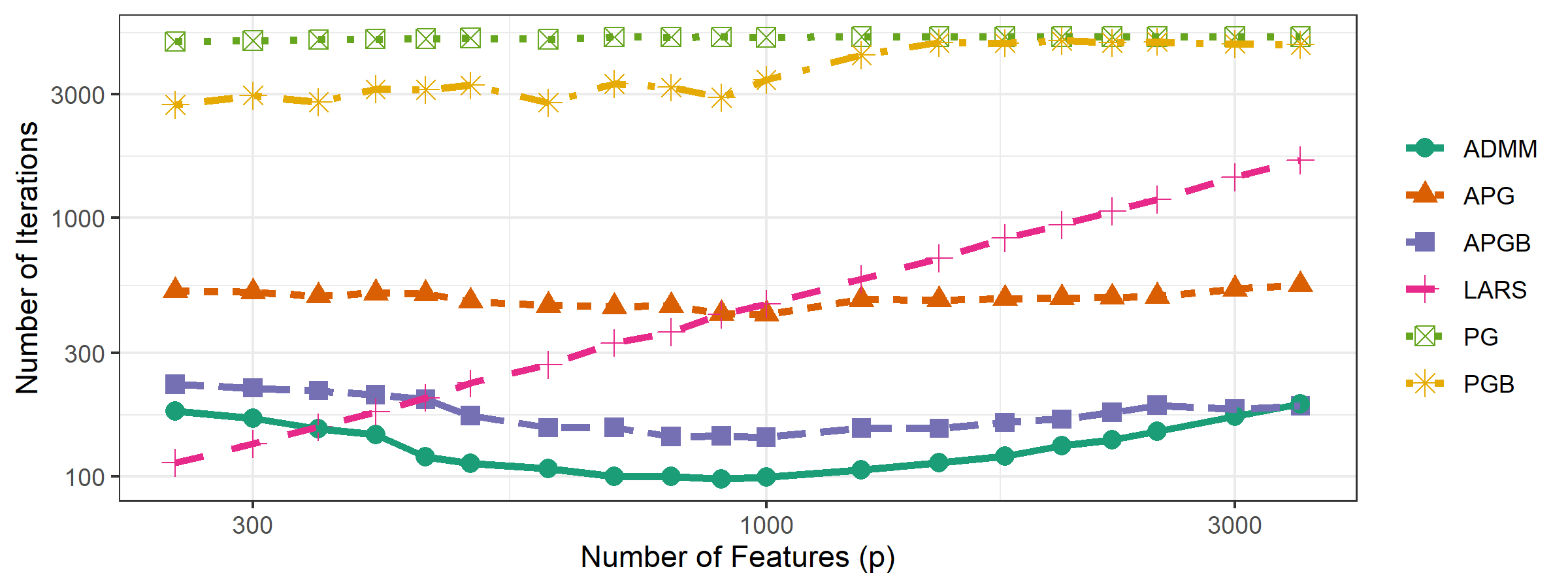}
        \caption{Number of iterations}
        \label{fig:scalIter}
    \end{subfigure}

    \begin{subfigure}[b]{0.9\textwidth}
        \includegraphics[width=\textwidth]{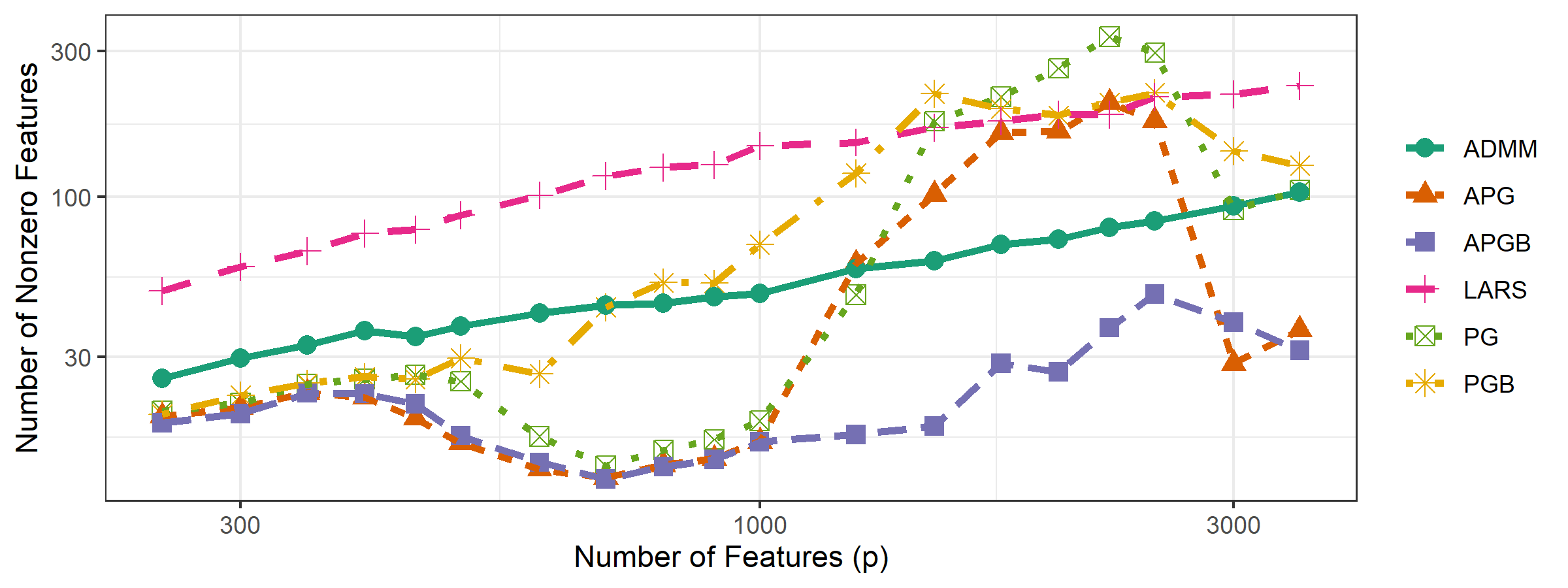}
        \caption{Cardinality}
        \label{fig:scalFeat}
    \end{subfigure}

    \caption{Average run-time (in seconds), number of iterations performed, and cardinality of returned solution plotted as a function of number of features $p$; values of each statistic were averaged across $50$ trials. All axes use logarithmic scale.}\label{fig:scaling}
\end{figure}

Figure~\ref{fig:scaling} summarizes the results of these simulations. We note that the accelerated proximal gradient and alternating direction method of multipliers consistently outperform the traditional LARS method in terms of run-time and number of iterations performed with $p$ is large. 
In particular, these methods require significantly fewer iterations and terminate in less time than LARS for $p > 1000$. The approximate slopes of the plots of average run-times indicate that this phenomena will only be amplified as we increase $p$ further, since the slope of the curve for LARS exceeds that of APG, APGB, and ADMM. This largely agrees with the phenomena predicted by the operation counts discussed in~Section~\ref{sec: comp}.
On the other hand, the proximal gradient methods (PG, PGB) typically do not converge within the maximum number of iterations, which undermines any improvements to computational complexity due to their relatively inexpensive iterations.

We should note that we expect the cardinality of our obtained discriminant vectors to increase as a function of $p$, since the size of the blocks of entries with elevated values in the class-means $\bs{\mu_1}$ and $\bs{\mu_2}$ grows linearly with $p$.
This agrees with the plotted curves in Figure~\ref{fig:scalFeat}.
This also explains the linear increase in number of iterations before termination of the LARS method, since the number of iterations depends on the desired number of nonzero entries; in turn, this, along with increase in per-iteration cost as $p$ increases, explains the increase in total run-time of LARS as $p$ increases. Finally, the cardinality of returned discriminant vectors scales similarly for the four proximal gradient methods (PG, PGB, APG, APGB) and LARS. The discriminant vectors returned by ADMM consistently contain fewer nonzero entries than the four other methods, which agrees with the behaviour observed in the Section~\ref{sec:convtrials}.

This comparison is somewhat unfair to the LARS heuristic since the number of nonzero discriminant vector features, which should correlate with the number of large magnitude entries of the difference of means $\bs{\mu}_1 - \bs{\mu}_2$, is increasing linearly with $p$. Thus, we seek relatively dense discriminant vectors, scaling linearly with $p$. This is exactly the situation where we should expect ADMM and APG to be more efficient than LARS according to~\eqref{eq:APG-best}.

However, if the number of nonzero predictor variables is constant or grows relatively slowly as $p$ increases, we should expect LARS to be more efficient than APG and ADMM (again, according~\eqref{eq:APG-best}).
To confirm this empirically, we repeated this experiment but chose the mean vectors $\bs{\mu}_1$ and $\bs{\mu}_2$ to differ from zero in the first $100$ and second $100$ features, respectively, for all values of $p$. We stopped the LARS algorithm when an iterate with at least $200$ nonzero features was found or stopping tolerance $10^{-4}/\sqrt{p}$ is met; we then chose the iterate with minimum out-of-sample classification rate for our discriminant vector. All other experimental settings were kept identical to that in the previous discussion. We omit the unaccelerated proximal gradient methods PG/PGB since the previous analysis established that they are significantly less than the other methods under comparison.

\begin{figure}
    \centering
    \begin{subfigure}[b]{0.9\textwidth}
        \includegraphics[width=\textwidth]{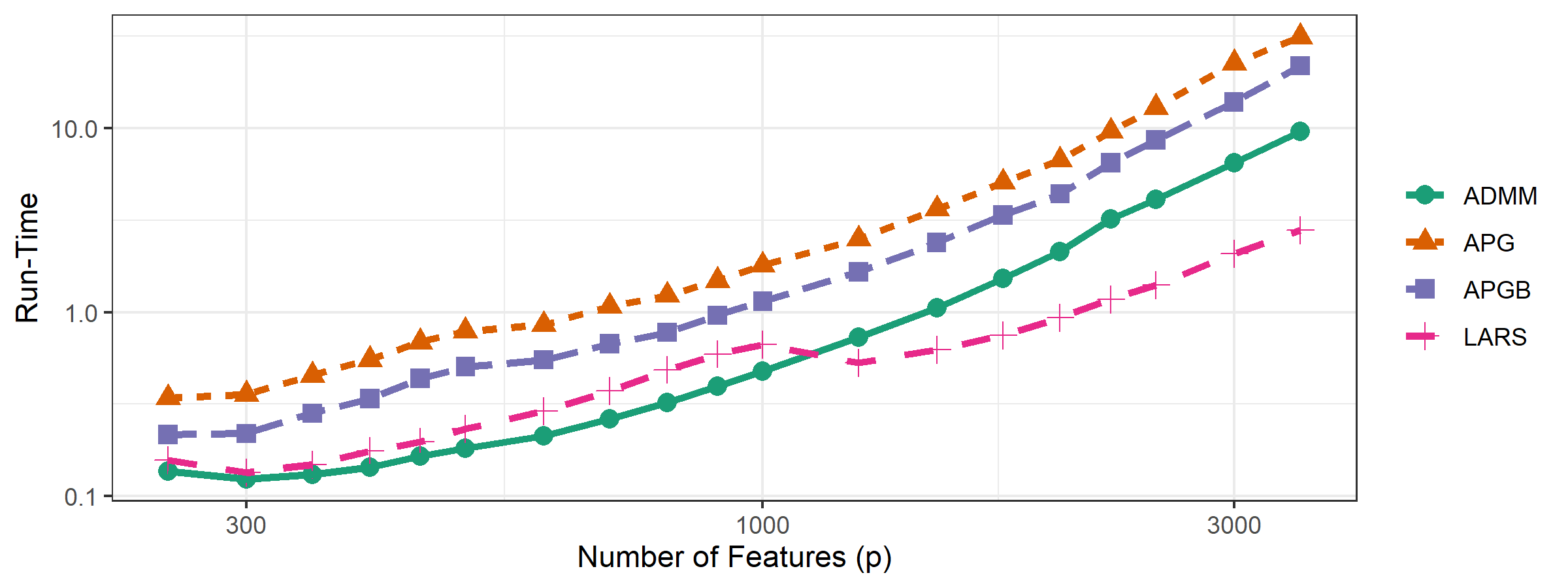}
        \caption{Time}
        \label{fig:spscalTime}
    \end{subfigure}

    \begin{subfigure}[b]{0.9\textwidth}
          \includegraphics[width=\textwidth]{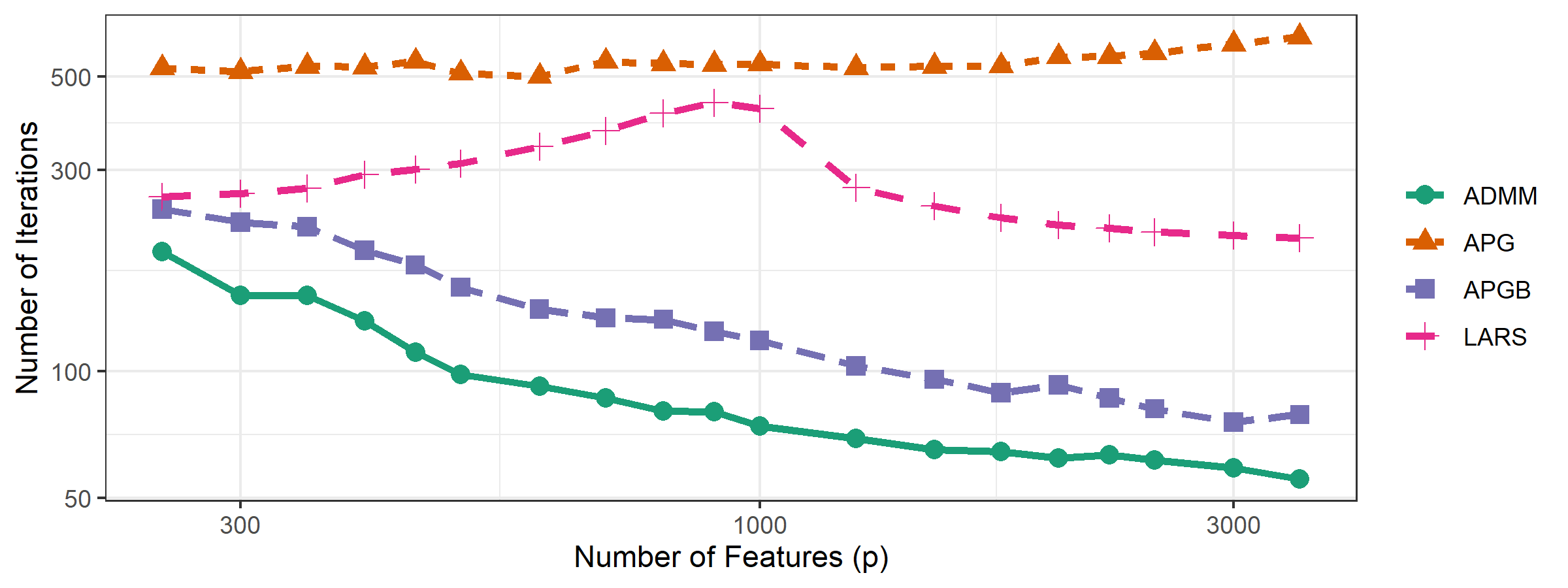}
        \caption{Number of iterations}
        \label{fig:SPscalIter}
    \end{subfigure}

    \begin{subfigure}[b]{0.9\textwidth}
        \includegraphics[width=\textwidth]{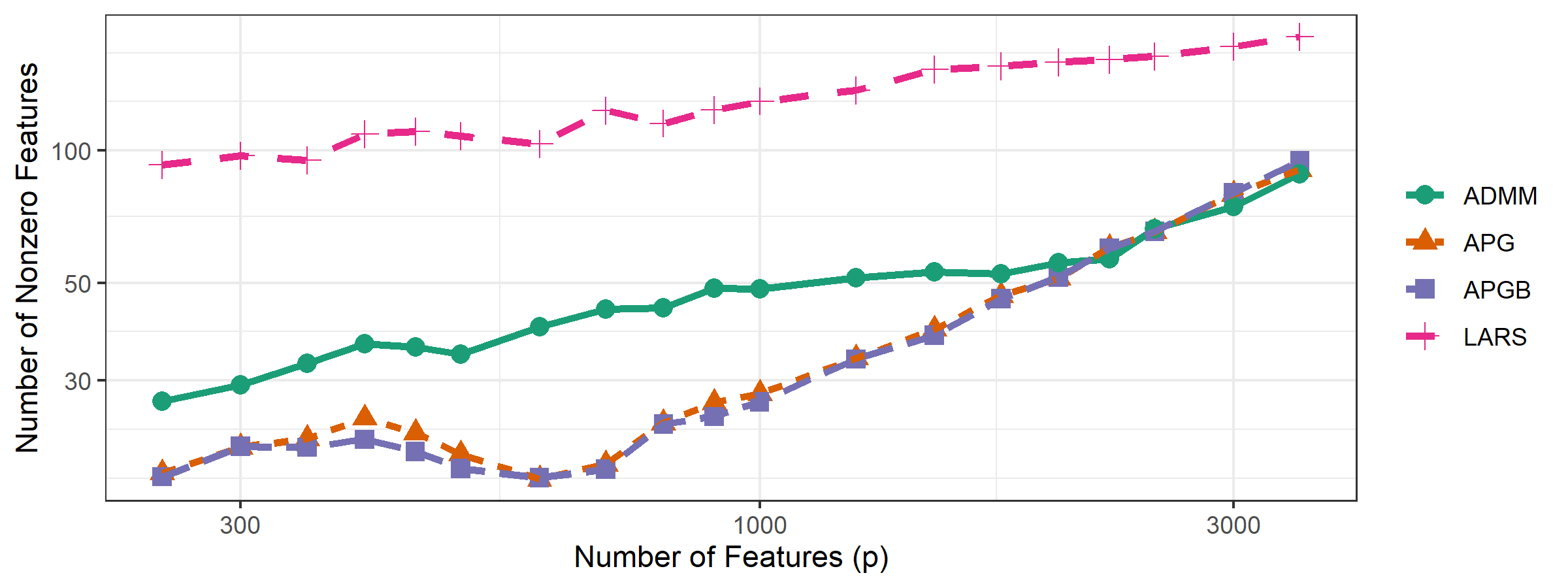}
        \caption{Cardinality}
        \label{fig:SPscalFeat}
    \end{subfigure}
    \caption{Average run-time (in seconds), number of iterations performed, and cardinality of returned solution for SOS problem with fixed number of nonzero entries of $\bs{\mu}_1 - \bs{\mu}_2$ ($200$) plotted as a function of number of features $p$. Values of each statistic were averaged across $50$ trials. All axes use logarithmic scale.}
    \label{fig:sparse-scaling}
\end{figure}

The results of this trial can be found in Figure~\ref{fig:sparse-scaling}. 
As expected, LARS tends to more efficient than APG and ADMM in this setting.
This suggests that one must be careful to consider problem dimension and the underlying properties of the desired classifiers (e.g., sparsity) when choosing a heuristic for solving~\eqref{eq: b prob}. Specifically, if $p$ is small or we seek a sparse discriminant vector with relatively few nonzero entries, then we should use LARS to solve~\eqref{eq: b prob}; otherwise, we should favor APG or ADMM.

\subsection{Classification of Real-World Data} \label{sec:real-data}

We performed similar analyses using data sets drawn from the UC Riverside Time-Series Clustering and Classification data repository~\cite{keogh2006ucr} to verify that the behaviour observed with synthetic data is also observed when classifying real-world data.
We applied each of the methods APG, ADMM, SZVD, and LARS to learn classification rules for each of the data sets in the UCR repository with number of training samples $n$ less than the number of predictive features $p$; this yielded a collection of $63$ data sets to analyze. We omit PG and the backtracking methods from this analysis because our analysis of synthetic data established that these methods are typically less effective than the remaining approaches.
We use each remaining sparse discriminant analysis heuristic to obtain $q = K-1$ sparse discriminant vectors
and then perform nearest-centroid classification after projection onto the subspace spanned
by these discriminant vectors.

In all experiments,
we set $\gamma = 10^{-3}$ and $\bs\Omega$ to be the
$p\times p$ identity matrix $\bs\Omega = \bs I$.
We choose $\lambda$ using $5$ fold cross validation from the set
of potential $\lambda$ of the form~$ \bar\lambda/2^{c}$ for $c=3,2,1,0,-1$
with $\bar\lambda$ defined by~\eqref{eq:barlam}
to be the value of $\lambda$ with fewest average number
of misclassification errors over training-validation splits
amongst all $\lambda$ which yield discriminant
vectors containing at most $25\%$ nonzero entries.
During the cross validation stage, we terminate each proximal algorithm in the inner loop after $1500$ iterations
or a $10^{-4}$ suboptimal solution is obtained.
After $\lambda$ is chosen via cross validation we solve~\eqref{eq: prob} using the full training data set. Here we terminate each proximal algorithm in the inner loop after $2000$ iterations or a $10^{-6}$ suboptimal solution is found and the outer loop
is stopped after one iteration if $K = 2$ or a maximum number of $250$ iterations
or a $10^{-3}$ suboptimal solution has been found otherwise.
The augmented Lagrangian parameter $\mu = 2$ was used in ADMM.

We calculate the full regularization path for $\lambda$ using the LARS algorithm, choosing the set of discriminant vectors from the full path with maximum in-sample classification accuracy\footnote{This choice of method of tuning~$\lambda$ differs from that used in earlier versions of this manuscript, where we chose $\lambda$ via cross-validation or based on out-of-sample classification rate. The results of this set of analyses largely agree with those of the earlier manuscripts, except with significantly decreased run-times for LARS when compared to the approach applying cross-validation, and modestly increased misclassification error when compared to those trained using out-of-sample accuracy.}.
We terminate each call to LARS in the inner loop after a solution with $0.25p$ nonzero entries is found, or stopping tolerance $10^{-6}$ is met, or $3000$ iterations are performed.

We use the augmented Lagrangian parameter $\beta = 5$ and choose the regularization parameter $\gamma$ in SZVD
from the exponentially spaced grid  $\bar{\gamma}/2^c$ for $c = 3,2,1,0,-1, $ with $\bar\gamma$ defined by~\eqref{eq:bargam}
using $5$ fold cross-validation;
$\gamma$ is chosen to minimize average validation set misclassification error amongst
all sets of discriminant vectors with at most $35\%$ nonzero entries.
We stop SZVD after a maximum of 250 iterations or a solution satisfying the stopping
tolerance of $10^{-5}$ is obtained.
These parameters were chosen experimentally to ensure that all methods converge for each data set in the benchmarking set. It is likely that some variation in performance of the heuristics across data sets could be eliminated by more carefully tuning parameters for each individual data set. We assigned the same choice of parameters for each experiment to avoid having to tune parameters separately for all 63 data sets.

\begin{figure}[!htbp]
  \centering
  \includegraphics[width=\textwidth]{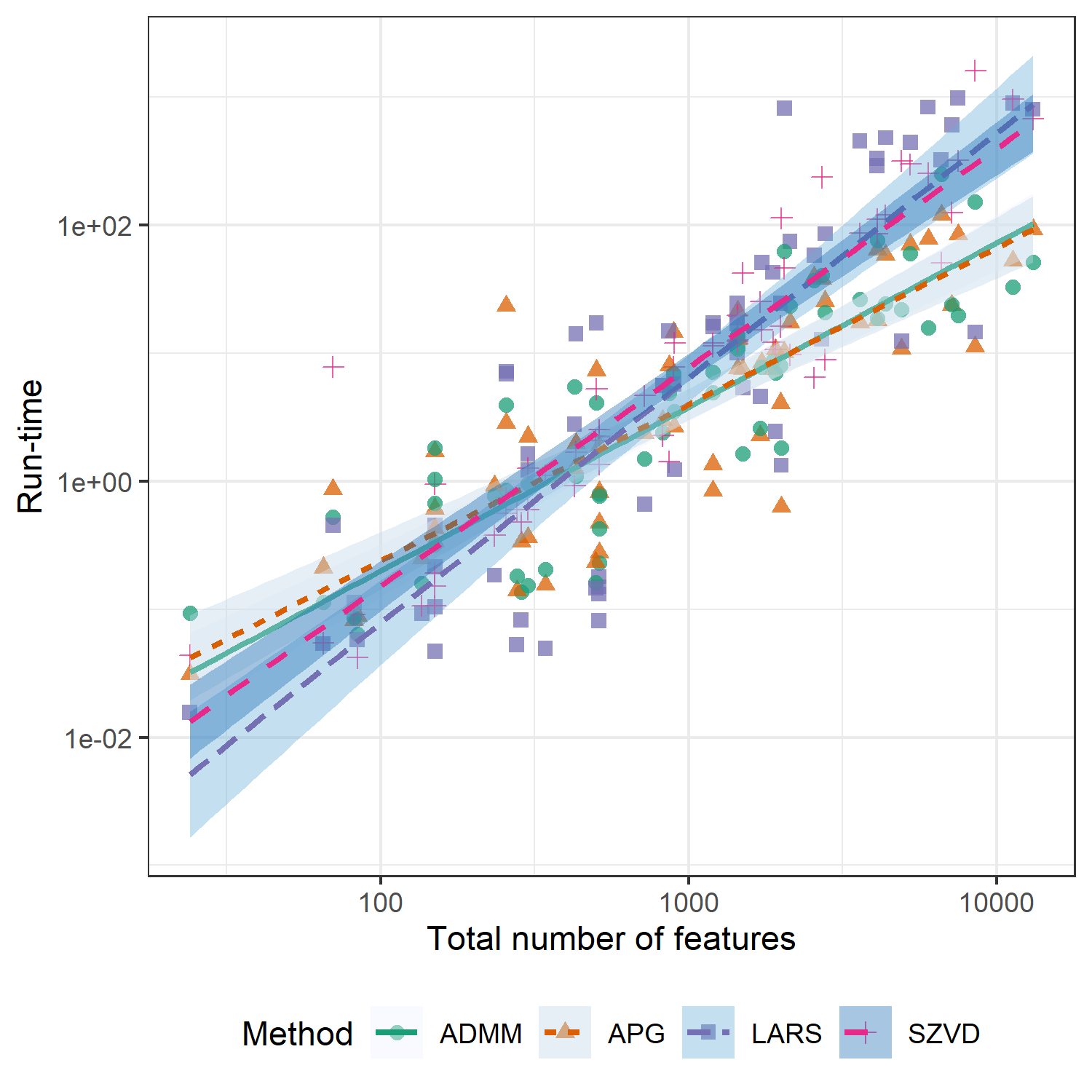}
  \caption{Plot of run-time in seconds of each sparse discriminant heuristic as a function of the total number of predictive features ($p$ times number of discriminant vectors). We also fit a line to the set of run-times for each method and $95\%$-confidence intervals for these linear models. Both axes use logarithmic scale.}
  \label{fig:realtime}
\end{figure}

\renewcommand{\time}{\mathrm{time}}
\begin{figure}[!htbp]
  \centering
  \begin{subfigure}[b]{0.48\textwidth}
    \includegraphics[width=\textwidth]{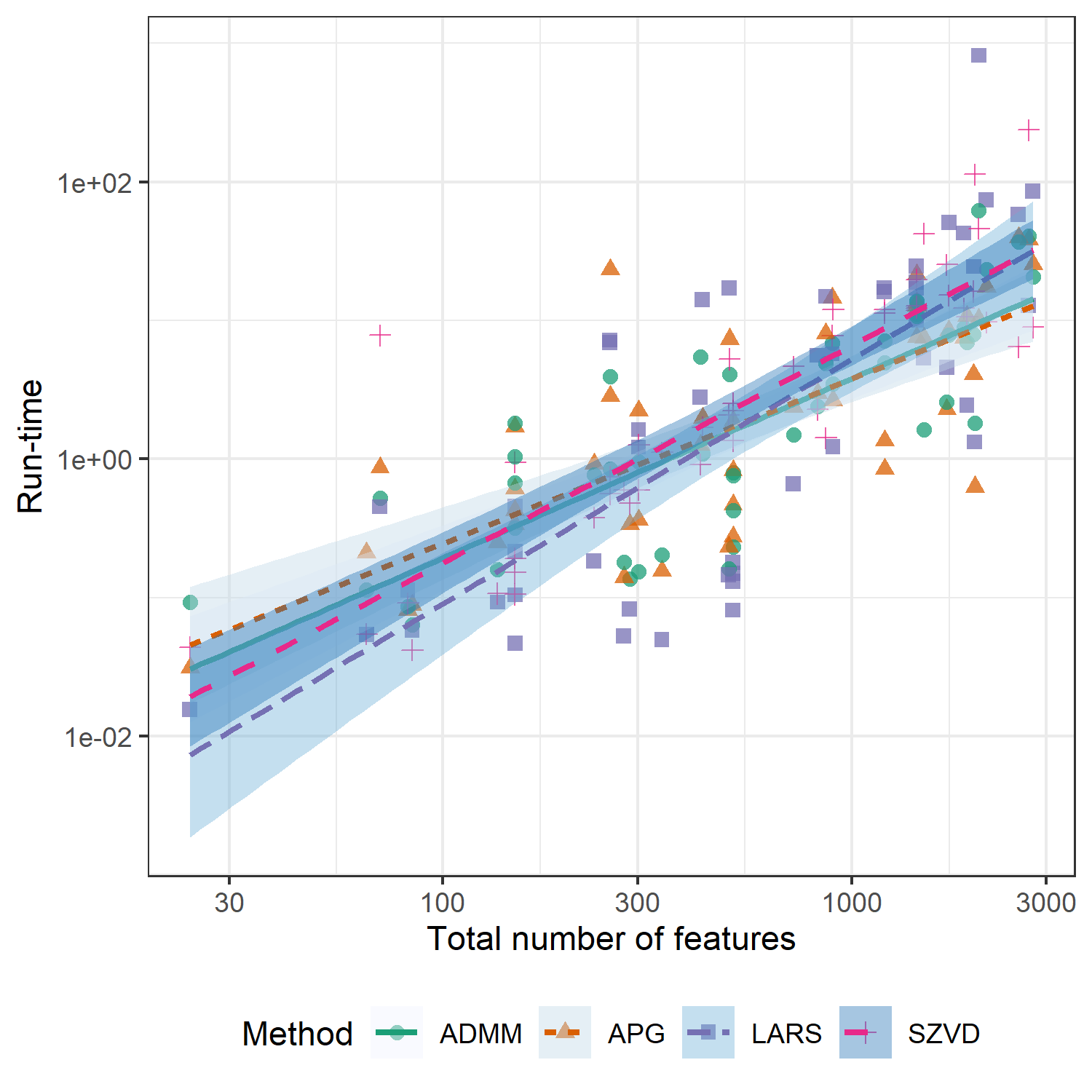}
    \caption{Run-times for Small-scale Data}
    \label{fig:smallRT}
  \end{subfigure}
  ~
  \begin{subfigure}[b]{0.48\textwidth}
    \includegraphics[width=\textwidth]{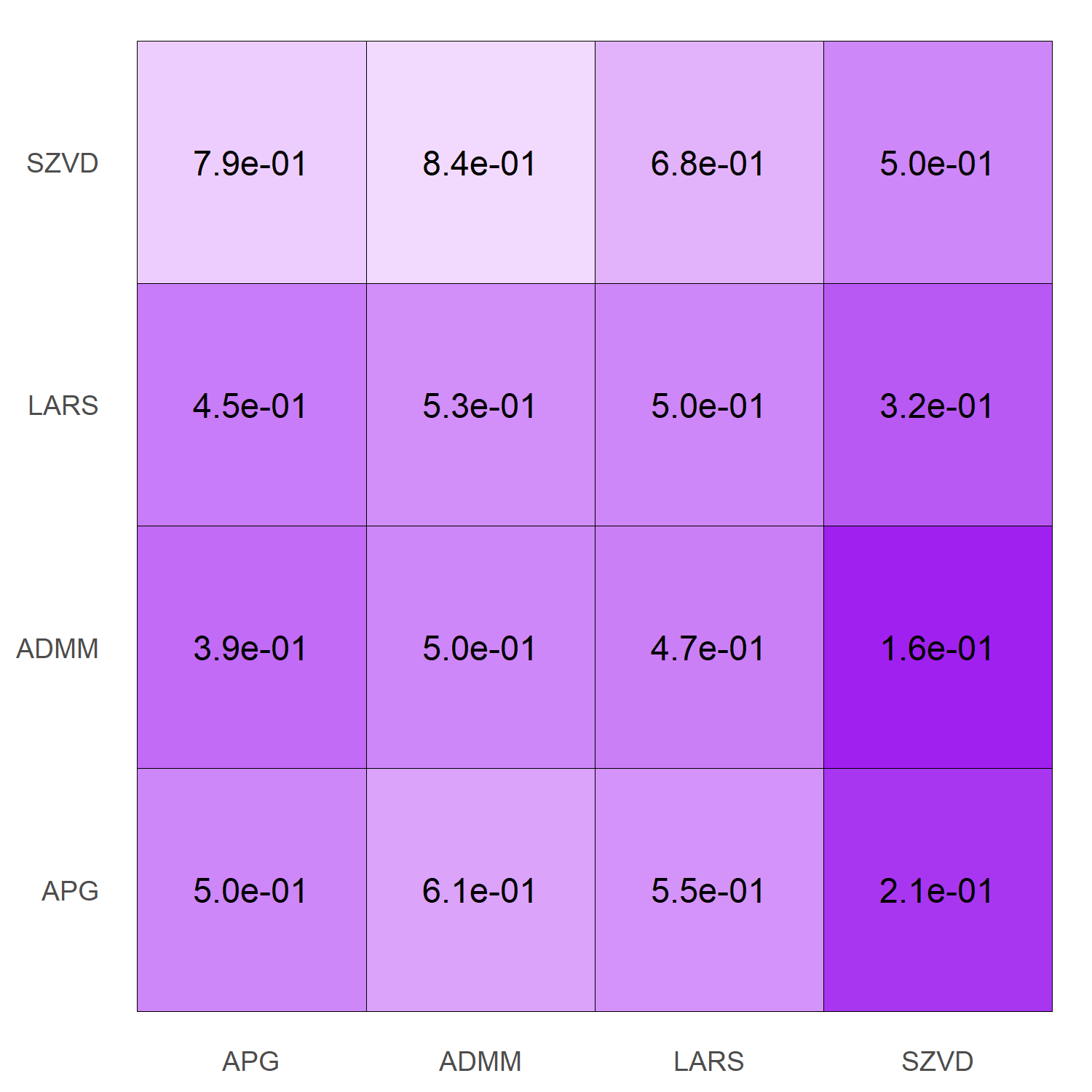}
    \caption{Hypothesis Tests for Small-scale Data}
    \label{fig:p-timeSmall}
  \end{subfigure}

  \begin{subfigure}[b]{0.48\textwidth}
    \includegraphics[width=\textwidth]{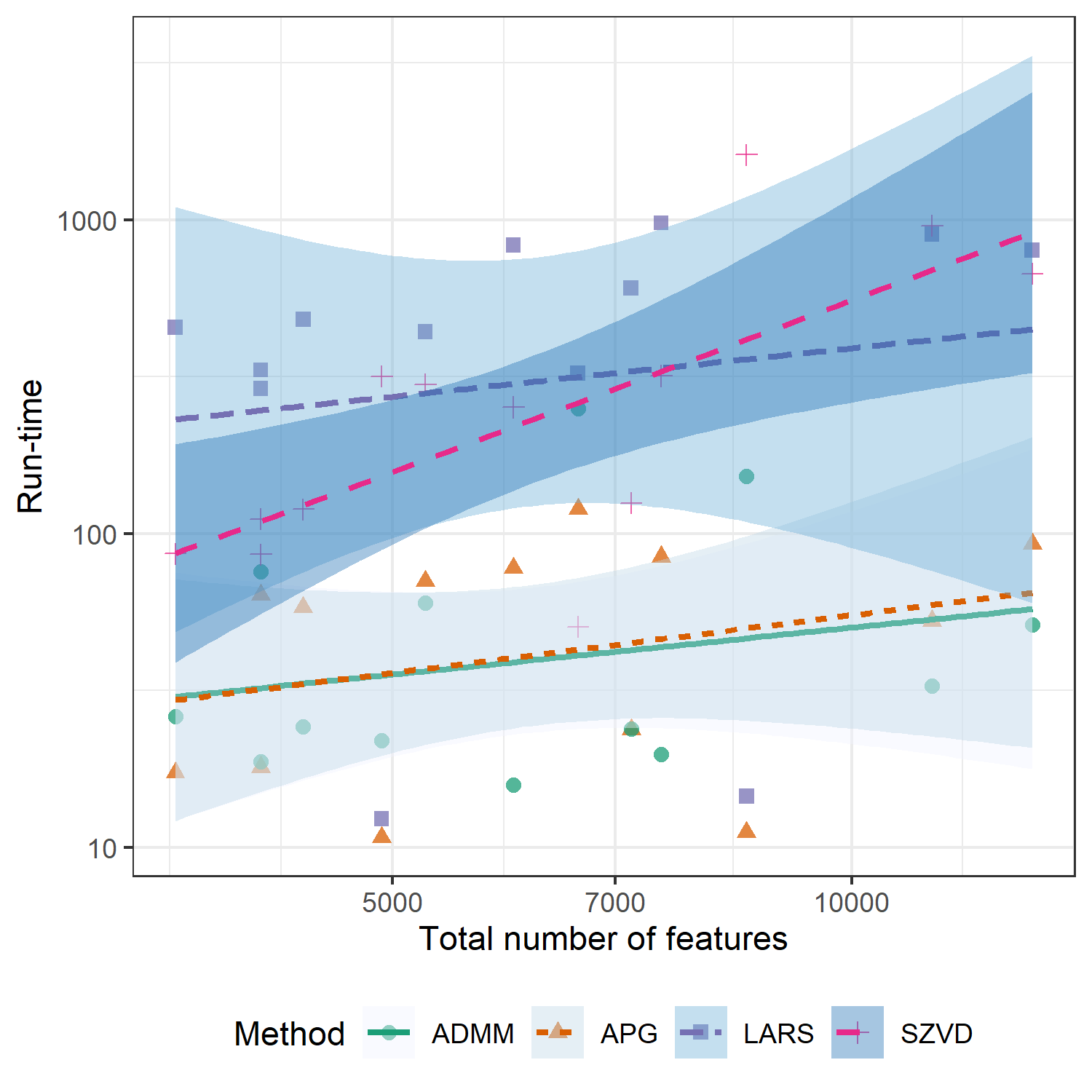}
    \caption{Run-times for Large-scale Data}
    \label{fig:rankingsTime}
  \end{subfigure}
  ~
  \begin{subfigure}[b]{0.48\textwidth}
    \includegraphics[width=\textwidth]{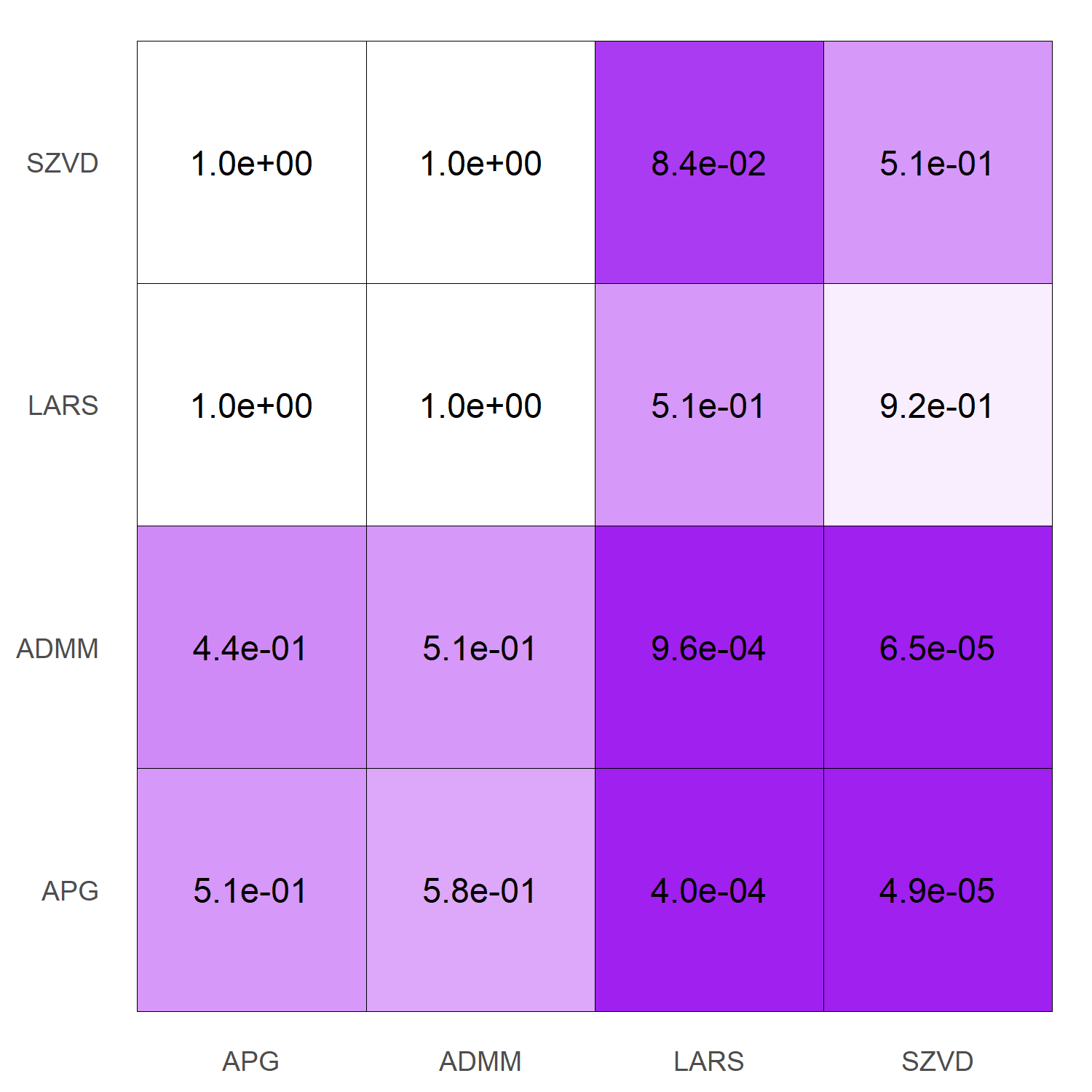}
    \caption{Hypothesis Tests for Large-scale Data}
    \label{fig:p-timeBig}
  \end{subfigure}
  \caption{Plots of run-time for small-scale (total number of features less than 3000) and larger-scale data (total number of features exceeding 3000) included in the UCR benchmarking repository.
  We also include the results of from one-sided Wilcoxon signed-rank tests  comparing computational efficiency of each pair of sparse discriminant analysis heuristics for each subset of benchmarking data.
  The $(i,j)$ box represents the observed $p$-value for the test with null hypothesis $H_0: \time(i) \ge \time(j)$ and alternative hypothesis $H_a: \time(i) < \time(j)$ where $\time(x)$ denotes the expected total run-time of heuristic $x$ on a given data set; darker colors correspond to smaller $p$-values or higher significance. }
  \label{fig:realtime-split}
\end{figure}

\newcommand{\err}{\mathrm{err}}
\begin{figure}[!htbp]
  \centering
  \begin{subfigure}[b]{0.48\textwidth}
      \includegraphics[width=\textwidth]{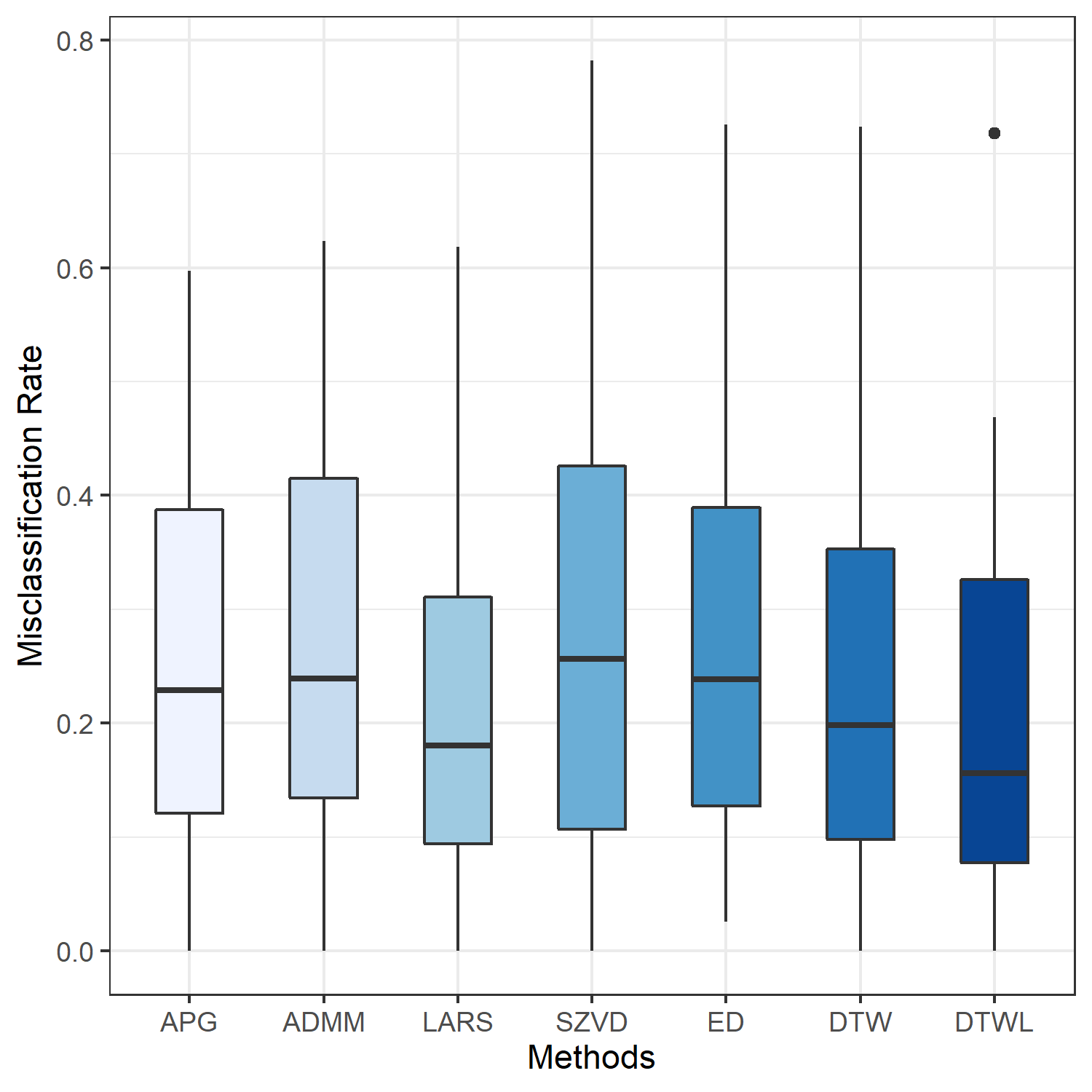}
      \caption{Misclassification Rate}
      \label{fig:boxAcc}
  \end{subfigure}
  ~
  \begin{subfigure}[b]{0.48\textwidth}
    \includegraphics[width=\textwidth]{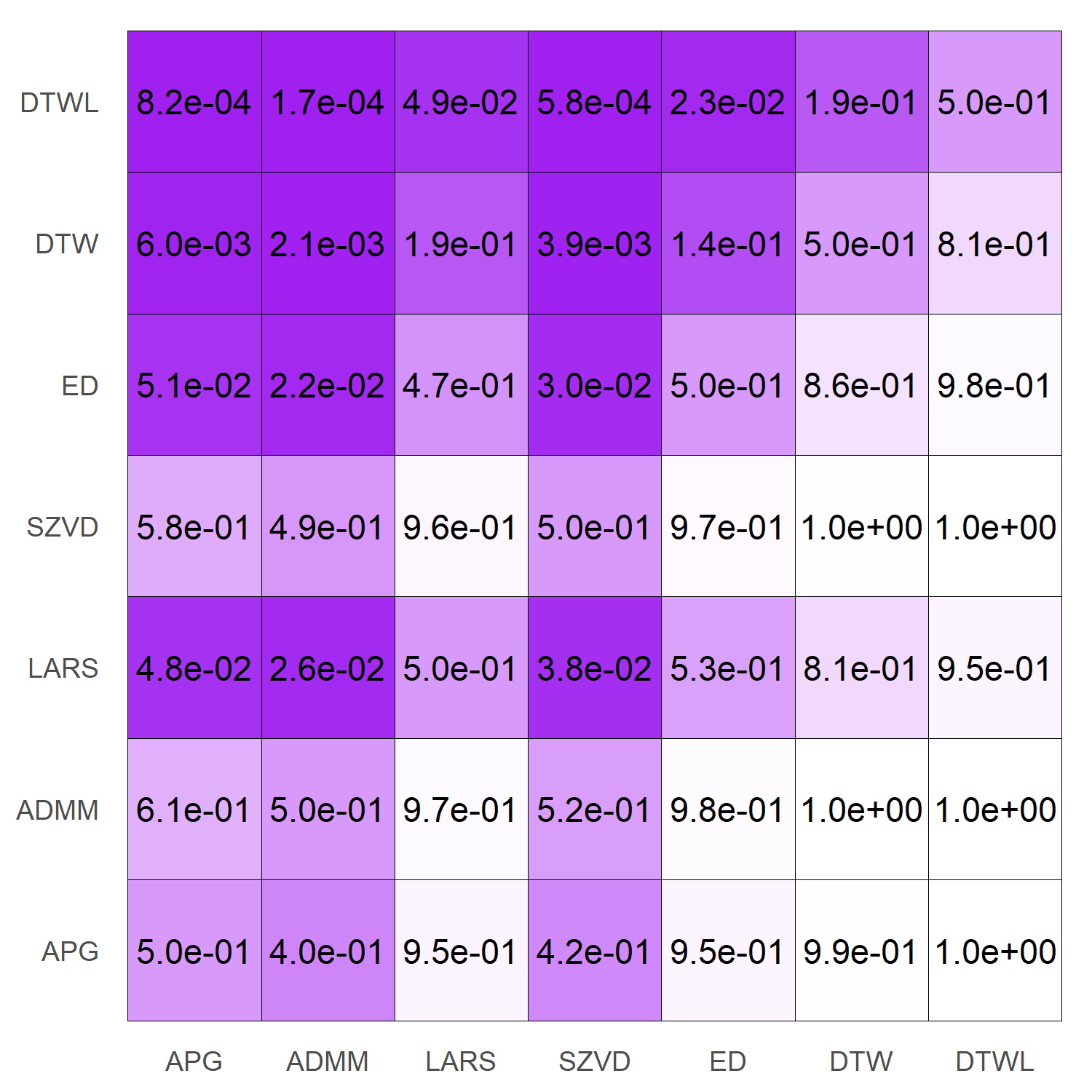}
    \caption{Results of Hypothesis Tests for Accuracy}
    \label{fig:p-accuracy}
  \end{subfigure}

  \caption{Box plots of out-of-sample misclassification rates. We also include box plots for misclassification rate for nearest neighbor classification using Euclidean distance (ED) and Dynamic Time Warping distance  with fixed warping constraint parameter $w = 100$ (DTW), and learned $w$ (DTWL).
  We also plot results of one-sided Wilcoxon signed-rank tests for misclassification rate. The $(i,j)$ box represents the observed $p$-value for the test with null hypothesis $H_0: \err(i) \ge \err(j)$ and alternative hypothesis $H_a: \err(i) < \err(j)$ where $\err(x)$ denotes the expected fraction of misclassified test observations by classification heuristic $x$.
  }
  \label{fig:real-accuracy}
\end{figure}

\newcommand{\card}{\mathrm{card}}
\begin{figure}
  \centering
  \begin{subfigure}[b]{0.48\textwidth}
      \includegraphics[width=\textwidth]{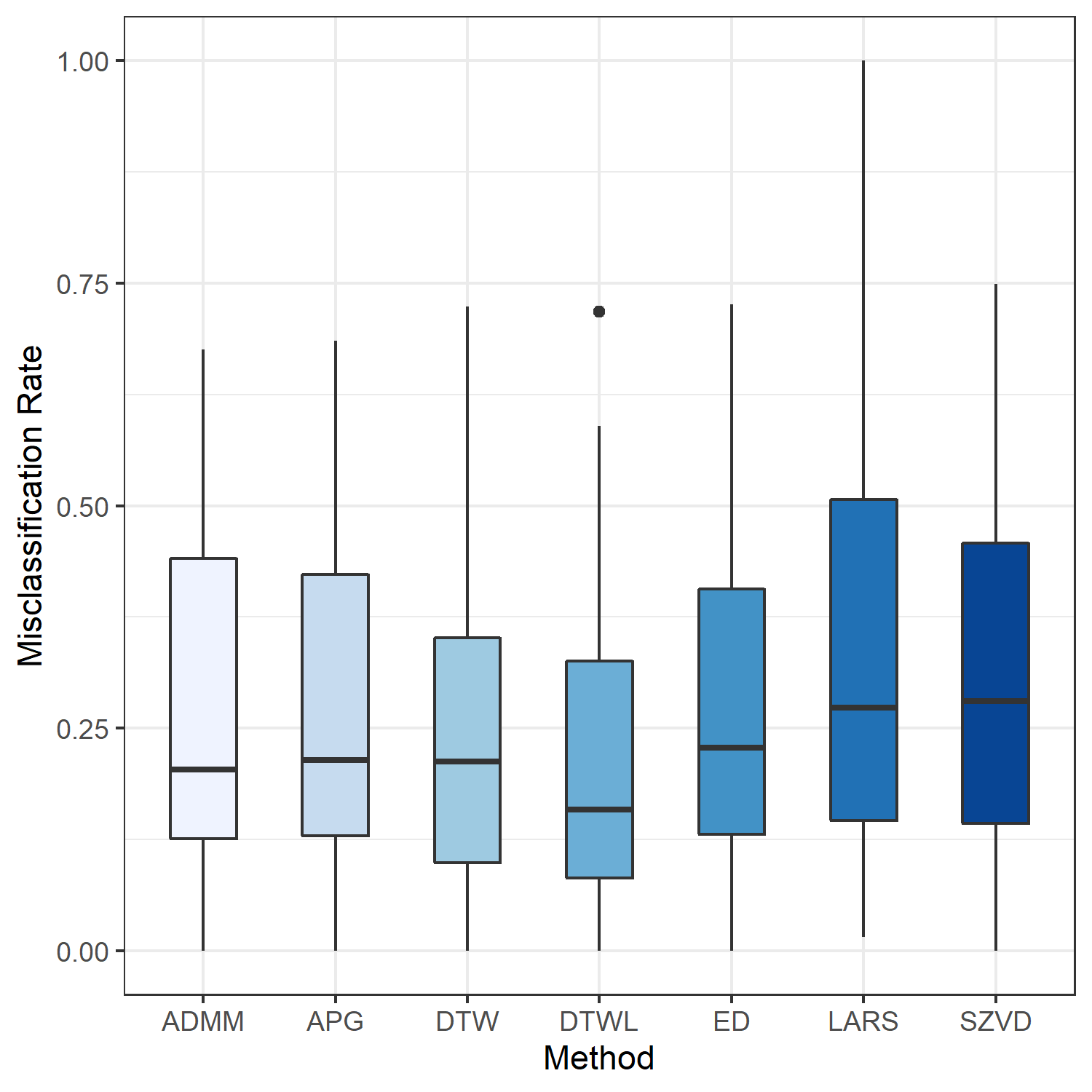}
    \caption{Misclassification Rate with Data Sets Omitted}
    \label{fig:boxSparsity}
  \end{subfigure}
  ~ 
  \begin{subfigure}[b]{0.48\textwidth}
      \includegraphics[width=\textwidth]{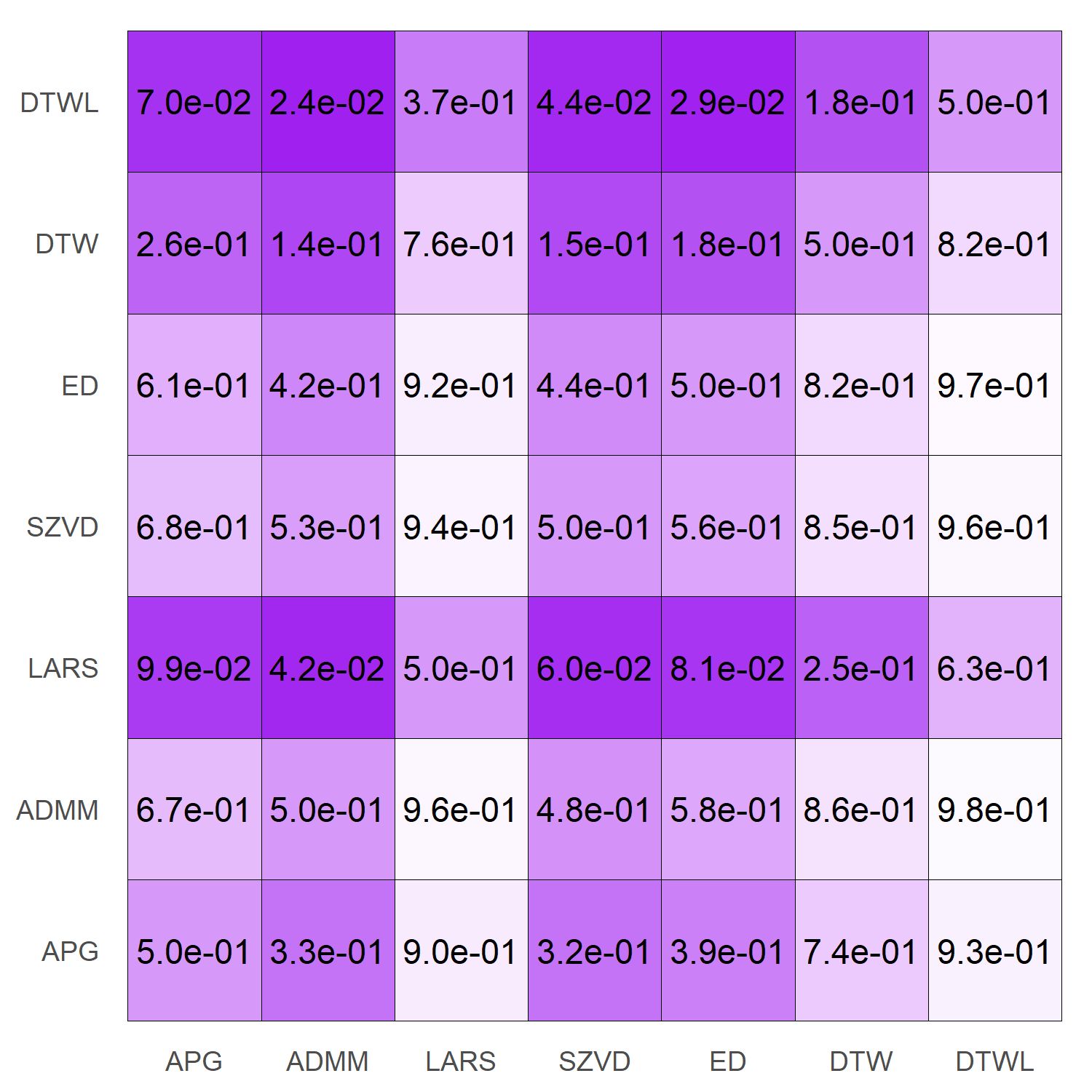}
      \caption{Results of Hypothesis Tests for Accuracy with Data Sets Omitted}
      \label{fig:p-sparsity}
  \end{subfigure}
  \caption{Box plots and attained significance/$p$-values of out-of-sample misclassification rates with data sets with error at least $70\%$ omitted. Here, we observe no significance difference in classification error ($p$-value less than $0.05$) between APG/ADMM classifiers and DTWL.
  }
  \label{fig:real-acc-drops}
\end{figure}

\subsubsection{Comparison of Classification Accuracy}

To empirically test accuracy of each proposed classification heuristic, we calculated the out-of-sample misclassification rate for each data set in the UCR repository. We include baseline accuracies based on the classification results of 1-Nearest Neighbor classifiers: each test observation was assigned the class label of its nearest training observation. As a measure of similarity, we use both \emph{Euclidean distance} (ED) and \emph{dynamic time warping distance} with fixed warping window width $w = 100$ (DTW) and learned window width (DTWL). The out-of-sample misclassification rate for each of these classifiers is provided  by the UCR Time Series Archive; we direct the reader to refer to \cite[Section II]{dau2019ucr} for further details.

For each ordered pair of classification heuristics, we perform a one-sided Wilcoxon signed-rank test. Specifically, for each pair of classification heuristics $(i,j)$ we perform a one-sided Wilcoxon test to test the null hypothesis $H_0: \err(i) = \err(j)$ against the alternative hypothesis $H_a: \err(i) < \err(j)$, where $\err(x)$ denotes the population average misclassification error rate for classifier $x$.
Figure~\ref{fig:real-accuracy} provides a box plot visualizing average misclassification rate for each method, as well as a table of $p$-values for the one-sided Wilcoxon significance tests.
We observe a significant difference between our APG and ADMM classifiers
and nearest neighbors classifiers using dynamic time warping distance with learned window width (DTWL) and fixed length.
We  also observe evidence that APG and ADMM are, on average, less accurate than the LARS classifier ($p$-values $0.048$ and $0.026$, respectively).
On the other hand, the results of these hypothesis tests suggests a significant improvement in accuracy when using the DTWL classifier over all classifiers except the DTW classifier, and all methods provide improved accuracy over SZVD.

The observed differences between the accuracy of nearest neighbors classifiers, particularly those using dynamic time warping distances, and SOS classifiers can be partially explained by the extremely poor accuracy of SOS classifiers for a limited number of data sets.
Linear discriminant analysis-based classifiers are only applicable under the assumption that data is linearly separable following projection onto a lower dimensional subspace and that data from all classes are sampled from distributions with shared covariance matrix. The poor accuracy of the SOS classifiers suggest that these assumptions are not satisfied by this subset of the benchmarking repository.
If we omit the data sets in the UCR repository for which the APG classifier yields a misclassification rate of at least $60\%$, we obtain the average misclassification rates and attained significance visualized in Figure~\ref{fig:real-acc-drops}.
After restricting the benchmarking data set in this way, we observe no significant difference between the APG-trained classifiers and the nearest neighbor classifier DTWL or the LARS-trained classifier at a significance level of $p < 0.05$.

We note that we do not consider this reduced benchmarking set in an attempt to overstate the classification accuracy of the proposed methods APG and ADMM relative to DTW-based nearest neighbors methods. Instead, we want to emphasize that the average difference in accuracy is due to the relatively poor performance of SOS models for a modest number of problems, for which linear classifiers are possibly ill-suited, rather than the SOS models performing poorly on average over all benchmarking data sets.
This suggests that SOS classifiers, including those trained using the LARS algorithm, exhibit comparable classification performance to the baseline provided by nearest neighbors classifiers when restricted to instances where their use is appropriate, i.e., their underlying statistical assumptions are met.

\subsubsection{Comparison of Run-Times}
In terms of computational complexity, we can observe two general trends: LARS is consistently more efficient than APG and ADMM when the total number of predictor variables, $qp$, is small (say, $qp < 1000$), and APG/ADMM are significantly more efficient than LARS when the number of features is moderate to large ($qp \ge `000$). Figure~\ref{fig:realtime} plots the total run-time of each sparse discriminant heuristic (including training of parameters by cross-validation), along with linear models fit to observed run-times. There are clear bifurcation points between $500 <qp < 1000$, where the linear models for run-time of APG and ADMM cross that of LARS. To investigate this phenomena further, we isolated run-times for data sets with $qp < 3000$ and $qp \ge 3000$ and performed one-sided Wilcoxon tests for the null hypothesis $H_0: \time(i) = \time(j)$ and alternative hypothesis $H_a: \time(i) < \time(j)$ under both settings.
The results of these significance tests can be found in Figure~\ref{fig:realtime-split}, along with plots of run-times.
These tests strongly suggest that
both APG and ADMM require significantly less computation than LARS when the total number of predictor variables is greater than $3000$: we observe $p$ values on the order of $10^{-4}$ when testing $H_0: \time(\text{APG}) = \time(\text{LARS})$ and $H_0: \time(\text{ADMM}) = \time(\text{LARS})$.
On the other hand, we see modest evidence that
LARS is more efficient than both ADMM and APG when the total number of predictor variables is less than $3000$, but not at a statistically significant level.
We should also note that the computational cost of the LARS algorithm is at least partially inflated by the fact that we seek somewhat dense discriminant vectors (with terminating cardinality $0.25p$), although LARS regularly terminated with sparse optimal solutions, obtained in fewer than $0.25p$ iterations.
Finally, all methods were more efficient than the SZVD method.
This scaling of computational cost largely agrees with that observed for synthetic data in Section~\ref{sec:scaling} and predicted by~\eqref{eq:APG-best}.

We should note that we do not include nearest neighbor classifiers in this discussion of computational efficiency, although we used these methods to obtain a baseline accuracy to benchmark our proposed methods against. The accuracies of these classifiers were provided with the UCR repository and, thus, did not require retraining of the classifiers. Naive implementation of nearest neighbors methods requires at least $\O(n^2p)$ operations to calculate pairwise Euclidean distances and $\O(n^2p^2)$ flops to calculate DTW distances (plus additional operations for training the window width $w$); optimized methods for DTW reduce this computational cost to $\O(n^2p)$ flops. Therefore, we expect our methods to scale at least as well as these nearest neighbors methods.


	\subsection{Multispectral X-ray images and $\bs\Omega$ of varying rank}
	\label{sec: multispec}

	To demonstrate the improvement in run-time obtained by using a low-rank $\bs\Omega$ in the elastic-net penalty, we perform pixelwise classification on multispectral X-ray images, as presented in~\cite{einarsson2017foreign}. The multispectral X-ray images are scans of food items, where each pixel contains 128 measurements (channels) corresponding to attenuation of X-rays emitted at different wavelengths (see Figure~\ref{fig:meats}). The measurements in each pixel thus give us a profile for the material positioned at that pixel's location (see Figure~\ref{fig:profile}).

  \begin{figure}[!htbp]
  		\centering
  		\includegraphics[width=0.24\columnwidth]{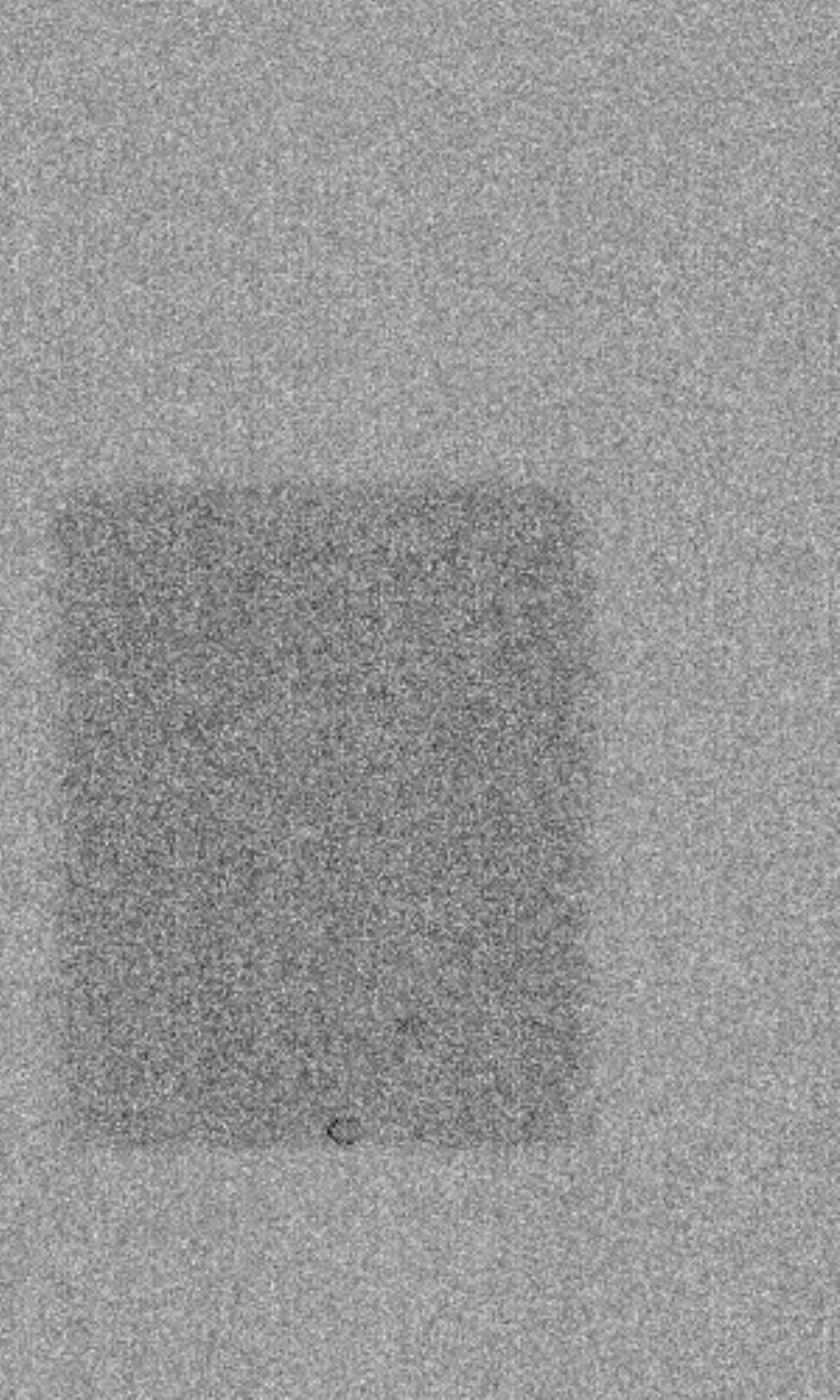}
  		\includegraphics[width=0.24\columnwidth]{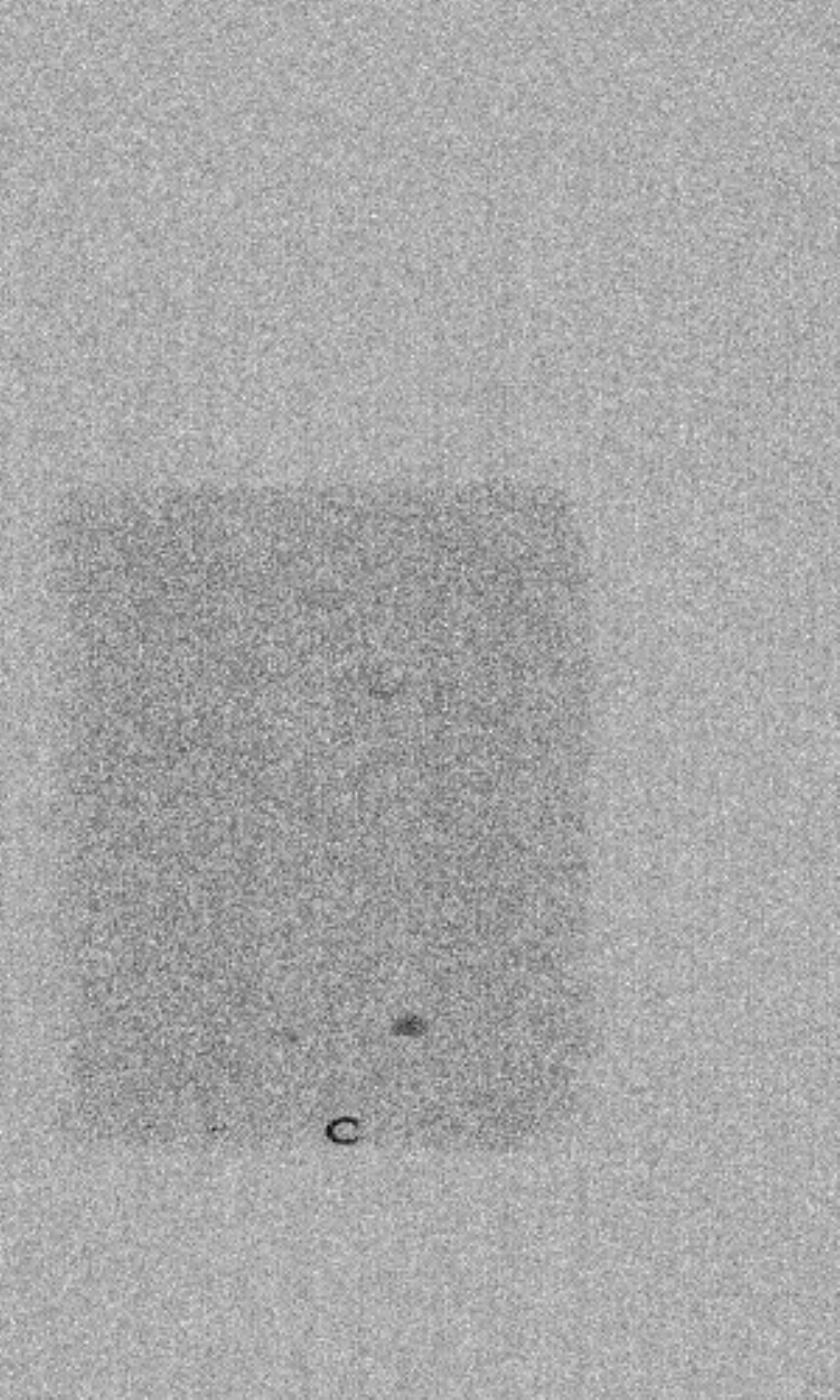}
  		\includegraphics[width=0.24\columnwidth]{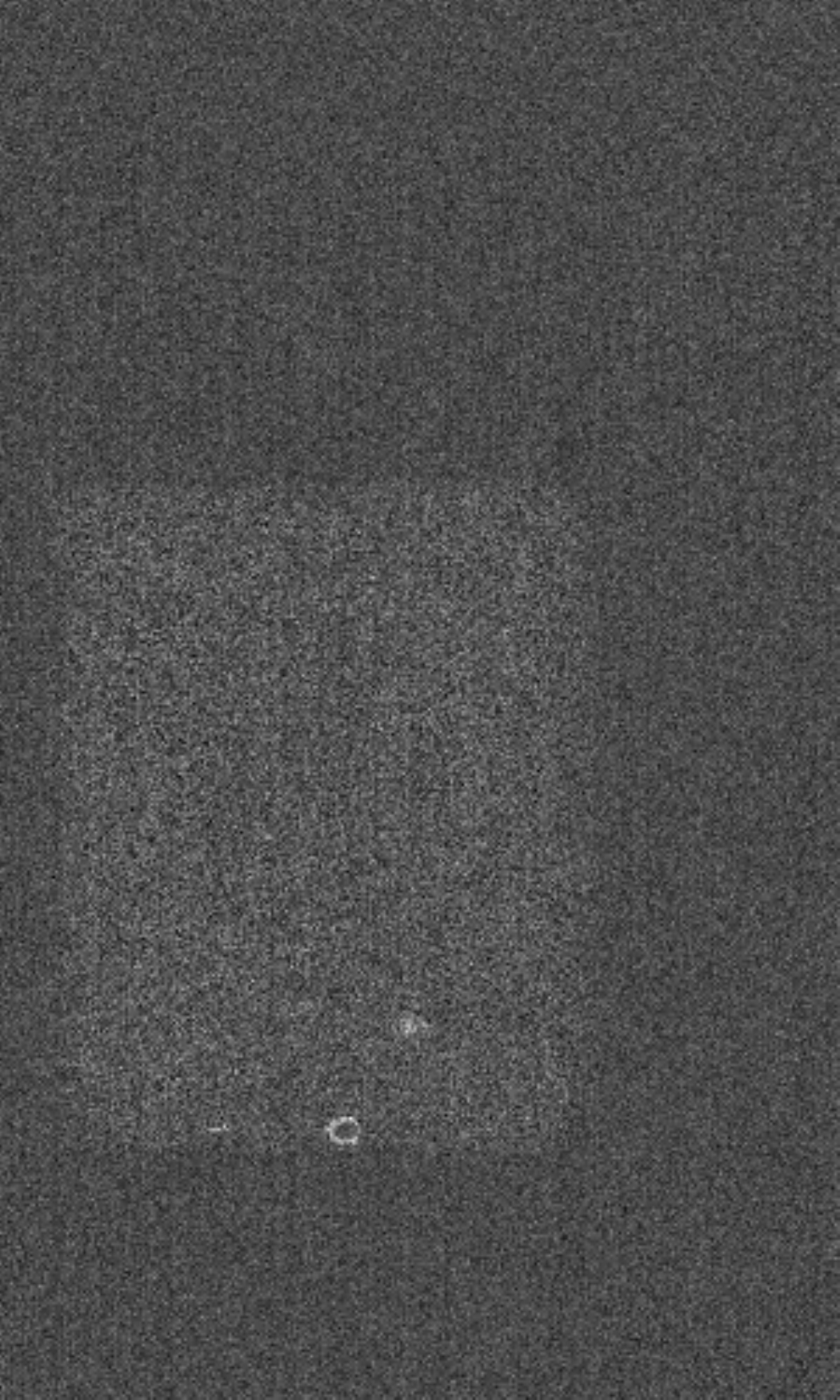}
  		\includegraphics[width=0.24\columnwidth]{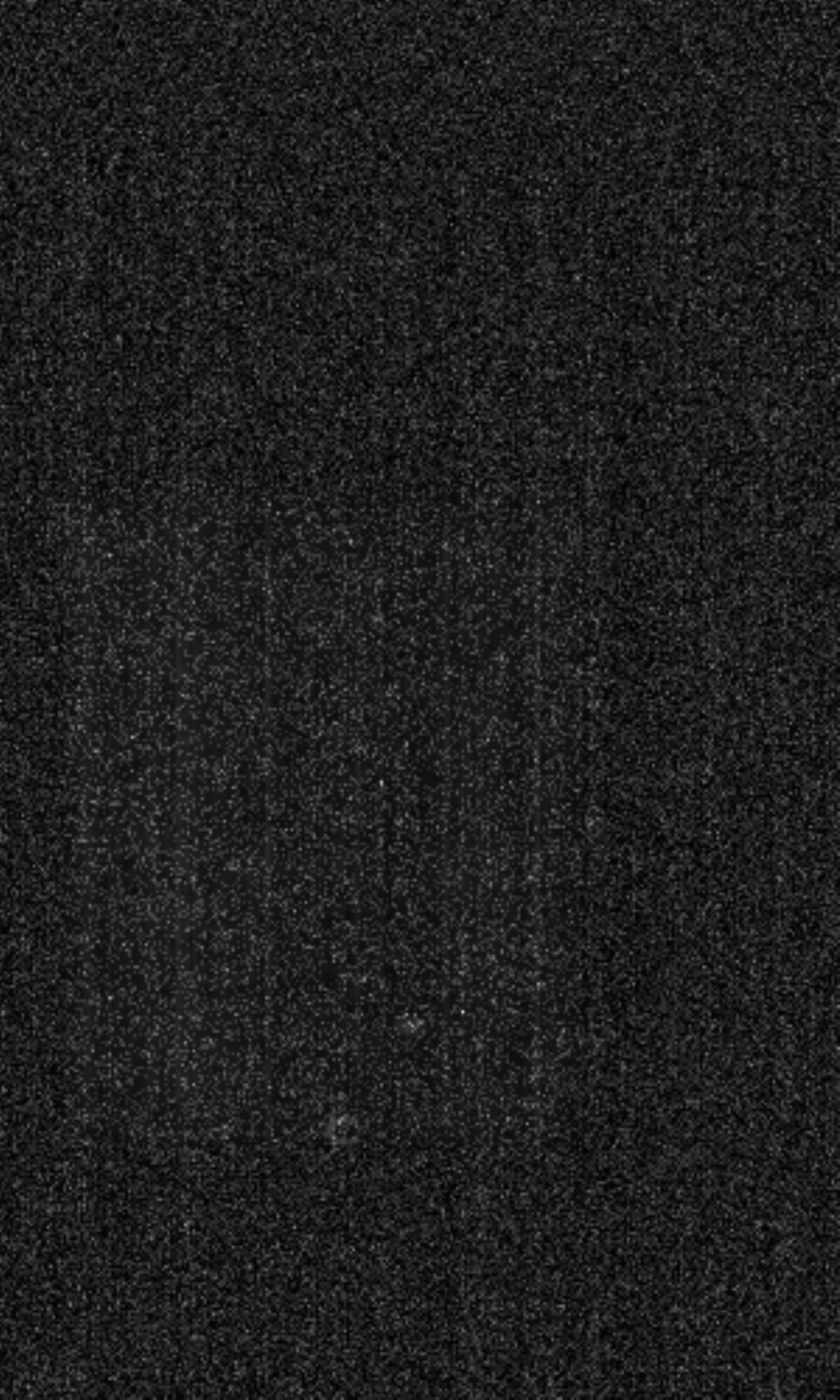}
  		\caption{Grayscale images of different channels from a minced meat sample generated with a multispectral X-ray scanner after all preprocessing. From left to right are channels 2, 20, 50 and 100. The contrast decreases the higher we go in the channels and the variation in the measurements increases. Some foreign objects can be seen as small black dots.}
      \label{fig:meats}
  	\end{figure}
  	\begin{figure}[!htbp]
  		\centering
  		\includegraphics[width=1\columnwidth]{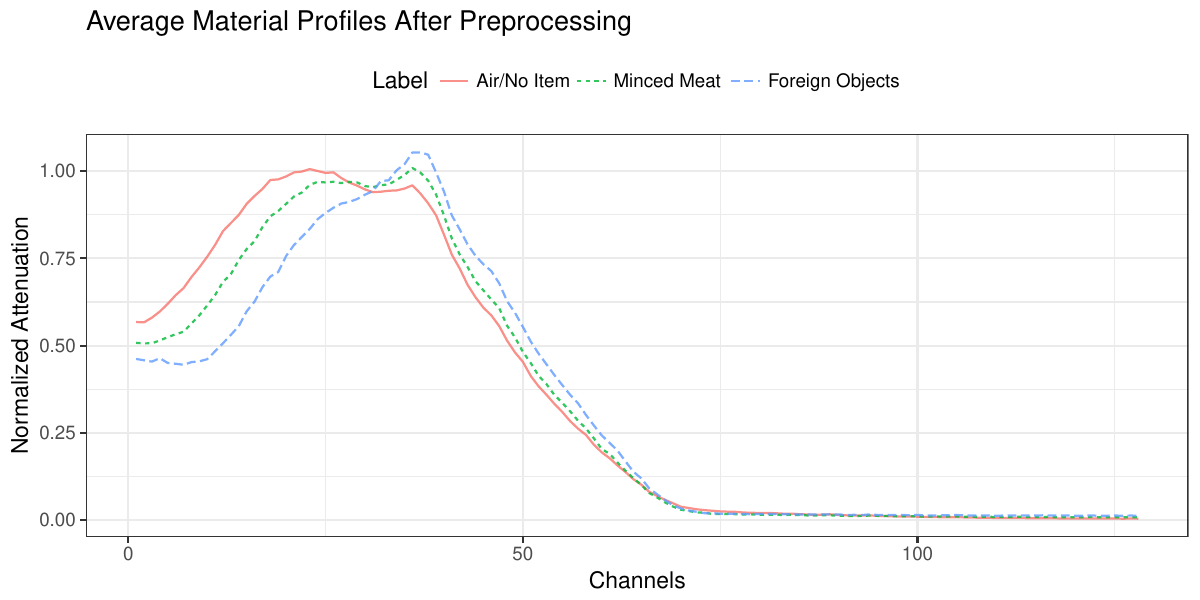}
  		\caption{Profiles of materials seen in Figure~\ref{fig:meats} over the 128 channels. The profile for each type of material, displayed here, is averaged over 500 pixels.}
      \label{fig:profile}
  	\end{figure}

	We start by preprocessing the scans as in \cite{einarsson2017foreign} in order to remove scanning artifacts and normalize the intensities between scans. We scale the measurements in each pixel by the 95\% quantile of the corresponding 128 measurements instead of the maximum. This scaling approach is more robust in the sense that it is less sensitive to outliers compared to using the maximum.
	We create our training data by manually selecting rectangular patches from six scans. We have three classes, namely \textit{background}, \textit{minced meat} and \textit{foreign objects}. We further subsample the observations to have balanced number of observations, where the class \textit{foreign objects} was under represented. In the end we have 521 observations per class, where each observation corresponds to a single pixel. This data was used to generate Figure~\ref{fig:profile}.  For training we use 100 samples per class, and the rest is allocated to a final test set.
	This process yields 128 variables per observation, but in order to get more spatially consistent classification, we also include data from the pixels located above, to the right, below and to the left of the observed pixel. Thus we have $p=5\cdot 128=640$ variables per observation. The measurements corresponding to our observation are thus indexed according to spatial and spectral position, i.e., observation $\mathbf{x}_i$ has measurements $x_{ijk}$, where $j\in\{ 0,1,2,3,4 \}$ indicates which pixel the measurement belongs to (\textit{center, above, right, bottom, left}), and $k\in\{1,2,...,128\}$ indicates which channel.

	We can impose priors according to these relationships of the measurements in the $\bs\Omega$ regularization matrix. We assume that the errors should vary smoothly in space and thus impose a Mat{\'e}rn covariance structure on $\bs\Omega^{-1}$ \cite{matern2013spatial}:
	\begin{equation} \label{eq: matern}
		C_{\nu}(d) = \sigma^2 \frac{2^{1-v}}{\Gamma(\nu)}\left( \sqrt{2\nu}\frac{d}{\rho} \right)^{\nu}K_\nu\left( \sqrt{2\nu}\frac{d}{\rho} \right).
	\end{equation}
	The Mat{\'e}rn covariance structure \eqref{eq: matern} is governed by the distance $d$ between measurements. In \eqref{eq: matern}, $\Gamma$ refers to the gamma function and $K_{\nu}$ is the modified Bessel function of the second kind. For this example we assume that all parameters are 1, except that $\nu$ is $0.5$. We further assume that the distance between measurements $x_{ijk}$ and $x_{ij^{\prime}k^{\prime}}$ from observation $i$ is the Euclidean distance between the points $(x_j,y_j,z_k)$ and $(x_{j^{\prime}},y_{j^{\prime}},z_{k^{\prime}})$, where $x_j,y_j,x_{j^{\prime}},y_{j^{\prime}}\in\{-1,0,1\}$ and $z_k,z_{k^{\prime}}\in\{1,2,...,128\}$. The distance is thus the same as in the image grid (center, top, bottom, left, right pixel location), and $z$-dimension corresponds to the channel.

	We use a stopping tolerance of $10^{-5}$ and a maximum of 1000 iterations for the inner loop using the accelerated proximal algorithm, and a stopping tolerance of $10^{-4}$ and maximum 1000 iterations for the outer block-coordinate loop. The regularization parameter for the $l_1$-norm is selected as $\lambda = 10^{-3}$ and $\gamma = 10^{-1}$ for the Tikhonov regularizer. We present the run-time for varying $r$ in Figure~\ref{fig:rankTime} and the accuracy with respect to varying $r$ in Figure~\ref{fig:accTime}. There is a clear linear trend in rank $r$ for the increase in run-time; this agrees with the analysis of Section~\ref{sec: comp}.
	We also estimate the accuracy for a  identity regularization matrix, i.e., $\bs{\Omega} = \I$, with the same regularization parameters $\gamma$ and $\lambda$ and achieve accuracy of 0.948, which is approximately the same accuracy as when using $\boldsymbol{\Omega}^{400}$.
	To demonstrate the effect that the rank of $\bs\Omega$ has on computational complexity, we obtain the singular value decomposition of $\bs\Omega = \sum_{i=1}^p \sigma_i \u_i \v_i^T$, and construct a low-rank approximation to $\bs\Omega$ using the first $r$ singular vectors and singular values: $\bs{\Omega^{r}} = \sum_{i=1}^r \sigma_i \u_i \v_i^T$.
	We supplied the same parameters to the function \texttt{sda} from the library \texttt{sparseLDA};  \texttt{sda} required 267 seconds to run and achieved an accuracy of $0.949$. The maximum accuracy is achieved with the full regularization matrix, which is $0.957$.

	\begin{figure}
		\centering
		\includegraphics[width=1\columnwidth]{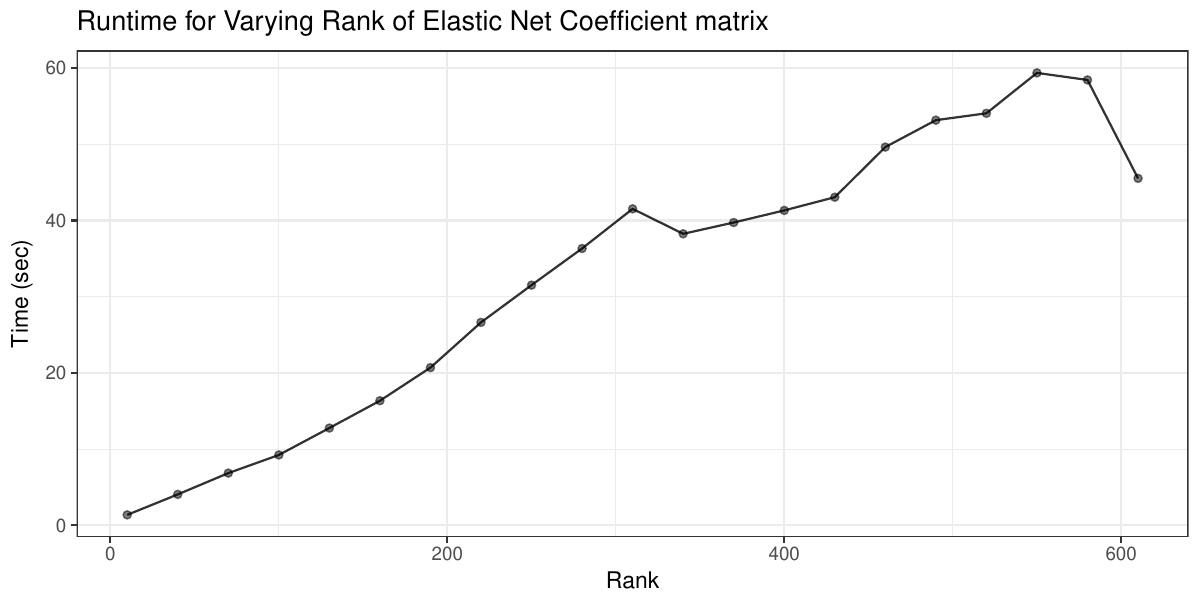}
		\caption{run-time as function of $\rank(\bs\Omega)$. The run-time also includes the creation of the low-rank approximated $\boldsymbol{\Omega}$ matrix.}
		\label{fig:rankTime}
	\end{figure}

	\begin{figure}
		\centering
		\includegraphics[width=1\columnwidth]{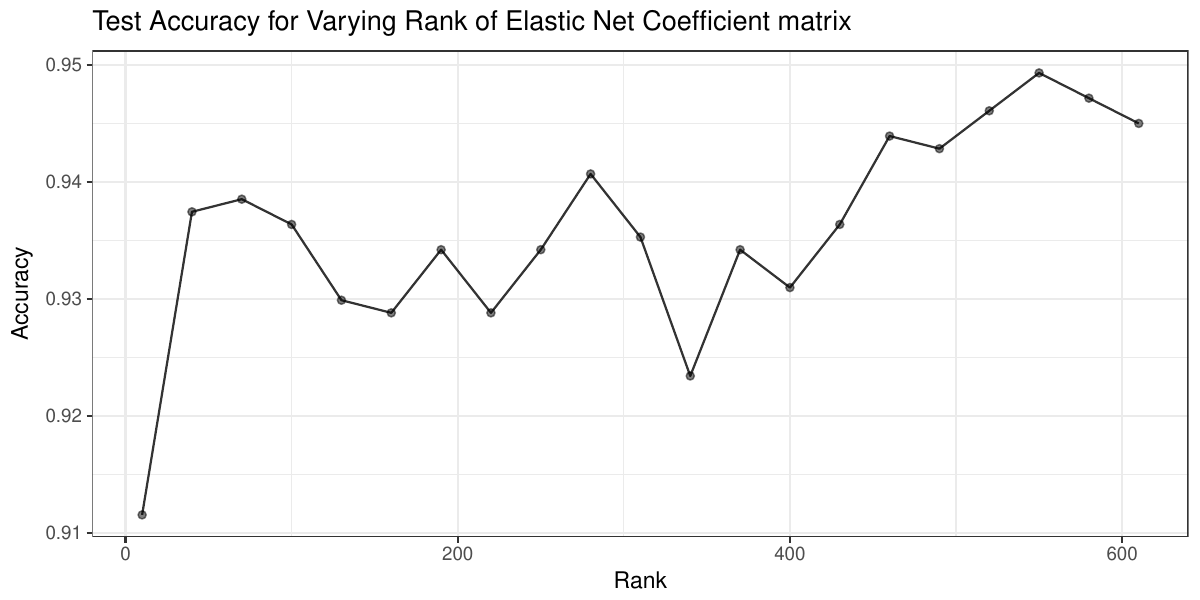}
		\caption{Test accuracy as function of $\rank(\bs\Omega)$}.
		\label{fig:accTime}
	\end{figure}

  \subsection{Summary}

	Our proposed proximal methods for sparse discriminant analysis provide	a decrease in classification error over the existing LARS approach	in almost all experiments.
  Moreover, we see a significant improvement in terms of computational resources used by the	accelerated proximal gradient method (APG and APGB) and alternating direction method of multipliers (ADMM) over LARS in medium to large-scale problem instances, i.e., when the number of predictor variables $p$ exceeds $1000$, without significant loss of classification accuracy when compared to standard nearest neighbors classifiers (ED, DTW); we remind the reader that SOS with APG and ADMM does not match the classification accuracy of the nearest neighbor classifier using dynamic time wapring distance with learned window width (DTWL) due to poor classification performance of SDA for a small number of data sets for which SDA does not seem applicable, and that our methods are more efficient than DTWL, offer an element of feature selection via sparsity, and are amenable to learning tasks other than classification as a general dimension reduction tool.
	We should note that this agrees with the theoretical estimates of computational cost of these methods given in Section~\ref{sec: comp} and Appendix~\ref{sec:it-comp}.
  Specifically, both APG and ADMM converge linearly with per-iteration complexity on the order of $\O(p)$ floating point operations per-iteration, which leads to overall computation time, as measured in floating point operations, to be far less than the
	$\O(p^3)$ flops of the classical LARS method.
  In our experiments, the decrease in run-time is most
	significant when $p$ is large, where
	the cost of $\O(p^3)$ flops for LARS becomes prohibitive.
  It is important to note that the slow convergence of the proximal gradient method (PG/PGB)
	without acceleration yields
	significantly longer run-times despite the decreased per-iteration cost.
  Finally, we note that there appears to be limited benefit from the use of backtracking line search, when compared to a constant step size given by the Frobenius norm estimate $\|\A\|_F$ of the Lipschitz constant.
  Specifically, the results of these experiments indicate that using a constant step length yields similar classification performance to the backtracking approach, but without a significant increase in run-time due to repeated calculation of $\nabla f$.

  \subsubsection{Comparison of ADMM and APG}
  \label{sec:ADMMvsAPG}

  The results of our empirical analysis suggest that the use of either APG and ADMM to solve~\eqref{eq: b prob} may yield a significant improvement over the classical LARS-EN heuristic, particularly when $p$ is large and we seek dense discriminant vectors.
  However, which of these two methods is most efficient seems to vary under different experimental conditions.
  Specifically, APG generally requires less overall run-time than ADMM when analyzing data sets from the UCR benchmarking repository considered in Section~\ref{sec:real-data}.
  On the other hand, ADMM tends to converge more quickly and require less computation than APG  when analyzing synthetic data, particularly in our convergence tests (Sect.~\ref{sec:convtrials}) and scaling tests (Sect.~\ref{sec:scaling}).
  This suggests that performance of the ADMM heuristic is sensitive to the choice of the augmented Lagrangian parameter $\mu$, as this is the only parameter that varies between these different analyses.

  We performed the following analysis to further illustrate this sensitivity to the choice of $\mu$. We generated $100$ different data sets containing $n=200$ training observations sampled from each of the $p=2000$ dimensional Gaussian distributions $N(\bs{\mu}_1, \bs{\Sigma})$ and $N(\bs{\mu}_2, \bs{\Sigma})$ as in Section~\ref{sec:convtrials}.
  For each problem instance, we (approximately) solved~\eqref{eq: prob} using APG and ADMM with $\mu \in \{1/25, 1/5, 1, 5, 25, 125\}$ to solve~\eqref{eq: b prob}. We set $\gamma = 10^{-3}$, $\bs{\Omega} = \bs{I}$, $\lambda = 0.05\bar{\lambda}$ and terminate APG or ADMM if their respective stopping criteria are met with tolerance $10^{-4}/\sqrt{p}$.
  For each data set, we record the objective function value of~\eqref{eq: b prob} at each iteration, as well as cardinality of the obtained discriminant vector, number of iterations performed before termination, total run-time, and out-of-sample classification accuracy for a testing set of $200$ observations drawn from each of $N(\bs{\mu}_1, \bs{\Sigma})$ and $N(\bs{\mu}_2, \bs{\Sigma})$.
  Note that we are essentially repeating our analysis from Section~\ref{sec:convtrials}, except this time we focus only APG and ADMM under varying choices of $\mu$.

  \begin{figure}
    \centering
    \includegraphics[width=\textwidth]{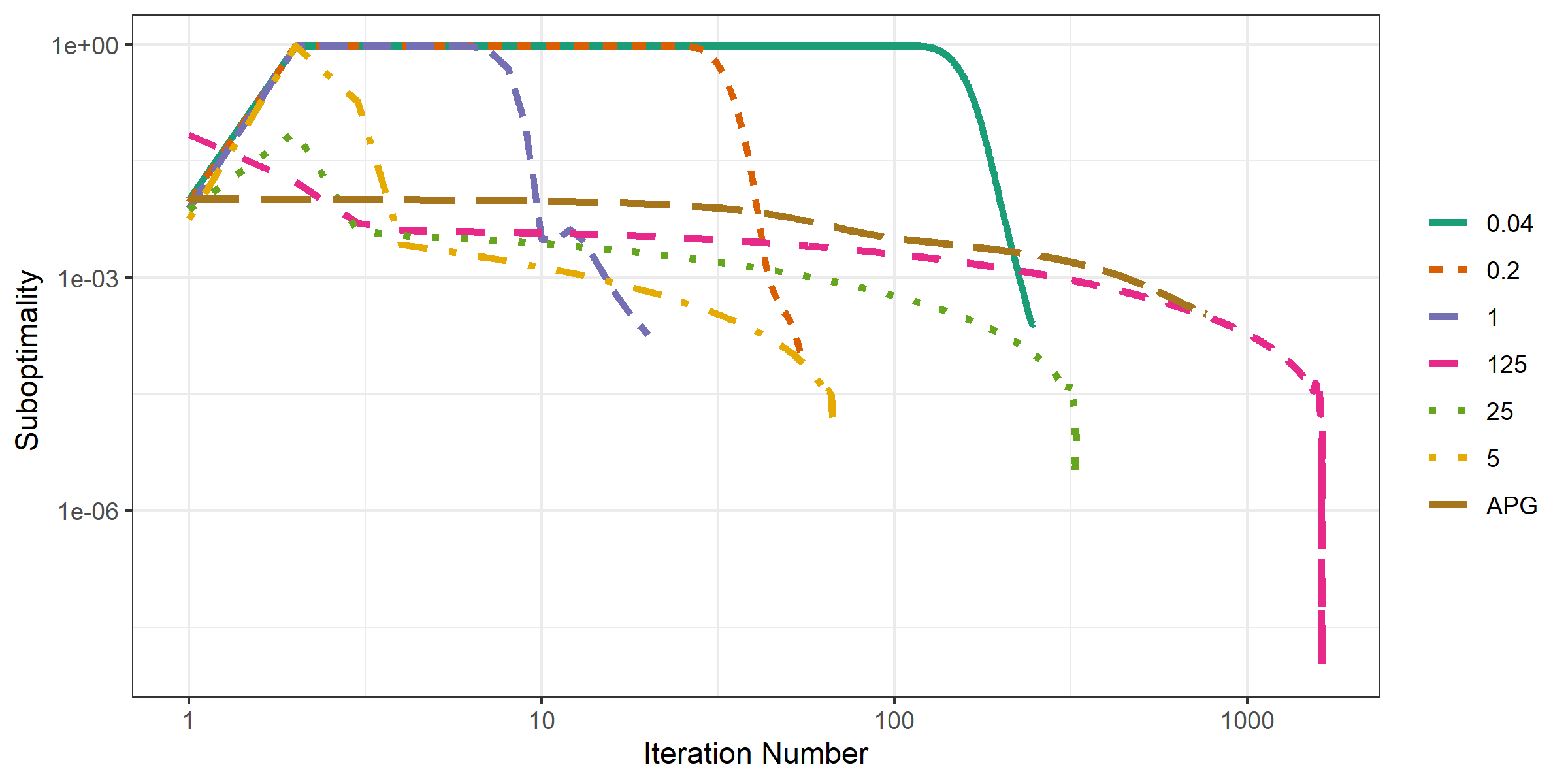}
    \caption{Difference of objective value and optimal value of each iterate averaged across $100$ Gaussian data sets.
    We note that ADMM converges in fewer iterations than APG for all choices of $\mu$ except $\mu = 125$.}
    \label{fig:admm-conv}
  \end{figure}

  \begin{table}[t]
    \centering
    \begin{tabular}{rlll}
      \toprule
     Method & Cardinality & Run-Time & Number of Iterations \\
      \midrule
      APG & 536.1 (92) & 2.902 (0.531) & 766 (136.9) \\
    $\mu = 1/25$ & 280.4 (7.2) & 0.941 (0.066) & 248.7 (3.4) \\
    $\mu = 1/5$ & 291.3 (8.7) & 0.29 (0.022) & 55.9 (0.7) \\
    $\mu = 1$ & 410.1 (12.7) & 0.171 (0.013) & 20.7 (0.6) \\
    $\mu = 5$ & 426.6 (12.8) & 0.33 (0.026) & 67.5 (2.9) \\
    $\mu = 25$ & 428.4 (12.3) & 1.212 (0.1) & 328.6 (14.6) \\
    $\mu = 125$ & 428.7 (12.3) & 5.63 (0.475) & 1639.7 (73.4) \\
       \bottomrule
    \end{tabular}
    \caption{Average number of nonzero entries (of $p=2000$), average run-time in seconds, and average number of iterations performed before termination, with standard deviation in parentheses. We see that ADMM yields sparser solutions than APG for all choices of $\mu$ and is more efficient than APG for all $\mu$ except $\mu = 125$, in terms of both total run time and number of iterations performed. All methods achieved $100\%$ out-of-sample classification rate for all $100$ training/testing data sets.}
    \label{tab:admm-conv}
  \end{table}

  Figure~\ref{fig:admm-conv} and Table~\ref{tab:admm-conv} summarize the results of this analysis.
  Recall that ADMM follows an (approximate) dual ascent applied to the dual functional of~\eqref{eq: ADMM prob}. The augmented Lagrangian parameter $\mu$ controls emphasis between the objective function of~\eqref{eq: ADMM prob} and the quadratic penalty function $\|\x - \y\|^2_2$.
  When $\mu$ is large, iterations of ADMM generally decrease disagreement between $\x$ and $\y$ while making only modest decreases, or even increases, in the objective $\frac{1}{2} \x^T \A \x + \d^T \x + \lambda \|\y\|_1$; this corresponds to the slow convergence observed when $\mu =125$.
  On the other hand, when $\mu$ is very small, we have significant decrease in the objective function each iteration, but many iterations are required before disagreement between $\x$ and $\y$ is small. We observe a two-stage phenomena in our experiments, where early iterations of ADMM feature slow decrease or increase in objective value while the gap between $\x$ and $\y$ decreases, followed by sharp descent in objective value; this initial period is longest when $\mu$ is small.
  This provides empirical evidence that we need to carefully choose the penalty parameter $\mu$ in order to optimize convergence of our ADMM heuristic, and partially explains the gap in efficiency between APG and ADMM observed in our experiments. Specifically, ADMM is more efficient than APG only if a suitable choice of $\mu$ is used; if we do not carefully tune this penalty parameter, APG can be significantly more efficient than ADMM.

	\section{Conclusion}
	We have proposed new algorithms for solving the sparse optimal scoring problem for
	high-dimensional linear discriminant analysis based
	on block coordinate descent and proximal operator evaluations. We observe that these algorithms provide significant
	improvement over existing approaches for solving the SOS problem
	in terms of efficiency and scalability. These improvements are most acute in the case that specially structured Tikhonov
	regularization is employed in the SOS formulation;
	for example, the computational resources required for each iteration scale linearly with
	the dimension of the data if either a diagonal or low-rank matrix is used.
	Moreover, we establish that any convergent subsequence of iterates generated by one of our
	algorithms converges to a stationary point.
	Finally, numerical simulation establishes that our approach provides an improvement
	over existing methods for sparse discriminant analysis, particularly when the number of nonzero predictor variables in the discriminant vectors is relatively large.

	These results present several exciting avenues for future research.
	Although we focus primarily on the solution of the optimal scoring problem
	under regularization in the form of a generalized elastic net penalty,
	our approach should translate immediately to
	formulations with any nonsmooth convex penalty function.
	That is, the framework provided by Algorithm~\ref{alg: BCD} can be applied
	to solve the SOS problem \eqref{eq: prob} obtained by applying
	an arbitrary convex penalty to the objective of the optimal scoring
	problem \eqref{eq: os prob}. The resulting optimization problem
	can be approximately solved by alternately minimizing
	with respect to the score vector $\bt$
	using the formula \eqref{eq: theta formula} and with respect to the
	discriminant vector $\bb$ by solving a modified version of
	\eqref{eq: b prob}.
	The proximal methods outlined in this paper can be applied to
	minimize with respect to $\bb$ if the regularization function is convex;
	however it is unlikely that the computational resources necessary for
	this minimization will scale as favorably as with the generalized
	elastic net penalty.
	On the other hand, the convergence analysis presented in
  Section~\ref{convergence}
  extends immediately to this more general framework.
	Of particular interest is the modification of this approach to provide
	means of learning discriminant vectors for data containing ordinal
	labels, data containing corrupted or missing observations,
	and semi-supervised settings.

	Finally, the results found in Section~\ref{convergence}, as well as Appendices~\ref{sec: convergence1} and~\ref{sec: convergence2} 
  establish that
	any convergent subsequence of iterates generated by our block coordinate descent approach must converge to a stationary point. However, it is still unclear when this sequence of iterates is convergent, or
	at what rate these subsequences converge;
	further study is required to better understand the convergence properties
	of these
	algorithms.
	Similarly, despite the empirical evidence provided in Section~\ref{sims},
	it is unknown what conditions ensure that data is classifiable 
	using sparse optimal scoring and, more generally, linear discriminant
	analysis.
	Extensive consistency analysis is needed to determine theoretical
	error rates for distinguishing random variables
	drawn from distinct distributions.


	\subsection{Acknowledgements}{
	We are grateful to Mingyi Hong for his helpful comments and suggestions, and to the anonymous reviewers, whose suggestions significantly improved this manuscript.
  B.~Ames was supported in part by National Science Foundation Grants \#20212554 and \#2108645, University of Alabama Research Grants RG14678 and RG14838, and a UA Cyberseed Grant. G.~Einarsson's PhD scholarship was funded by the Lundbeck foundation and the Technical University of Denmark.
	S.~Atkins was part of the University Scholars Program at the University of Alabama and received a graduate student fellowship funded by the University of Florida while this research was conducted.
	This work was made possible in part by a grant of high performance computing resources and technical support from the Alabama Supercomputer Authority.}
\vskip 0.2in

\bibliographystyle{spmpsci}      
\bibliography{SDAref}

\appendix

\section{Detailed Calculation of Per-Iteration Complexity}
\label{sec:it-comp}
The most expensive step of both the proximal gradient method (Algorithm~\ref{alg: pg}) and the accelerated proximal gradient method (Algorithm~\ref{alg: apg}) is the evaluation of the gradient $\nabla f$.
Given a vector $\bb \in \R^p$, the gradient at $\bb$ is given by
\begin{align*}
\nabla f(\bb) &= \A \bb = 2\rbra{ \X^T \X + \gamma \bO } \bb= 2 \X^T \X \bb + 2 \gamma \bO \bb.
\end{align*}
The product $ \X^T \X \bb $ can be computed using $\O(np)$ floating point operations (flops)
by computing $\y = \X \bb$  and then $\X^T \y$. 
On the other hand, the product $\bO\bb$ requires $\O(p^2)$  flops for unstructured $\bO$.
However, if we use a \emph{structured} regularization parameter $\bO$ we can
significantly decrease this computational cost. Consider the following examples:

\begin{itemize}
	\item
	Suppose that $ \bO $ is a diagonal matrix: $ \bO = \Diag(\u) $ for some vector
	$\u \in \R^p_+$. Then the product $\bO \bb $ can be computed using
	$\O(p)$ flops:
	$(\bO\bb)_i = u_i \beta_i$.
	Moreover, we can estimate the Lipschitz constant $\|\A\|$
	for use in choosing the step size $\alpha$ by
	$
	\|\A\| \le 2 \gamma \| \bO \| + 2 \|\X\|^2_F
	= 2 \gamma \|\u \|_\infty + 2 \|\X\|^2_F,
	$
	which requires $\O(np)$ flops, primarily to compute the norm
	$\|\X\|^2_F$.
	\item
	If the use of diagonal $\bO$ is inappropriate, we could store $\bO$
	in factored form $\bO = \bs R \bs R^T$ where $\bs R \in \R^{p\times r}$,
	and $r$ is the rank of $\bO$.
	In this case, we have
	$
	\bO \bb = \bs R (\bs R^T \bb),
	$
	which can be computed at a cost of $\O(rp)$ flops.
	Thus, if we use a low-rank parameter $\bO$,
	say $r \le \O(n)$,
	we can compute the gradient using $\O(np)$ flops.
	Similarly, we can estimate the step size $\alpha$ using
	$
	\|\A\| \le 2 \|\bs R \|_F^2 + 2\|\X\|^2_F
	$
	(computed at a cost of $\O(rp + np)$ flops).
\end{itemize}
In either case, using a diagonal $\bO$ or low-rank factored $\bO$, each iteration of
the proximal gradient method or the accelerated proximal gradient method requires
$\O(np)$ flops. Similar improvements can be made if $\bO$ is tridiagonal, banded, sparse, or otherwise nicely structured.

Similarly, the use of structured $\bO$ can lead to significant improvements in computational efficiency
in our ADMM algorithm.
The main computational bottleneck of this method is the solution of the linear system
in the update of $\x$:
\[
(\mu \I + \A) \x^{i+1} = \d + \mu \y^i - \z^i.
\]
Without taking advantage of the structure of $\A$, we can solve this system using a Cholesky
factorization preprocessing step (at a cost of $\O(p^3)$ flops) and substitution to solve
the resulting triangular systems (at a cost of $\O(p^2)$ flops per-iteration).
However, we can often use the Sherman-Morrison-Woodbury Lemma to solve this system more
efficiently using the structure of $\A$.
Indeed, fix $t$ and let $\b = \d + \mu \y^i - \z^i$.
Then we update $\x$ by
$\x = (\mu \I + \A)^{-1} \b.$
If $\M = \mu \I + 2 \gamma \bO$ then we have
\begin{align*} \notag
(\mu \I + \A)^{-1} &= \rbra{ \mu \I + 2 \gamma \bO + 2\X^T \X }^{-1} 
= \rbra{ \M + 2 \X^T \X}^{-1} \nonumber\\
&=\M^{-1} - 2 \M^{-1} \X^T \left({ \I +2\X \M^{-1} \X^T }\right)^{-1} \X\M^{-1}.\nonumber
\label{eq: SMW x}
\end{align*}
The matrix $ \I + 2 \X \M^{-1} \X^T$ is $n\times n$, so we may solve any linear system with this coefficient matrix using $\O(n^3)$ flops; a further $\O(n^2p)$ flops are needed to
compute the coefficient matrix if given $\M^{-1}$.
Thus, the main computational burden of this update step
is the inversion of the matrix $\M$. As before, we want to choose $\bO$ so that we can exploit its structure.
Consider the following cases.
\begin{itemize}
	\item
	If $\bO = \Diag( \u) $ is diagonal, then  $\M$ is also diagonal with
	\[
	[\M^{-1}]_{ii} = \frac{1}{\mu + 2 \gamma u_i}.
	\]
	Thus, we require $\O(p)$ flops to compute $\M^{-1} \bs v$ for any vector $\bs v \in \R^p$.
	\item
	On the other hand, if $\bO = \bs R \bs R^T$, where $\bs R \in \R^{p\times r}$,
	then we may use the
	Sherman-Morrison-Woodbury identity to compute $\M^{-1}$:
	\[
	\M^{-1}
	= \frac{1}{\mu}{\I} - \frac{2 \gamma}{\mu^2} \bs R \rbra{\I + \frac{2 \gamma}{\mu} \bs R^T \bs R }^{-1}
	\bs R^T.
	\]
	Therefore, we can solve any linear system with coefficient matrix $\M$
	at a cost of $\O(r^2 p)$ flops (for the formation and solution of the
	system with coefficient matrix $\I + \frac{2\gamma}{\mu} \bs R^T \bs R$).
\end{itemize}
In either case, we never actually compute the matrices $\M^{-1}$ and $(\mu \I + \A )^{-1} $
explicitly. Instead,  we update $\x$ as the solution of a sequence of linear systems and
matrix-vector multiplications, at a total cost of $\O(n^2 p)$ flops (in the diagonal case)
or $\O((r^2 + n^2) p)$ flops (in the factored case). Thus, if the number of observations $n$
is much smaller than the number of features $p$, then the per-iteration computation scales
roughly linearly with $p$. Table~\ref{flops} summarizes these estimates of per-iteration computational costs for each proposed algorithm.
Further, we should note that these bounds on per-iteration cost
assume that the iterates $\bb$ and $\x$ are dense; the soft-thresholding
step of the proximal gradient algorithm typically induces $\bb$
containing many zeros, suggesting that further improvements can be
made by using sparse arithmetic.
\begin{table}[t]
\begin{center}
	\resizebox{0.95\textwidth}{!}{

		\begin{tabular}{l l l l l l} \toprule
			& & Diagonal $\bs{\Omega}$ & Rank $r$ $\bs{\Omega}$  & Full rank $\bs{\Omega}$ \\ \midrule
			Proximal Gradient &	$\nabla f$ & $\O(np)$ & $\O(rp + np)$ & $\O(p^2)$ \\
			& Bound on $\|\A \|$ &
			$\O(np) $& $\O(rp + np) $& $\O(p^2 \log p)$ \\ 
			ADMM &$ (\mu \I + \A)\x = \b $& $\O(n^3 + n^2p)$ & $\O(n^3 + n^2 p + r^2p)$ & $\O(p^3)$
			\\ \bottomrule
		\end{tabular}
    }

	\end{center}
  \caption{Upper bounds on floating point operation counts for most time consuming steps of each algorithm.
  \label{flops}
}
\end{table}

\section{Proof of Lemma~\ref{lem: theta update}}
\label{sec: theta}

\newcommand{\w}{{\bs{w}}}
\renewcommand{\v}{{\bs{v}}}
We note that \eqref{eq: t prob} has trivial solution $\bt = \e_k$ for every $\bb=\bb^i\in \R^p$ and $j=1$.
Indeed, $\Y\e_k = \e_n$ by the structure of the indicator matrix $\Y$ and $ \sum_{i=1}^n x_{ij}=  0$ for all $j=1, 2,\dots, p$ because our data has
been centered to have sample mean equal to $\bs 0$.
Therefore, we may reformulate \eqref{eq: t prob} as
\begin{equation} \label{eq: tp}
 \begin{array}{rl}
  \ds{\min_{\bt\in\R^K} }& \|\Y\bt - \X\bb\|^2 \\
  \st &\bt^T \Y^T\Y \bt = n,\\
	& \; \bt^T \Y^T\Y \e = 0,\;\\
	& \bt^T \Y^T \Y \bt_\ell = 0\;\; \ell < j,
  \end{array}
\end{equation}
to avoid this trivial solution.
We wish to show that \eqref{eq: tp}
has optimal solution $\hat\bt$ given by
\begin{equation} \label{eq: tsol}
\hat \bt = \frac{\sqrt{n}\bs{w}}{\sqrt{\bs{w}^T \Y^T\Y \bs{w}}},
\end{equation} where
$ \bs{w} = (\I - \frac{1}{n} \bs Q_j \bs Q_j^T \Y^T \Y)(\Y^T \Y)^{-1} \Y^T \X {\bb}$.

To do so, note that \eqref{eq: tp} satisfies the linear independence constraint qualification
because the set of constraint function gradients \[ \{ 2 \Y^T\Y \bt,  \Y^T \Y \e,  \Y^T \Y \bt_1, \dots,  \Y^T\Y \bt_{j-1} \} \]
is linearly independent.
Moreover, the optimal value of \eqref{eq: tp} is bounded below by $0$. Therefore, \eqref{eq: tp} has global minimizer,
$\hat \bt$, which must satisfy the Karush-Kuhn-Tucker conditions, i.e., there exists $\v \in \R^j$, $\psi \in \R$
such that
\newcommand{\h}{\hat\bt}
\renewcommand{\Q}{\bs{Q}}
\begin{equation} \label{eq: stat}
\Y^T\Y\hat\bt -  \Y^T \X \bb +  \psi \Y^T \Y \hat \bt +  \Y^T\Y \Q_j \v = \bs 0,
\end{equation}
where $\Q_j = [\e, \bt_1, \bt_2, \dots, \bt_{j-1}]$.
We consider the following two cases.

First, suppose that $\Y^T \X \bb \notin \range \rbra{\Y^T\Y \Q_j}$.
Rearranging \eqref{eq: stat} yields
\begin{equation} \label{eq: ht}
\h = \frac{ 1}{1 + \psi} \rbra{ \Y^T\Y}^{-1} \rbra{ \Y^T \X\bb - \Y^T \Y \Q_j \v}.
\end{equation}
We choose the dual variables $\psi$ and $\v$ so that $\h$ is feasible for \eqref{eq: tp}.
It is easy to see that the conjugacy constraints are equivalent to
$
\Q_j^T \Y^T\Y \h = \bs 0,
$
which holds if and only if
\begin{align*}
\bs 0 &= \Q_j^T (\Y^T \X\bb - \Y^T \Y \Q_j \v) = \Q_j^T \Y^T \X\bb - \Q_j^T \Y^T \Y \Q_j \v \\
&= \Q_j^T \Y^T \X\bb - n \v,
\end{align*}
where the last equality follows from the fact that $\e^T \Y^T\Y \e = \bt_i^T \Y^T\Y\bt_i = n$
for all $i = 1,2,\dots, j-1$.
It follows immediately that
\begin{equation} \label{eq: v}
\v = \frac{1}{n} \Q_j^T \Y^T \X \bb.
\end{equation}
Substituting \eqref{eq: v} into \eqref{eq: ht} yields
\begin{align}
\h &= \frac{1}{1 + \psi} \rbra{ (\Y^T\Y)^{-1} \Y^T \X\bb - \frac{1}{n} \Q_j\Q_j^T \Y^T\X \bb } \notag \\
&= \frac{1}{1+ \psi} \rbra{ \I - \frac{1}{n}\Q_j\Q_j^T \Y^T \Y } (\Y^T \Y)^{-1} \Y^T\X \bb \notag\\ \label{eq: ht 2}
&= \frac{1}{1+ \psi} \w,
\end{align}
where we  choose $\psi \in \R$ so that $\h^T \Y^T \Y \h = n$:
\begin{equation} \label{eq: psi}
\sqrt{n} (1 + \psi) = \pm \sqrt{\w^T \Y^T \Y \w} = \pm \|\Y\w\|.
\end{equation}
To complete the argument, note that
\begin{align*}
\|\Y\h - \X\bb\|^2
    =  n &\mp \frac{2 \sqrt{n}}{\|\Y\w\|} \bb^T \X^T \rbra{\I - \frac{1}{n} \Q_j \Q_j^T \Y^T\Y}(\Y^T\Y)^{-1} \Y^T \X \bb \\
  & + \|\X\bb\|^2.
\end{align*}
Note further that the matrix $\Y \Q_j\Q_j^T \Y^T$ has decomposition
\begin{align*}
\Y \Q_j\Q_j^T \Y^T
= \Y\e\e^T \Y^T + & \Y\bt_1\bt_1^T \Y^T + \cdots  + \Y\bt_{j-1} \bt_{j-1}^T \Y^T.
\end{align*}
The conjugacy of the columns of $\Q_j$ implies that eigenvectors of $\Y\Q_j \Q_j^T \Y^T$
are $\Y \bt_1$, $\dots,$ $\Y \bt_{j-1}$,
and $\Y \e$, each with eigenvalue $n$; since $\Y \Q_j\Q_j^T \Y^T$ has rank equal to $k$, all remaining eigenvalues of $\Y \Q_j\Q_j^T \Y^T$ must be equal to $0$.
Moreover, for any $\z = \Y\bt_{i}$, $i=1,2,\dots,j-1$, or $\z = \Y\e$, we have
\[
\Y \rbra{(\Y^T \Y)^{-1} - \frac{1}{n} \Q_j\Q_j^T} \Y^T\z = \z - \z = 0.
\]
The matrix $(\Y^T\Y)^{-1}$ is a positive definite diagonal matrix, with $i$th diagonal entry $1/|C_i|$, where $|C_i|$ denotes the number of observations belonging to class $i$;
this implies that $$\Y (\Y^T \Y)^{-1} \Y^T$$ is positive semidefinite.
This establishes that the matrix $ \Y ((\Y^T \Y)^{-1} - \frac{1}{n} \Q_j\Q_j^T)\Y^T$ is positive semidefinite and,
thus,  $\|\Y\bt - \X\bb\|^2$ is minimized by $\h$ with $\psi = + \|\Y\w\|/\sqrt{n} - 1$.


Second, suppose that $\Y^T \X \bb \in \range \rbra{\Y^T\Y \Q_j}$.
This implies that there exists some $\v \in \R^K$ such that
\[
\Y^T \X \bb = \Y^T\Y \Q_j \v.
\]
Substituting into the objective of \eqref{eq: tp}, we see that
\begin{align*}
\|\Y \bt - \X\bb\|^2 &= \bt^T \Y^T \Y \bt - 2 \bt^T \Y^T \X \bb + \bb^T \X^T \X \bb \\
&= n - 2 \bt^T  \Y^T\Y \Q_j \v + \bb^T \X^T \X \bb \\
&= n + \bb^T \X^T \X \bb
\end{align*}
for every feasible solution $\bt$ of \eqref{eq: tp}.
This implies that every feasible solution of \eqref{eq: tp} is also optimal in this case.
In particular, $\h$ given by \eqref{eq: ht 2} is feasible for \eqref{eq: tp} and, therefore, optimal. \qed

\section{Proof of Theorem~\ref{thm: f convergence}}
\label{sec: convergence1}

We next prove Theorem~\ref{thm: f convergence}, which establishes that Algorithm~\ref{alg: BCD} converges in function value.

\begin{proof}
			Suppose that, after $t$ iterations, we have iterates $(\bt^i, \bb^i)$ with objective function
			value $F(\bt^i, \bb^i)$.
			Recall that we obtain $\bb^{i+1}$ as the solution of \eqref{eq: b prob}.
			Moreover, note that $\bb^i$ is also feasible for \eqref{eq: b prob}. This immediately implies that
			\[
			F(\bt^i, \bb^{i}) \ge F(\bt^i, \bb^{i+1}).
			\]
			On the other hand, $\bt^{i+1}$ is the solution of \eqref{eq: t prob} with $\bb = \bb^{i+1}$.
			Therefore, we have
			\[
			F(\bt^i, \bb^{i}) \ge F(\bt^i, \bb^{i+1}) \ge F(\bt^{i+1}, \bb^{i+1}).
			\]
			It follows that the sequence of function values $\{F(\bt^i, \bb^{i})\}_{i=1}^\infty$
			is nonincreasing. Moreover, the objective function $F(\bt, \bb)$
			is nonnegative for all $\bt$ and $\bb$.
			Therefore, $\{F(\bt^i, \bb^{i})\}_{i=1}^\infty$ is convergent as a monotonic bounded sequence.
      \qed
\end{proof}
\section{Proof of Theorem~\ref{thm: convergence}}
\label{sec: convergence2}
To prove Theorem~\ref{thm: convergence}, we first establish the following lemma, which establishes that the limit point
		$(\bt^*, \bb^*)$
		minimizes $F$ with respect to each primal variable with the other fixed;
		that is,
		$\bt^*$ minimizes $F(\cdot, \bb^*)$ and
		$\bb^*$ minimizes $F(\bt^*, \cdot) $.

		\begin{lemma} \label{lem: limit optimality}
			Let $\{(\bt^i, \bb^i)\}_{i=1}^\infty$ be the sequence of points generated by
			Algorithm~\ref{alg: BCD}.
			Suppose that $\{(\bt^{t_j}, \bb^{t_j})\}_{j=1}^\infty$ is a convergent subsequence of
			$\{(\bt^i, \bb^i)\}_{i=1}^\infty$ with limit $ ( \bt^*, \bb^*) $.
			Then
			\begin{align}
			F(\bt, \bb^*) & \ge F(\bt^*, \bb^* )  \label{eq: t opt}  \\
			F(\bt^*, \bb) & \ge F(\bt^*, \bb^*)   \label{eq: b opt}
			\end{align}
			for all feasible $\bt \in \R^k$ and $\bb \in \R^p$.
		\end{lemma}

			\begin{proof}
				We first establish \eqref{eq: b opt}. Consider $(\bt^{t_j}, \bb^{t_j})$. By our update step for $\bb$, we note that
				\[
				\bb^{t_j} = \argmin_{\bb\in \R^p} F(\bt^{t_j}, \bb).
				\]
				Thus, for all $j = 1,2,\dots,$ we have
				$ F(\bt^{t_j}, \bb) \ge F(\bt^{t_j}, \bb^{t_j})$
				for all $\bb\in \R^p$.	Taking the limit as $j \ra \infty$ and using the continuity of $F$ establishes
				\eqref{eq: b opt}.

				Next, note that, for every $j = 1,2,\dots $, we have
	                 		 \begin{align*}
				 \begin{array}{rl}
	   				\bs\theta^{t_j+1}=\ds{\argmin_{\bs \theta \in \R^K}} &F(\bs \theta,\bb^{t_j})\\
	   				\st & \bs \theta^T\Y^T\Y\bs \theta = n,\;\; \bs \theta^T\Y^T\Y\bs\theta_{\ell} = 0 \;\; \forall \ell <k.
	 			\end{array}
				\end{align*}
				This implies that
				\begin{align*}
					F(\bt, \bb^{t_j}) \ge F(\bt^{t_j + 1}, \bb^{t_j}) \ge F(\bt^{t_j+1}, \bb^{t_j+1})
					\ge F(\bt^{t_{j+1}}, \bb^{t_{j+1}})
				\end{align*}
				by the monotonicity of the sequence of function values and the fact that $t_j < t_j + 1 \le t_{j+1} $.
				Taking the limit as $j\to\infty$ and using the continuity of $F$ establishes \eqref{eq: t opt}.
			%
				This completes the proof of Lemma~\ref{lem: limit optimality}.
        \qed
			\end{proof}
		\bigskip

		We are now ready to prove Theorem~\ref{thm: convergence}.
		\begin{proof}\emph{of Theorem~\ref{thm: convergence}}
		\newcommand{\g}{\bs{g}}
		\renewcommand{\v}{{\bs{v}}}
		The form of the subdifferential of $\L$ implies that
		$(\g_\bt, \g_\bb)$ belongs to the subdifferential $\partial \L(\bt, \bb, \psi ,\v)$  if and only if
\begin{align}
			\g_\bt &= 2 (1 + \psi) \Y^T \Y \bt - 2 \Y^T \X \bb+  \bs{U}^T \v \label{eq: t diff} \\
			\g_\bb &\in 2 (\X^T \X + \gamma \bO) \bb - 2 \X^T \Y \bt + \lambda \partial \|\bb\|_1 \label{eq: b diff}
			\end{align}
			for all $\v \in \R^{j-1}$ and $\psi \in \R$.
			It is easy to see from \eqref{eq: b opt}  that
			$\bb^* = \argmin_{\bb \in \R^p} F(\bt^*, \bb)$.
			Thus, by the first order necessary conditions for
			unconstrained convex optimization, we must have
			\begin{align} \label{eq: b subdiff}
			\bs 0 &\in \partial \rbra{\frac{1}{2} (\bb^*)^T \A \bb^* + \d^T \bb^* + \lambda \|\bb^*\|_1} \\
&= 2 (\X^T \X + \gamma \bO) \bb^* - 2 \X^T \Y \bt^* + \lambda \partial \| \nonumber
			\bb^* \|_1;
			\end{align}
			here $\partial \|\bb\|_1$ denotes the subdifferential of the $\ell_1$-norm at the point $\bb$.

			On the other hand, \eqref{eq: t opt} implies
			\begin{equation}\label{eq: t* prob}
			 \begin{array}{rl}
			   \bt^* = \ds {\argmin_{\bt\in\R^K}} & \|\Y \bt - \X \bb^*\|^2 \\
			   \st & \bt^T \Y^T\Y \bt = n,\;\;2\bt^T \Y^T \Y \bt_\ell = 0, \; \forall \, \ell < j .
	                   \end{array}
	                  \end{equation}
			Moreover, the problem~\eqref{eq: t* prob} satisfies the linear independence constraint qualification.
			Indeed, the set of active constraint gradients\\
			$ \{ 2 \Y^T\Y \bt, 2 \Y^T \Y \bt_1, \dots, 2 \Y^T\Y \bt_{j-1} \} $
			is linearly independent for any feasible $\bt \in \R^K$  by the $\Y^T\Y$-conjugacy of $\{ \bt, \bt_1, \dots, \bt_{j-1}\}$.
			Therefore, there exist Lagrange multipliers $\psi^{*}$, $\v^{*}$ such that
			\begin{equation} \label{eq: t opt2}
			\bs 0 = 2 (1 + \psi^{*}) \Y^T \Y \bt^* - 2 \Y^T \X \bb^* +  \bs{U}^T \v^*
			\end{equation}
			by the first-order necessary conditions for optimality (see \cite[Theorem~12.1]{nocedal2006numerical}).
			%
			We see that $\bs0 \in \partial \L({\bt}^*, {\bb}^*, \psi^*, \v^*)$
			by combining \eqref{eq: b subdiff} and \eqref{eq: t opt2}.
			This completes the proof. \qed
		\end{proof}
%

\end{document}